\documentclass[a4paper, 11pt]{article}
\usepackage[left=2.5cm,right=2.5cm,top=2.5cm,bottom=2.5cm]{geometry}
\usepackage[english]{babel}
\usepackage[inline]{enumitem}

\usepackage[table,dvipsnames]{xcolor}
\usepackage{pdfpages} 
\usepackage{dsfont} % contains the command \mathds{1} for the indicator function
\usepackage[font={footnotesize}]{caption}
\usepackage[font={footnotesize}]{subcaption} % for subfigures
\usepackage{graphicx}
\usepackage{enumitem}
\usepackage{amsmath, bbm, amsthm}
\usepackage{bm}
\usepackage{amssymb}
\usepackage{multicol}
\usepackage[section]{placeins}
\usepackage{lipsum}
\usepackage[colorinlistoftodos,prependcaption,textsize=scriptsize]{todonotes}
\usepackage{caption}
\usepackage{subcaption}
\usepackage{xr-hyper} 
\usepackage{pdflscape}%sideway
\usepackage{algorithm}
\usepackage{algpseudocode}
\usepackage{blkarray}% for matrix with labels

\usepackage{booktabs}
\usepackage{longtable}
\usepackage{array}
\usepackage{multirow}
\usepackage{wrapfig}
\usepackage{float}
\usepackage{pdflscape}
\usepackage{tabu}
\usepackage{threeparttable}
\usepackage{threeparttablex}
\usepackage[normalem]{ulem}
\usepackage{makecell}
\usepackage{adjustbox}
\usepackage{lscape}
\usepackage{rotating}
\usepackage{footmisc}

\usepackage{setspace}
\usepackage[longnamesfirst]{natbib}

\usepackage[hypertexnames=false]{hyperref} % load last
\let\OldHref\href
\renewcommand{\href}[2]{\OldHref{#1}{\texttt{#2}}}
\hypersetup{
	colorlinks   = true, %Colours links instead of boxes
	urlcolor     = black, %Colour for external hyperlinks
	linkcolor    = blue, %Colour of internal links
	citecolor   = blue %Colour of citations
}

\newcommand{\size}{s}
\newcommand{\mW}{\bm W}

\newcommand{\mtheta}{\boldsymbol{\theta}}
\newcommand{\mTheta}{\boldsymbol{\Theta}}
\newcommand{\mlambda}{\boldsymbol{\lambda}}
\newcommand{\momega}{\boldsymbol{\omega}}
\newcommand{\mOmega}{\boldsymbol{\Omega}}

\newcommand{\myV}{\mathsf{V}}
\newcommand{\myS}{\mathsf{S}}

\newcommand{\cov}{\operatorname{Cov}}
\newcommand{\var}{\operatorname{Var}}
\newcommand{\RV}{\operatorname{RV}}
\newcommand{\MCB}{\mathsf{MCB}}
\newcommand{\DSC}{\mathsf{DSC}}
\newcommand{\UNC}{\mathsf{UNC}}
\newcommand{\Sp}{\overline{\mathsf{S}}}
\newcommand{\MCBp}{\overline{\mathsf{MCB}}}
\newcommand{\DSCp}{\overline{\mathsf{DSC}}}
\newcommand{\UNCp}{\overline{\mathsf{UNC}}}
\newcommand{\Ss}{\widehat{\mathsf{S}}}
\newcommand{\MCBs}{\widehat{\mathsf{MCB}}}
\newcommand{\DSCs}{\widehat{\mathsf{DSC}}}
\newcommand{\UNCs}{\widehat{\mathsf{UNC}} }
\newcommand{\uMCBp}{\overline{\mathsf{uMCB}}}
\newcommand{\cMCBp}{\overline{\mathsf{cMCB}}}

\newcommand{\R}{\mathbb{R}}
\newcommand{\E}{\mathbb{E}}
\newcommand{\A}{\mathsf{A}}
\renewcommand{\O}{\mathsf{O}}

\newcommand{\interior}{\operatorname{int}}

\DeclareMathOperator*{\argmin}{arg\,min}
\newcommand{\sgn}{\operatorname{sgn}}

\newcommand{\N}{\mathbb N}
\renewcommand{\P}{\mathbb P}

\renewcommand{\H}{\mathcal H}
\renewcommand{\a}{\alpha}

\newcommand{\toP}{\stackrel{\mathbb{P}}{\longrightarrow}}
\newcommand{\toD}{\stackrel{d}{\longrightarrow}}

\theoremstyle{plain}
\newtheorem{assumption}{Assumption}[section]
\newtheorem{thm}[assumption]{Theorem}

\newtheorem{prop}[assumption]{Proposition}
\newtheorem{lemma}[assumption]{Lemma}

\date{\today}

\title{Statistical Inference for Score Decompositions}

\author{
	Timo Dimitriadis\thanks{Corresponding Author. Faculty of Economics and Business, Goethe University Frankfurt, Campus Westend, Theodor-W.-Adorno-Platz 4, 60629 Frankfurt am Main, Germany, and Heidelberg Institute for Theoretical Studies (HITS), \href{mailto:dimitriadis@econ.uni-frankfurt.de}{dimitriadis@econ.uni-frankfurt.de}.} \and  Marius Puke\thanks{Institute of Economics, University of Hohenheim, Schloss Hohenheim 1 C, 70593 Stuttgart, Germany, \href{mailto:marius.puke@uni-hohenheim.de}{marius.puke@uni-hohenheim.de}.}
}

\begin{document}
	\maketitle
	\begin{abstract}
		We introduce inference methods for score decompositions, which partition scoring functions for predictive assessment into three interpretable components: miscalibration, discrimination, and uncertainty. 
		Our estimation and inference relies on a linear recalibration of the forecasts and is applicable to general % (multi-step ahead) 
		point forecasts such as means and quantiles due to its validity for non-smooth scoring functions.
		This approach ensures non-negative decomposition terms in finite samples, enables asymptotic inference under model misspecification, and establishes a direct connection to the classical Mincer-Zarnowitz regression.
		The resulting inference framework facilitates novel tests for equal linearized forecast calibration or discrimination, which yield three key advantages. 
		They enhance the information content of predictive ability tests by decomposing scores, can improve detection power in scenarios where predictive differences are attributable to specific score components, and formally connect scoring-function-based evaluation to traditional calibration tests, such as financial backtests.
		Applications demonstrate the method's utility. 
		We find that for survey inflation forecasts, discrimination abilities can differ significantly even when overall predictive ability does not. 
		In an application to financial risk models, our tests provide deeper insights into the calibration and information content of volatility and Value-at-Risk forecasts. 
		By disentangling forecast accuracy from backtest performance, the method exposes critical shortcomings in current banking regulation.\\[0.1cm]
		
		
	\end{abstract}

	\pagebreak 
	\section{Introduction}
	\label{sec:Intro}
	
	Forecasting---whether of macroeconomic conditions, financial market movements, or other uncertain outcomes---has become an indispensable element of modern economic analysis.
	From monetary policy design to investment planning and risk management, economic decisions increasingly rely on the ability to anticipate what lies ahead. 
	The value of these decisions, however, depends critically on the accuracy and credibility of the underlying forecasts. 
	Therefore, establishing rigorous standards for evaluating predictions is  essential, ensuring that economic forecasts are not only informative but also empirically grounded and verifiable.
	
	Two main concepts have emerged for evaluating a forecast $X_t$ of an uncertain future real-valued outcome $Y_t$. 
	First, the statistical theory of forecast comparison ranks competing forecasts according to their empirical score (or loss) \citep{S_1971,G_2011a}. 
	The appropriate scoring function depends on the forecast target---e.g., whether the aim is to predict the full distribution of $Y_t$, its mean, or a particular quantile. 
	For meaningful comparisons, these scoring functions must be \emph{strictly consistent}: the ideal forecast should minimize the expected score, ensuring that lower score values reliably indicate superior predictive performance.
	
	Second, forecasts are evaluated in terms of their alignment with the realized outcomes---a property referred to as \emph{calibration} in statistics \citep{Gneiting2007probabilistic, GR_2023}, \emph{forecast optimality}, \emph{rationality} or \emph{efficiency} in economics \citep{D_1998, Elliott2016}, and \emph{backtesting} in finance \citep{B_2023_vola,B_2019}.
	For example, when forecasting the mean, the average forecast should coincide with the average realized outcome—both unconditionally and conditional on the forecast itself, the latter being referred to as auto-calibration.
	In economics, this property is most commonly assessed using the test of \citet[henceforth MZ]{MZ_1969}, which examines whether the forecast $X_t$ equals the conditional expectation of $Y_t$ given $X_t$ through a linear regression framework.
	Extensions to other targets, such as quantiles, follow naturally by replacing the mean regression with the corresponding quantile regression \citep{Gaglianone2011, Guler2017, BD_2020}.
	
	These two evaluation paradigms are known to occasionally produce counterintuitive results, a point illustrated by the following motivational example.
	In this paper, we develop asymptotic inference for the components of so-called score decompositions, which provide a transparent link between the two forecast evaluation regimes described above. Beyond delivering a sharper theoretical understanding of the respective strengths and limitations of these approaches, the motivational example underscores their practical relevance by demonstrating how, in applied settings, the two methods can lead to seemingly contradictory conclusions.

	\subsection{Motivational example}
	\label{sec:motex}
	
	Consider the daily logarithmic returns % from January 1, 2000 to January 1, 2022 
	of the S\&P E-mini future as the world’s most actively traded equity index futures contract.
	We evaluate one-day ahead forecasts for its variance as well as its $1\%$-quantile, known as the Value-at-Risk (VaR), between   January 1, 2008 and January 1, 2022.
	Both forecast types are important tools in banking regulation; see
	\citet[\href{https://www.bis.org/basel_framework/chapter/CRE/36.htm?inforce=20191215&published=20191215&export=pdf}{CRE 36.153}]{B_2023_vola} and \citet[\href{https://www.bis.org/basel_framework/chapter/MAR/32.htm?inforce=20230101&published=20200327}{MAR 32.5} and \href{https://www.bis.org/basel_framework/chapter/MAR/99.htm?inforce=20230101&published=20200327}{MAR 99.9} -- \href{https://www.bis.org/basel_framework/chapter/MAR/99.htm?inforce=20230101&published=20200327}{MAR 99.21}]{B_2019} for volatility and VaR backtesting, respectively.
	Section~\ref{sec:appl_vola} provides additional details on this application.
	
	We use the following state of the art forecasting models:
	First, the Heterogeneous Auto-Regressive (HAR) model of \citet{C_2009} uses intraday trading information through the Realized Variance (RV) estimator to forecast the variance through a linear model.
	Second, we employ a baseline GARCH model of \citet{B_1986}. For both the HAR and GARCH model, we obtain quantile forecasts through a Gaussianity assumption on the model residuals.
	For the case of VaR forecasts, we also consider the basic Historical Simulation (HS) approach, which uses the unconditional quantile of the past 250 trading days as its quantile forecast.

	\begin{table}
		\caption{\citet{MZ_1969} test $p$-values, average scores (QLIKE for variance, and check loss for the quantile) and its decomposition terms for variance and VaR forecasts. See  Sections~\ref{sec:motex} and \ref{sec:appl_vola} for details.}
		\label{tab:motivation}
		\centering
		\begin{tabular}[t]{l l r l r r r r}
			\toprule
			Forecast Type &	Model $i$ & MZ $p$-value & & Score &\multicolumn{3}{c}{Score Decomposition} \\
			\cmidrule(l{3pt}r{3pt}){6-8}  
			& &     && $\myS_i$ & $ \MCB_i$ & $ \DSC_i $  & $ \UNC $\\
			
			\midrule
			
			%\multicolumn{7}{l}{ \underline{a) Variance Forecasts}} \\[0.5em]
			\multirow{2}{*}{Variance}  & HAR  	& 0.001 	  &&  1.775 & 0.034  &  6.732  & \multirow{ 2}{*}{$ 8.473 $} 		\\ 
			& GARCH &  0.991 && 2.522	&  0.000  & 5.951 	& \\
			
			\midrule
			
			% \multicolumn{7}{l}{ \underline{b) 1\%-VaR Forecasts}} \\[0.5em]
			\multirow{3}{*}{$1\%$-Quantile} & HAR  	& 0.003 && 4.018 & 0.270 & 1.831  &  \multirow{ 3}{*}{$5.579$}  	\\
			& GARCH & 0.006 && 4.089	&  0.148 & 1.638	&  	\\
			& HS  	& 0.321 &&  5.114& 0.077  & 0.543 & 	\\ 
			\bottomrule
		\end{tabular}
	\end{table}
	
	Table~\ref{tab:motivation} shows that, for both forecast types, the HAR model achieves the lowest average scores. These improvements are statistically significant at the $1\%$ level according to the \citet[henceforth DM]{DM_1995} test, with the sole exception of the HAR–GARCH comparison for VaR forecasts.
	Thus, the HAR forecasts deliver the strongest predictive performance.
	At the same time, the MZ test decisively rejects auto-calibration of the HAR forecasts. In contrast, the GARCH variance forecasts yield significantly higher scores, yet auto-calibration cannot be rejected. 
	A similar pattern emerges for the HS VaR forecasts: they record the largest scores but constitute the only forecasting sequence for which calibration is not rejected.
	
	This example illustrates that superior predictive performance and proper auto-calibration do not necessarily coincide and may diverge sharply. Since both evaluation approaches are widely applied in practice, developing tools to reconcile and interpret these seemingly contradictory conclusions is of central importance.
	
	We shed light on these results by examining the \emph{score decomposition} reported in Table~\ref{tab:motivation}. 
	Score decompositions have a long tradition in meteorology \citep{M_1973,Dawid1986,MW_1987,B_2009,BF_2014} and have recently attracted renewed interest in statistics \citep{DGJ_2021,GR_2023,G_etal_2023}, yet they have received comparatively little attention in economics. 
	In essence, the average score $\myS_i$ of forecaster $i$ can be decomposed additively into three non-negative components capturing miscalibration ($\MCB_i$), discrimination ($\DSC_i$), and uncertainty ($\UNC$):
	\begin{align} 
		\label{eq:DecompositionIntro}
		\myS_{i} = \MCB_{i} - \DSC_{i} + \UNC,
	\end{align}
	whose exact specification we explain below around \eqref{eqn:PopScoreDecomp} and \eqref{eqn:FiniteSampleDecomposition}.
	
	Intuitively, the $\MCB$ component measures the extent of forecast miscalibration by quantifying how much the score improves after recalibrating the forecasts by regressing outcomes on forecasts as in the MZ test.
	% the forecasts were perfectly calibrated---thus capturing a concept closely related to what the MZ test evaluates.
	Likewise, the $\DSC$ component reflects the forecaster’s ability to distinguish (or ``discriminate'') between higher and lower outcome realizations by comparing the scores of the linearly recalibrated forecasts with a constant reference forecast.
	The remaining $\UNC$ term is forecast-independent and therefore functions as a reference level for the decomposition.

	The last columns of Table~\ref{tab:motivation} help reconcile the discrepancies between the score-based and calibration-based assessments. For both forecast types, the HAR model attains the highest (and thus best) $\DSC$ values, which is intuitive given that it incorporates the richest information set through high-frequency returns. However, it also yields the largest (worst) $\MCB$ values, accounting for its rejection in the MZ test.
	By contrast, the GARCH model---and even more so the HS model---exhibits weaker discrimination but achieves substantially better calibration.
	Yet, despite the clear interpretive value of these decomposition terms, formal inference methods for them are not available, limiting their practical usefulness as reflected in \citet[p.\ 355]{ET_2016}, who point out that ``am important limitation of these decompositions is that there are no objective measures for how large or small the terms should be''.

	This gap is particularly problematic in settings such as banking regulation under the Basel framework, where calibration-based backtesting plays a central role. 
	Such procedures risk creating misguided incentives: the best-performing HAR forecasts may be labeled ``invalid,'' whereas HS forecasts---whose poor discriminatory ability makes them slow to respond to emerging crises---may be deemed acceptable.
	These tensions highlight the need for both a deeper understanding of the distinct properties captured by the decomposition and the development of rigorous statistical inference for its components---needs that this paper directly addresses.

	\subsection{Contributions}
	\label{sec:Contributions}
	
	In this paper, we propose augmenting scoring-function-based forecast comparisons with the corresponding decomposition terms, as illustrated in Table~\ref{tab:motivation}, and we are the first to develop asymptotic inference for these components.
	To estimate the decomposition components in \eqref{eq:DecompositionIntro}, we advocate the use of a linear regression---either mean or quantile, depending on the forecast target. 
	This specification is motivated by several considerations: its stability and resistance to overfitting; its good empirical fit in our applications; its close conceptual connection to the \citet{MZ_1969} framework; its attractive non-negativity guarantees for the resulting decomposition terms; and, crucially, its tractability for developing valid asymptotic inference procedures under model misspecification.
	Recent contributions have suggested nonparametric (isotonic) regression approaches instead \citep{DGJ_2021, GR_2023, arnold2024decompositions}. However, their asymptotic properties remain unclear---particularly in time-series settings---and they are prone to overfitting, which can introduce biases that further complicate or even undermine asymptotic inference \citep{roelofs2022mitigating}.
	
	We derive the joint asymptotic distribution of the estimated $\MCB$ and $\DSC$ components for competing forecasts allowing for possibly non-smooth scoring functions, thereby enabling new tests of equal miscalibration or equal discrimination for, e.g., mean or quantile forecasts.
	These tests extend the classical framework of \citet{DM_1995} by incorporating two layers of uncertainty: the sampling variation arising from averaging over time and the additional variation introduced by the (possibly misspecified) linear ``recalibration regression'' used in constructing the $\MCB$ and $\DSC$ components.
	A key complication is that the limiting distribution depends on whether the population values of $\MCB$ or $\DSC$ are zero or strictly positive. 
	Our testing procedures account for these cases through a $p$-value adjustment inspired by a hybrid intersection–union/union–intersection (IU/UI) testing principles \citep{Berger1997likelihood}.
	
	Our simulations demonstrate that the proposed tests are valid—though potentially conservative due to the underlying IU/UI principles---under their respective null hypotheses for both mean and quantile forecasting settings. 
	We also show that the tests exhibit strong power properties.
	In particular, we show that the proposed tests can provide increased sensitivity in settings where differences in average scores are driven exclusively by a single component, either $\MCB$ or $\DSC$. 
	In such cases, they yield higher rejection frequencies than the classical DM test. 
	This reflects that our tests isolate a specific aspect of forecast performance---either discrimination or miscalibration---whereas score-based tests aggregate both components into a single, potentially noisier measure.

	Our first application compares the predictive performance of two survey forecasts for U.S.~inflation—the Survey of Professional Forecasters (SPF) and the Michigan Survey of Consumers. While professional forecasters unsurprisingly attain better average scores than consumers, these differences are typically insignificant \citep{EGJK_2016, P_2020}. In contrast, the score decomposition shows that both groups exhibit similar calibration, whereas our new tests detect a significant difference in discrimination. This finding is consistent with the higher rejection frequencies observed in the corresponding simulation designs. Given that professional forecasters draw on a richer information set, their superior discrimination is a natural consequence.
	
	The second application extends the motivating example of Section~\ref{sec:motex} and clarifies the relationship between scoring-function-based forecast evaluation and backtests for both, variance and VaR forecasts.
	Moreover, we disentangle several widely used VaR backtests by characterizing them through the conditioning sets they implicitly impose within the framework of conditional quantile calibration, thereby revealing the precise properties they actually test.

	\subsection{Literature}
	\label{sec:Literature}
	
	Early approaches to score decompositions resembling \eqref{eq:DecompositionIntro} date back to \citet{sanders1963subjective}, \citet{T_1966}, and \citet{M_1973}, and primarily focus on the squared error in the context of probability forecasts for binary outcomes, which can be interpreted as a mean. 
	\citet{degroot1981assessing} and \citet{B_2009} demonstrated that such decompositions can be derived for any strictly proper scoring rule, including those for full distributional forecasts. 
	More recently, \citet{GR_2023} developed a general framework for score decompositions for scoring functions, extending these concepts to point forecasts. 
	This line of research has spurred a growing body of both theoretical and applied work; see, for example, \citet{BF_2014, S_2017, Pohle2020, FLM_2022, Wetal_2025, Betal_2024}. 
	For a more detailed historical overview of score decompositions, see \citet[Section~1.1]{mitchell2019score}.

	In practice, the score decompositions rely on estimates of recalibrated forecasts, and the choice of recalibration method remains an active area of research. Recently, isotonic regression—which estimates, for example, a conditional mean nonparametrically under a monotonicity constraint—has been widely used for both point and distributional forecasts \citep{DGJ_2021, G_etal_2023, DGJV_2024, arnold2024decompositions, allen2025assessing}. A key reason for its popularity is that it satisfies attractive finite-sample non-negativity conditions, which facilitate the interpretation of the decomposition terms; these conditions are generally not satisfied by alternative methods such as binning or kernel regressions \citep{DGJ_2021}. We show that, under mild conditions, these non-negativity properties are also preserved when recalibration is performed using (possibly misspecified) linear regressions.
	
	As we develop asymptotic inference methods for the estimated decomposition terms on the right-hand side of \eqref{eq:DecompositionIntro}, our approach is closely related to statistical tests for predictive performance, which focus on the sampling variability of the left-hand side of \eqref{eq:DecompositionIntro}. The seminal contribution of \citet{DM_1995} laid the foundation for a large econometric literature on testing for equal forecast performance. Numerous extensions have since been proposed, including methods for model-based forecasts \citep{West1996asymptotic, clark2001tests}, multiple forecast comparisons \citep{Hansen2005test, HLN_2011}, conditional predictive ability \citep{Giacomini2006tests, Li2022conditional}, and anytime-valid procedures \citep{Henzi2022valid, Choe2024comparing}. For comprehensive surveys of this literature, see, for example, \citet{west2006forecast}, \citet{ClarkMcCracken2013}, and \citet{Diebold2015comparing}.
	
	In contrast, the assessment of sampling uncertainty for the right-hand side of \eqref{eq:DecompositionIntro} has received relatively little attention. Existing contributions, such as \citet{S_2014} and \citet{GR_2023}, propose ad-hoc (resampling) methods to quantify the uncertainty for a single forecaster, but they do not provide tools for comparative inference across multiple forecasters. 
	We address this gap by developing a unified framework for inference on the decomposition components.
	
	VaR backtests are widely used to evaluate the reliability, i.e., the calibration, of risk models via their quantile forecasts.
	Classic procedures include the unconditional coverage test of \citet{K_1995}, which underlies the traffic-light system of \citet{B_1996, B_2019}, as well as the conditional coverage test of \citet{C_1998}, which additionally assesses the independence of exceedances.
	A large body of subsequent work tests slightly differing notions of conditional calibration \citep{BCP_2011, Gaglianone2011, BCP_2011, HD_2023, FGP_2023, wang2025backtesting}; see \citet{HKP_2022} for a recent overview.
	Our score decompositions---and in particular the application in Section \ref{sec:appl_vola}---show that well-calibrated forecasts do not necessarily exhibit strong predictive performance.
	This underscores the importance of evaluating VaR forecasts through both calibration and discrimination, rather than relying solely on traditional backtesting metrics; also see Section~\ref{sec:VaRBacktestsTheory}.

	A growing strand of the macroeconomic literature---including, among many others, \citet{CG_2015}, \citet{BGMS_2020}, \citet{kohlhas2021asymmetric}, \citet{AHS_2021}, \citet{broer2024forecaster}, and \citet{FNS_2024}---examines how forecast errors respond to observable information.
	This line of work effectively tests necessary conditions for conditional calibration, asking whether forecast errors are linearly predictable given specific observable information available at the time forecasts are made.
	Our methodology suggests a natural refinement of these approaches by additionally assessing forecast discrimination, since calibration captures only one dimension of predictive performance.
	This refinement is relevant because forecasts that discriminate well are often found to be miscalibrated, and thus deviate from rational expectations more frequently.

	\subsection{Plan of the article}
	
	Section~\ref{sec:Fundamentals} introduces key concepts of forecast evaluation. 
	The theory in Section~\ref{sec:singlemain} develops asymptotic inference under high-level conditions, which are subsequently verified for mean and quantile forecast.
	Section~\ref{sec:sim} contains simulations and we present two applications in Section~\ref{sec:appl}.
	Section~\ref{sec:conclusion} concludes. 
	
	In the supplementary material, Appendix~\ref{sec:Proofs} contains all proofs and we discuss the relation to further decompositions in Appendix~\ref{sec:connecttolit}.
	Appendix~\ref{app:QuantileSimDGP} presents technical details for the simulations and Appendix~\ref{sec:AddFigs} contains additional figures.
	
	We provide replication material at \href{https://github.com/marius-cp/SDI_replications}{https://github.com/marius-cp/SDI\_replications} for the simulations and applications, which draws on the \textsf{SDI} package for the statistical software \textsf{R} \citep{R2026}, available at \href{https://github.com/marius-cp/SDI}{https://github.com/marius-cp/SDI}.

	\section{Fundamentals of forecast evaluation}
	\label{sec:Fundamentals}
	
	After fixing notation in Section~\ref{sec:Setup}, we discuss relative and absolute forecast evaluation in Sections~\ref{sec:RelativeFCEval} and \ref{sec:ForecastCalibration}.

	\subsection{Setup and notation}
	\label{sec:Setup}
	
	We follow \citet{G_2011a} and \citet{HogaDimi2023} and consider a complete probability space $(\Omega, \mathcal{F}, \mathbb{P})$ with a series of stationary random variables $Y_t: \Omega \to \mathsf{O} \subseteq \R$, $t \in \N$, that take values in the interval-valued observation domain $\mathsf{O}$.
	For some information set (sigma-algebra) $\mathcal{F}_t \subset \mathcal{F}$, where $Y_t$ is in general not $\mathcal{F}_t$-measurable, we denote the conditional distribution of $Y_t$ given $\mathcal{F}_t$ by $F_t \in \mathcal{P}$ for some convex class of distributions $\mathcal{P}$ taking values in $\mathsf{O}$.\footnote{Formally, the conditional distributions are random variables $F_t(\bullet, \omega) = \P(Y_t \le \bullet \mid \mathcal{F}_t)(\omega)$, where we assume that $F_t(\bullet, \omega) \in \mathcal{P}$ for all $\omega \in \Omega$, for some convex class of distributions $\mathcal{P}$ taking values in $\mathsf{O}$, which contains the Dirac measures $\delta_y$ for all $y \in \mathsf{O}$.
		We write $F_t(\bullet)$ for the random variable defined by $\omega \mapsto F_t(\omega, \bullet)$. }
	The information set $\mathcal{F}_t$ represents the full information that could be used for forecast $Y_t$, but competing forecasters might have less information at hand.
	We deliberately use the notation $\mathcal{F}_t$ (opposed to e.g.,  $\mathcal{F}_{t-1}$) to allow for both time series (where $t$ represents time points) as well as cross-sectional (where $t$ represents cross-sectional units) applications.

	For an interval-valued action domain $\mathsf{A} \subseteq \R$, we consider a stationary sequence of forecasts for some target functional $\Gamma: \mathcal{P} \to \mathsf{A}$ of the conditional distribution $F_t$.
	Prominent functionals are the conditional expectation $\E [ Y_t \mid \mathcal{F}_{t}]$, or a conditional $\alpha$-quantile $Q_\alpha(Y_t \mid \mathcal{F}_{t})$ for some $\alpha \in (0,1)$, which is assumed to be unique for simplicity.
	We sometimes abuse notation and write $\Gamma( Y_t \mid \mathcal{F}_{t})$ instead of $\Gamma({F}_{t})$ for a functional of the conditional distribution.
	
	The point forecasts are denoted by $X_{it}$ where the index $i=1,2$ denotes a pair of competing forecasts.
	These forecasts are assumed to be $\mathcal{F}_{t}$-measurable throughout the article but they are in general based on less information, $\mathcal{I}_{it} \subseteq \mathcal{F}_{t}$, or might be misspecified through other aspects such as model mispecification or estimation noise.
	
	Throughout, (in)equalities of random variables are meant to hold almost surely (a.s.) unless stated otherwise, and for a vector $\mathbf{z} \in \mathbb{R}^k$, $k \in \mathbb{N}$, $\lVert \mathbf{z} \rVert$ denotes the Euclidean norm and $\mathbf{z}^\top$ its transpose. 
	We write $\nabla_{\mathbf{z}} f$ for the (column-vector) gradient of the mapping $\mathbf{z} \mapsto f(\mathbf{z}) \in \mathbb{R}$ and $\nabla_{\mathbf{z} \mathbf{z}} f$ for its second derivative, the Hessian matrix.
	In the univariate case, $f'(z)$ denotes the derivative of $f(z)$ with respect to $z \in \mathbb{R}$. For functions with multiple
	real-valued arguments, a prime indicates partial differentiation with respect to the first argument, and multiple primes denote higher-order derivatives.

	\subsection{Forecast comparison through scoring functions}
	\label{sec:RelativeFCEval}
	
	Following \citet{G_2011a}, point forecasts are evaluated using strictly
	consistent scoring (or loss) functions
	$\myS:\mathsf{A}\times\mathsf{O}\to\R$, where lower values indicate better
	predictive performance. A scoring function is \emph{(strictly) consistent} for
	a functional $\Gamma$ if the expected score is (uniquely) minimized by the
	optimal forecast. A functional is called \emph{elicitable} if it admits a
	strictly consistent scoring function. \citet[Theorem~1]{HE_2014} extend this
	concept to conditional information sets
	$\mathcal{G}_t\subseteq\mathcal{F}_t$, yielding
	\begin{align}
		\label{eqn:StrictConsistencyConditional}
		\E \big[ \myS \big(\Gamma(Y_t \mid \mathcal{G}_t), Y_t \big) \mid \mathcal{G}_{t} \big]
		\le
		\E \big[ \myS(\widetilde{X}_t, Y_t) \mid \mathcal{G}_{t} \big],
	\end{align}
	for every $\mathcal{G}_t$-measurable random variable $\widetilde{X}_t$, with
	almost sure equality if and only if
	$\widetilde{X}_t=\Gamma(Y_t\mid\mathcal{G}_t)$ almost surely.
	
	If the target functional $\Gamma$ is the conditional mean, under mild regularity conditions, all strictly consistent scoring functions are characterized by the so-called \emph{Bregman class},
	\begin{align}
		\label{eqn:BregmanLoss}
		\myS(x,y) = \phi(y) - \phi(x) - \phi'(x)(y-x) 
	\end{align}
	parametrized by a strictly convex function $\phi$ with subgradient $\phi'$ \citep[Theorem~7]{G_2011a}.
	The omnipresent squared error  
	\begin{align}
		\label{eqn:SEscore}
		\myS_{\textsf{SE}}(x,y)=(x-y)^2.
	\end{align}
	arises for $\phi_{\textsf{SE}}(x) = x^2$.
	For $\phi_{\textsf{QL}}(x) = -\log(x)$ we obtain the QLIKE scoring function, 
	\begin{align}
		\label{eqn:QLIKEscore}
		\myS_{\textsf{QL}}(x,y) = \frac yx - \log\left( \frac yx \right)- 1,
	\end{align}
	which is regularly used for the evaluation of variance forecasts in financial econometrics with positive action and observation domain,  $\mathsf{A} = \mathsf{O} = (0, \infty)$; see e.g., \citet{P_2011}.
	Under the standard assumption of a zero conditional mean, the variance coincides with the second moment, thereby rendering the Bregman class applicable.
	
	If the target functional is the conditional $\alpha$-quantile, $\alpha \in (0,1)$, strictly consistent scoring functions are given by the class of generalized piecewise linear (GPL) loss functions 
	\begin{align}
		\label{eqn:GPLLoss}
		\myS(x,y) = (\mathds{1}\{y \le x\} - \alpha) (g(x) - g(y)),
	\end{align}
	where the function $g$ is strictly increasing on $\mathsf{A} = \mathsf{O} \subseteq \mathbb{R}$ \citep[Theorem~9]{G_2011a}.
	The GPL loss functions are generally non-differentiable at $x=y$ due to the indicator function.
	The check loss $\myS_\textsf{CL}(x,y) = (\mathds{1}\{y \le x\} - \alpha) (x - y)$ as the most popular candidate arises for $g_{\textsf{CL}}(x) = x$.
	
	In practice, scoring functions are employed to compare forecast sequences 
	$X_{it}$, $i=1,2$, $t=1,\dots,T$, for the target variable $Y_t$ through their
	average scores $\widehat{\myS}_{iT} = \frac{1}{T} \sum_{t=1}^{T} \myS(X_{it}, Y_t)$.
	The sample average serves as an empirical approximation to the
	expectation appearing in \eqref{eqn:StrictConsistencyConditional}.
	Following \citet{DM_1995}, a rich body of asymptotic inference methods is available for constructing tests based on such average score differences.

	\subsection{Forecast validation and calibration assessment}
	\label{sec:ForecastCalibration} 
	
	A known practical limitation of purely score-based forecast evaluations lies in their inability to identify specific deficiencies of the forecasts, such as systematic biases, which we formalize through the theory of forecast calibration.
	A forecast $X_{it}$ for $\Gamma(Y_t)$ is said to be \emph{conditionally calibrated} with respect to the information set $\mathcal{G}_{it} \subseteq \mathcal{F}_t$ if 
	\begin{align}
		\label{eqn:ConditionalCalibration}
		X_{it} = \Gamma(Y_t \mid \mathcal{G}_{it}) \qquad \text{a.s.}
	\end{align}
	This property asserts a certain compatibility of the forecast $X_{it}$ with its forecasting target $\Gamma(Y_t \mid \mathcal{G}_{it})$.
	A conditionally calibrated forecast w.r.t.\ the full information set $\mathcal{F}_t$ minimizes the $\mathcal{F}_t$-conditional expectation of the scoring function by \eqref{eqn:StrictConsistencyConditional}, such that the econometric literature also calls it \emph{optimal} or \emph{efficient} \citep[Chapter 15]{ET_2016}.
	As the forecast might not have used---or even have access to---all information in $\mathcal{F}_t$, it is, however, often more appropriate to consider calibration based on some reduced information set $\mathcal{G}_{it} \subset \mathcal{F}_t$ in \eqref{eqn:ConditionalCalibration}.
	
	An important special case is that of an \emph{auto-calibrated} forecast, obtained for $\mathcal{G}_{it} = \sigma\{X_{it}\}$,~i.e.,
	\begin{align}
		\label{eqn:AutoCalibration}
		X_{it} = \Gamma(Y_t \mid X_{it}), \qquad \text{a.s.}
	\end{align}
	Conditioning on the forecast itself guarantees that the information encoded in the forecast was indeed used, and it remains feasible even when the evaluator does not have access to the full information set of the forecaster $\mathcal{I}_{it}$.
	
	A common tool to assess auto-calibration is to plot $x \in \mathsf{A}$ against an estimate of $\Gamma(Y_t \mid X_{it} = x)$, where deviations from the diagonal indicate miscalibration.
	In the meteorological literature, this is known as a \emph{reliability diagram}, typically estimated nonparametrically \citep{MurphyWinkler1992, Pohle2020, DGJ_2021}.
	In economic mean forecast evaluation, this idea goes back to \citet{MZ_1969}, who regress $Y_t$ on $X_{it}$ and test whether the intercept is zero and the slope is one---essentially a linear reliability diagram. 
	\citet{Gaglianone2011, Guler2017, BD_2020} cover extensions to other functionals.
	
	A useful alternative characterization of calibration is based on
	\emph{identification functions}, which, under suitable smoothness conditions,
	correspond to derivatives of strictly consistent scoring functions
	\citep{NZ_2017}. An identification function
	$\myV:\mathsf{A}\times\mathsf{O}\to\R$ is \emph{strict} if
	$\E_{Y\sim F}[\myV(x,Y)]=0$ if and only if $x=\Gamma(F)$. Subject to regularity
	conditions, given $\mathcal{G}_{it} \supseteq \sigma\{X_{it}\}$, conditional calibration in \eqref{eqn:ConditionalCalibration} is
	equivalent to
	\begin{align}
		\label{eqn:CondCalibrationIDFunction}
		\E \big[ \myV(X_{it}, Y_t) \mid \mathcal{G}_{it} \big] = 0, \qquad \text{a.s.},
	\end{align}
	which reduces to auto-calibration when
	$\mathcal{G}_{it}=\sigma\{X_{it}\}$. Taking unconditional expectations in \eqref{eqn:CondCalibrationIDFunction} yields the notion of \emph{unconditional calibration}, $\E[\myV(X_{it},Y_t)]=0$.

	\section{Theory}
	\label{sec:singlemain}
	
	We introduce linear score decompositions in Section~\ref{sec:popandsampdec}, develop asymptotic inference in Section~\ref{sec:estimation}, discuss testable hypotheses in Section~\ref{sec:Testing}, derive low-level inference conditions for mean and quantile forecasts in Section~\ref{sec:LowLevelConditions}, and discuss connections to VaR backtests in Section~\ref{sec:VaRBacktestsTheory}.

	\subsection{Score decompositions based on linear recalibration}
	\label{sec:popandsampdec}
	
	For our score decomposition, let $\Gamma$ be an elicitable functional and
	$\myS$ a corresponding strictly consistent scoring function. We introduce the
	\emph{recalibration variables}
	$\mW_{it} = (1, X_{it}, \dots)^\top \in \mathbb{R}^k$, which generate the
	recalibration information set $\mathcal{W}_{it} = \sigma\{\mW_{it}\}$, where
	$\sigma\{X_{it}\} \subseteq \mathcal{W}_{it} \subseteq \mathcal{F}_t$.
	The standard specification is $\mW_{it} = (1, X_{it})$. More generally,
	$\mW_{it}$ may include higher-order transformations of $X_{it}$ (e.g.,
	polynomial terms) as well as additional variables on which conditioning is of
	interest. The canonical auto-calibration case, in the sense of
	\citet{MZ_1969}, corresponds to $\mathcal{W}_{it} = \sigma\{X_{it}\}$.
	Enlarging the information set $\mathcal{W}_{it}$ naturally accommodates
	stronger notions of calibration and establishes connections to VaR backtesting
	and related macroeconomic applications \citep{CG_2015}.
	
	We impose the following standard assumption in the field of misspecified regressions \citep{White1980, KimWhite2003, Komunjer2005, Angrist2006}.
	
	\begin{assumption}
		\label{ass:TrueLinearSpecification}
		$\myS$ is a strictly consistent scoring function for the functional $\Gamma$, the random variables $Y_t$, $X_{it}$ and $\mW_{it}$ are strictly stationary for all $t \in \N$, and the pseudo-true parameter 
		\begin{align}
			\label{eqn:PseudoTrue}
			\bar{\mtheta}_{i} = \underset{\mtheta \in \mTheta}{\arg \min} \; \E \left[ \myS \big(\mW_{it}^\top \mtheta , Y_t \big) \right]
		\end{align}
		is unique and in the interior of $\mTheta \subseteq \R^k$, where $\mTheta$ is convex and compact, the points $(0,1,0, \dots)^\top$ and $(\Gamma(Y_t),0,\dots)^\top$ are in the interior of $\mTheta$, and $\mW_{it}^\top \mtheta \in \mathsf{A}$ for all $\mtheta \in \mTheta$.
	\end{assumption}

	Besides stationarity and some technicalities on $\mTheta$, Assumption \ref{ass:TrueLinearSpecification} crucially imposes the existence of a unique pseudo-true parameter $\bar\mtheta_i$ for the linearly \emph{reCalibrated forecast} $X_{it}^\mathsf{C} = \mW_{it}^\top \bar{\mtheta}_i$, which we later estimate through linear (e.g., mean or quantile) regression.
	The forecast $X_{it}^\mathsf{C}$ is a recalibrated version of the original forecast that exploits the information contained in the original forecast.
	% is an ``improved'' version of the original forecast that utilizes the information of the original forecast through recalibration.
	Under correct specification of the linear recalibration, $X_{it}^\mathsf{C} = \Gamma(Y_t \mid \mW_{it})$ is $\mathcal{W}_{it}$-conditional calibrated by definition; see \eqref{eqn:ConditionalCalibration}.
	Under model misspecification, $X_{it}^\mathsf{C}$ satisfies the interpretation of a \emph{linearly} recalibrated forecast in the sense of a best linear predictor for $\Gamma(Y_t \mid \mW_{it})$.
	Notice that the linear model is motivated by the benchmark case of calibrated forecasts, for which the linear model is correctly specified \citep{MZ_1969}; throughout, however, we allow for possible misspecification.

	We denote the expected score of the original forecasts $X_{it}$ by $\Sp_{i} := \E\big[ \myS(X_{it},Y_t) \big]$.
	Based on the linearly recalibrated forecast $X_{it}^\mathsf{C}$, the corresponding linearly \emph{reCalibrated score} is denoted by $\Sp_{i}^\mathsf{C} = \E \big[ \myS(X_{it}^\mathsf{C}, Y_t) \big]$.
	We further define the \emph{Reference} forecast $\bar{r} = \Gamma(Y_t)$, which is unconditionally calibrated but constant, and the corresponding \emph{Reference score} $\Sp^{\mathsf{R}} = \E \big[ \myS( \bar{r}, Y_t) \big]$.
	With these ingredients, the (linearized) \emph{population} score decomposition (with non-negative components) is given by
	\begin{align}
		\label{eqn:PopScoreDecomp}
		\Sp_{i}
		&= \underbrace{\big( \Sp_{i} - \Sp_{i}^{\mathsf{C}} \big)}_{ \MCBp_{i}}
		- \underbrace{\big( \Sp^\mathsf{R} - \Sp_{i}^\mathsf{C} \big)}_{ \DSCp_{i}} 
		+ \underbrace{\Sp^{\mathsf{R}}}_{\UNCp},
	\end{align}
	whose intuition we explain in the following.
	
	In the canonical case $\mW_{it} = (1, X_{it})$
	and under correct specification of the linear model, the component $\MCBp_{i} = \Sp_{i} - \Sp_{i}^{\mathsf{C}}$ quantifies the increase in the score's predictive ability that can be achieved by recalibration, hence providing a measure of miscalibration such that low $\MCBp_{i}$ values are desirable.
	As the recalibrated forecast is a measurable transformation of $X_{it}$, it incorporates no additional information; thus, any score gains stem purely from improved calibration.
	If $\MCBp_{i} = 0$, recalibration yields no score improvement, indicating that the original forecast was already calibrated.

	The component $\DSCp_{i} = \Sp^\mathsf{R} - \Sp_{i}^\mathsf{C}$ compares the scores of two recalibrated forecasts:  $\Sp_{i}^\mathsf{C}$ based on information in $X_{it}$, and $\Sp^\mathsf{R}$ unconditionally without additional information.
	As this ``controls for calibration'', the remaining score improvement in $\DSCp_{i}$ reflects the forecast's information content, that is, its discrimination ability.
	Larger values are therefore desirable.
	The final component $\UNCp = \Sp^{\mathsf{R}}$ is forecast-independent and serves as a normalizing constant.

	In case of a misspecified linear recalibration model, the population score decomposition in \eqref{eqn:PopScoreDecomp} can be understood as a ``linearized'' version, similar to the ``isotonized'' score decomposition of \citet{allen2025assessing}.
	In the misspecified case, $\MCBp_{i}$ quantifies the increase in the score's predictive ability that can be achieved by \emph{linear} recalibration.
	As linearity holds under the desirable case of auto-calibration where $X_{it}^\mathsf{C} = \Gamma(Y_t \mid \mW_{it}) = (1, X_{it}) \bar{\mtheta}_i$, we do not expect substantial departures from linearity for reasonably well-specified forecasts. 
	Indeed, in our empirical applications we find no evidence against the linear specification; the corresponding diagnostic results are reported in Appendix~\ref{sec:AddFigs}.
	Otherwise, non-linear recalibration could be captured in our framework by including higher powers of $X_{it}$ in $\mW_{it}$.

	In practice, the components in \eqref{eqn:PopScoreDecomp} require estimated versions of both the recalibrated and reference forecasts.  
	While estimating the reference forecast as the unconditional functional is straightforward, the estimation of the recalibrated forecast remains debated, with proposals including binning methods \citep{MW_1987}, kernel regression \citep{Pohle2020}, and isotonic regressions \citep{DGJ_2021, GR_2023}.
	Instead, we propose using linear regression, which mitigates the overfitting bias associated with nonparametric methods that can invalidate inference, particularly in settings close to auto-calibration \citep{roelofs2022mitigating}. 
	Moreover, the linear specification is consistent with \citet{MZ_1969} and readily accommodates additional covariates beyond the forecasts themselves.

	For the finite-sample guarantees of Theorem~\ref{thm:PositivComponenets} below, it is crucial that we use M-estimation for the regression parameters using the \emph{same} scoring function that is used for the decomposition\footnote{While \citet{DFZ_CharMest} show that strictly consistent loss functions (and only these) produce consistent parameter estimates, the resulting estimates generally differ in value across loss functions \citep{Elliott2016, Muhlemann2021}, motivating the need to employ the \emph{same} scoring function in \eqref{eqn:PopScoreDecomp} and \eqref{eqn:Mestimator}.}
	\begin{align}
		\label{eqn:Mestimator}
		\widehat{\mtheta}_{iT} = \underset{\mtheta \in \mTheta}{\arg \min} \; \frac{1}{T} \sum_{t=1}^T \myS \big(\mW_{it}^\top \mtheta , Y_t \big),
	\end{align}
	where we assume for simplicity that the arg\,min in \eqref{eqn:Mestimator} is unique.
	Moreover, the empirical reference forecast $\widehat{r}_T$ is the unconditional functional $\Gamma$ of the sample $Y_1, \dots, Y_T$.
	Then, the \emph{empirical counterpart} of the score decomposition in \eqref{eqn:PopScoreDecomp} is given by
	\begin{align}
		\label{eqn:FiniteSampleDecomposition}
		\Ss_{iT}  
		= \MCBs_{iT} - \DSCs_{iT} + \UNCs_{T}
		:= \left( \Ss_{iT} - \Ss^{\mathsf{C}}_{iT} \right) - \left( \Ss^{\mathsf{R}}_T - \Ss^{\mathsf{C}}_{iT} \right) + \Ss^{\mathsf{R}}_T,
	\end{align}  
	where 
	\begin{align*}
		\Ss_{iT}   = \frac{1}{T} \sum_{t=1}^T \myS(X_{it},Y_t), \qquad 
		\Ss^{\mathsf{C}}_{iT} = \frac{1}{T} \sum_{t=1}^T \myS (\mW_{it}^\top \widehat{\mtheta}_{iT} , Y_t ) ,\qquad  
		\Ss^{\mathsf{R}}_T = \frac{1}{T} \sum_{t=1}^T \myS(\widehat{r}_T, Y_t).
	\end{align*}

	The following theorem ensures (strict) positivity of the sampling versions of the decomposition in \eqref{eqn:FiniteSampleDecomposition}.  
	\begin{thm}
		\label{thm:PositivComponenets}  
		Assume that the recalibrated forecasts are obtained by \eqref{eqn:Mestimator} using the same strictly consistent score $\myS$ as used in \eqref{eqn:FiniteSampleDecomposition}.
		Furthermore, let $\mW_{it}$ contain a constant and the forecast $X_{it}$, and let $(\widehat{r}_T,0,\dots)^\top \in \mTheta$ almost surely.
		Then, for the decomposition in \eqref{eqn:FiniteSampleDecomposition}, it holds that:
		\begin{enumerate}[label=(\alph*)]
			\item 
			\label{item:MCBs>0}
			$\MCBs_{iT} \ge 0$, which holds with equality if $X_{it} = \mW_{it}^\top \widehat{\mtheta}_{iT}$ a.s.\ for all $t=1, \dots, T$.
			
			\item 
			\label{item:MCBs=0}
			If the $\arg\,\min$ in \eqref{eqn:Mestimator} is unique, then $\MCBs_{iT} = 0$ a.s.\ implies $X_{it} = \mW_{it}^\top \widehat{\mtheta}_{iT}$ a.s.\ for all $t=1, \dots, T$.
			
			\item 
			\label{item:DSCs>0}
			$\DSCs_{iT} \ge 0$, which holds with equality if $\mW_{it}^\top \widehat{\mtheta}_{iT} =  \widehat{r}_T$ a.s.\ for all $t=1, \dots, T$.
			
			\item 
			\label{item:DSCs=0}
			If the $\arg\,\min$ in \eqref{eqn:Mestimator} and $\widehat{r}_T$ are unique, then $\DSCs_{iT} = 0$ a.s.\ implies $\mW_{it}^\top \widehat{\mtheta}_{iT} = \widehat{r}_T$ a.s.\ for all $t=1, \dots, T$.
		\end{enumerate}
	\end{thm} 
	
	The non-negativity guarantees of Theorem~\ref{thm:PositivComponenets} are essential, as zero miscalibration and discrimination components are routinely interpreted as perfect (linear) calibration and no (linear) discrimination, respectively \citep{DGJ_2021, GR_2023, allen2025assessing}.  
	The proof of Theorem~\ref{thm:PositivComponenets} applies to any recalibration method estimated by global score minimization and which is flexible enough to nest both a constant and the identity line, such as isotonic regression \citep{DGJ_2021, GR_2023}.  
	In contrast, kernel regressions \citep{Pohle2020}, logistic regressions and Beta-CDF models for binary events \citep{S_2017, ranjan2010combining} do not offer such guarantees.  
	Although kernels can represent constants and the identity line, they are not fitted via global score minimization; Beta-CDFs, being strictly increasing, cannot produce constants; and logistic regressions generally cannot reproduce the identity line.

	\subsection{Asymptotic approximations}
	\label{sec:estimation}
	
	We now provide new distributional convergence results for the considered linear score decomposition, which are derived under the following high-level regularity conditions, for which we derive more explicit conditions for mean and quantile forecasts in Section~\ref{sec:LowLevelConditions}.
	
	\begin{assumption}
		\label{ass:Normality}
		The following conditions hold for $i=1,2$:
		\begin{enumerate}[label=(\roman*)]
			\item 
			\label{item:EstOP1}
			The parametric estimators are of rate $\sqrt{T} \big( \widehat{\mtheta}_{iT} - \bar\mtheta_i \big) = \mathcal{O}_\P(1)$ and $\sqrt{T} \big( \widehat{r}_T - \bar r \big) = \mathcal{O}_\P(1)$.
			
			\item
			\label{item:CLT}
			\sloppy 
			The following central limit theorem (CLT) holds,
			\begin{align*}
				T^{-1/2} \mOmega_T^{-1/2} \sum_{t=1}^T \big(a_{1t}, c_{1t}, a_{2t}, c_{2t}, e_{t} \big)^\top \toD \mathcal{N}(0, I_5),
			\end{align*}
			where $a_{it} := \myS(X_{it},Y_t) - \E [ \myS(X_{it},Y_t)]$,
			$c_{it} := \myS( \mW_{it}^\top \bar\mtheta_i,Y_t) -  \E [ \myS(\mW_{it}^\top \bar\mtheta_i, Y_t)]$, 
			$e_{t} := \myS(\bar r,Y_t) - \E[ \myS(\bar r,Y_t)]$
			and where
			\begin{align}
				\label{eqn:OmegaT}
				\mOmega_T = \var \left( T^{-1/2} \sum_{t=1}^T \big(a_{1t}, c_{1t}, a_{2t}, c_{2t}, e_{t} \big)^\top \right).
			\end{align}

			\item 
			\label{item:ContExpScore}
			\sloppy
			The expected score $\E \big[ \myS( \mW_{it}^\top \mtheta, Y_t) \big]$ is continuously differentiable in $\mtheta$ for all $\mtheta \in \operatorname{int}(\mTheta)$.
			
			\item 
			\label{item:SE} 
			The empirical process
			\begin{align}
				\label{eqn:EmpProcessSE}
				\nu_{iT} (\mtheta) := T^{-1/2} \sum_{t=1}^T \Big( \myS( \mW_{it}^\top \mtheta, Y_t) - \E \big[ \myS (\mW_{it}^\top \mtheta , Y_t ) \big] \Big)
			\end{align}
			is stochastically equicontinuous, i.e., for all $\varepsilon > 0$ and $\eta > 0$, there exists $\delta > 0$ such that
			\begin{align}
				\label{eqn:StochasticEquicontinuity}
				\underset{T \to \infty}{\lim \sup} \; \P \left( \sup_{ \{\mtheta, \bm\tau \in \mTheta: \, \Vert \mtheta - \bm\tau \Vert \le \delta\} }    \big| \nu_{iT}(\mtheta) - \nu_{iT}(\bm\tau) \big| \ge \eta \right) < \varepsilon.
			\end{align}
		\end{enumerate}
	\end{assumption}
	
	Assumption~\ref{ass:Normality} provides high-level conditions on the scoring function and the parametric estimators. These conditions can accommodate a wide range of settings, including various estimation methods, target functionals, and forms of temporal dependence. They are verified for the mean and quantile cases in Section~\ref{sec:LowLevelConditions} below, and their generality allows for straightforward extensions to other functionals.
	In a nutshell, conditions \ref{item:EstOP1} and \ref{item:CLT} can be verified based on classical dependence and moment conditions on the forecast-realization pairs.
	Conditions \ref{item:ContExpScore} and \ref{item:SE} ensure applicability to non-smooth scoring functions, such as the check loss, by imposing smoothness on the \emph{expected} score and controlling empirical deviations via stochastic equicontinuity \citep{Andrews1994, wellner2013weak}.
	
	\begin{thm}
		\label{thm:joint_nomal_mcb_dsc}
		Let Assumptions~\ref{ass:TrueLinearSpecification} and \ref{ass:Normality} hold and assume that $\MCBp_{i}, \DSCp_{i} > 0$ for $i=1,2$.
		Then, as $T \rightarrow \infty$,
		\begin{align*}	
			\sqrt{T} \, \big( \mathbf{\Xi} \mathbf{\Omega}_T  \mathbf{\Xi}^\top \big)^{-1/2} 
			\begin{pmatrix}
				{\MCBs}_{1T} - \MCBp_{1} \\
				{\DSCs}_{1T} - \DSCp_{1} \\
				{\MCBs}_{2T} - \MCBp_{2} \\
				{\DSCs}_{2T}- \DSCp_{2}
			\end{pmatrix}  
			\toD \mathcal{N}(0,I_4), 
			\quad \text{with} \quad  
			\mathbf{\Xi} = 
			\begin{pmatrix}
				1 & -1  & 0 &  0  & 0 \\
				0 & -1  & 0 &  0  & 1 \\
				0 &  0  & 1 & -1  & 0  \\
				0 &  0  & 0 & -1  & 1  
			\end{pmatrix}.
		\end{align*} 
	\end{thm} 
	
	The joint asymptotic approximation of Theorem~\ref{thm:joint_nomal_mcb_dsc} is required for the construction of ``comparative'' confidence intervals and hypothesis tests---for instance, tests for equal miscalibration or discriminatory ability across the two forecast sequences.
	Notably, the resulting asymptotic distribution---with covariance matrix $\mOmega_T$ from \eqref{eqn:OmegaT}---is unaffected by the estimation uncertainty of
	$\widehat{\mtheta}_{iT}$ and $\widehat{r}_T$, 
	owing to an orthogonality between these estimators and the sampling error arising from replacing expected scores by their empirical counterparts. This robustness is particularly appealing in the presence of potential misspecification of the linear recalibration regression, as it obviates the need to estimate the more complex covariance structure that typically arises under misspecification.

	Theorem~\ref{thm:joint_nomal_mcb_dsc} however requires strict positivity of the population values $\MCBp_{i}, \DSCp_{i} > 0$, hence excluding the boundary cases $\MCBp_{i} = 0$ and $\DSCp_{i} = 0$, which necessitate a separate treatment. 
	We now derive the asymptotic distribution of $\widehat{\MCB_i}_T$ and $\widehat{\DSC_i}_T$ given their population counterparts are zero, which however requires \emph{smooth} scoring functions.
	For \emph{non-smooth} scoring functions as for quantile forecasts, we provide an alternative treatment based on MZ-type regressions, which we describe after Theorem~\ref{thm:MCBandDSC0}.
	
	\begin{assumption}
		\label{ass:AsymptoticsDSC0}
		Assume that the following conditions hold for $i=1,2$:
		\begin{enumerate}[label=(\roman*)]
			\item  
			\label{item:BahadurMestJoint}
			If $\MCBp_i=0$, the estimator $\widehat{\mtheta}_{iT}$ satisfies the asymptotic expansion
			\begin{align}
				\label{eq:MCB0_mestimators}
				\sqrt{T} \big(\widehat{\mtheta}_{iT} - \bar \mtheta_i \big) 
				= - \left( \frac{1}{T} \sum_{t=1}^T \myS''(\mW_{it}^\top \bar \mtheta_i, Y_t)  \mW_{it} \mW_{it}^\top\right)^{-1}  \left( T^{-1/2} \sum_{t=1}^T \myS'(\mW_{it}^\top \bar \mtheta_i, Y_t) \mW_{it} \right) + o_\P(1),
			\end{align}
			\sloppy
			where $T^{-1/2} \sum_{t=1}^T \myS'(\mW_{it}^\top \bar \mtheta_i, Y_t) \mW_{it}$ satisfies a CLT 
			and $\frac{1}{T} \sum_{t=1}^T \myS''(\mW_{it}^\top \bar \mtheta_i, Y_t)  \mW_{it} \mW_{it}^\top$ is invertible for $T$ large enough and satisfies a law of large numbers (LLN).
			
			Furthermore, if $\DSCp_i=0$ (such that $\bar r = \mW_{it}^\top \bar \mtheta_i$), the estimators $\widehat{r}_{T}$ and $\widehat{\mtheta}_{iT}$ satisfy 
			\begin{align}
				\label{eq:comb_mestimators}
				\sqrt{T}	
				\begin{pmatrix}
					\widehat{r}_{T} - \bar r   \\
					\widehat{\mtheta}_{iT} - \bar\mtheta_{i} 
				\end{pmatrix} 
				= -
				\begin{pmatrix}
					\left( \frac{1}{T} \sum_{t=1}^T \myS''(\bar r, Y_t) \right)^{-1},   & \mathbf{0} \\
					\multicolumn{2}{c}{\left( \frac{1}{T} \sum_{t=1}^T \myS''(\bar r , Y_t)   \mW_{it} \mW_{it}^\top\right)^{-1}}
				\end{pmatrix} 
				T^{-1/2} \sum_{t=1}^T \myS'(\bar r , Y_t) \mW_{it}  + o_\P(1),
			\end{align} 
			where $T^{-1/2} \sum_{t=1}^T \myS'(\bar r , Y_t) \mW_{it}$ satisfies a CLT 
			and the terms $\frac{1}{T} \sum_{t=1}^T \myS''( \bar r, Y_t)$ and $\frac{1}{T} \sum_{t=1}^T \myS''(\bar r , Y_t)  \mW_{it} \mW_{it}^\top$ are invertible for $T$ large enough and satisfy LLNs.
			
			\item
			\label{item:SmoothSDSC}
			The scoring function $\myS$ is three times continuously differentiable in its first argument and for $h_T(\mtheta) := \frac{1}{T} \sum_{t=1}^T \big| \myS'''(\mW_{it}^\top \mtheta, Y_t) \big| \Vert \mW_{it} \Vert^3$, we have that
			\begin{align*}
				T^{-1/2} \sup_{\Vert \mtheta - \bar\mtheta_i \Vert \le \delta}  \big| h_T(\mtheta ) - h_T(\bar{\mtheta}_{i}) \big| = o_\P(1) 
				\qquad \text{and} \qquad
				T^{-1/2} h_T(\bar{\mtheta}_{i}) = o_\P(1).
			\end{align*}
		\end{enumerate}
	\end{assumption}    
	
	As before, the high-level conditions in Assumption~\ref{ass:AsymptoticsDSC0} allow for case-by-case verifications for specific functionals, scoring functions, and estimators as exemplified in Section~\ref{sec:LowLevelConditions}.
	Condition~\ref{item:BahadurMestJoint} provides misspecification-robust asymptotic linear representations of the estimators, slightly strengthening Assumption~\ref{ass:Normality}~\ref{item:EstOP1}.
	While auto-calibration \eqref{eqn:AutoCalibration} implies (linearized) $\MCBp_{i} = 0$, the reverse is not true as there might exist a non-linear recalibration function that improves the score.
	This necessitates a misspecification-robust formulation of the terms in \eqref{eq:MCB0_mestimators}--\eqref{eq:comb_mestimators}.
	Condition~\ref{item:SmoothSDSC} applies to smooth loss functions and imposes a uniform LLN type requirement, although this condition is substantially weakened by the $T^{-1/2}$ scaling factor.

	\begin{thm}
		\label{thm:MCBandDSC0}
		Suppose that Assumptions~\ref{ass:TrueLinearSpecification} and \ref{ass:AsymptoticsDSC0} hold.
		Then, as $T \rightarrow \infty$,
		\begin{itemize}
			\item 
			if $\MCBp_{i} = 0$, then $T \, \widehat{\MCB_i}_T
			\toD \frac{1}{2} \bm N_{i}^\top \bm\Upsilon_{i}^{-1} \bm N_{i}$; and
			
			\item 
			if $\DSCp_{i} = 0$, then $T \, \widehat{\DSC}_{iT}
			\toD \frac{1}{2} \bm N_{i}^\top \bm H_{i}^{-1} \bm N_{i}$,
		\end{itemize}
		where 
		$\bm\Upsilon_{i} := \mathbb{E} \left[ \myS''(\mW_{it}^\top \bar{\mtheta}_i, Y_t)  \mW_{it} \mW_{it}^\top \right]$,
		$\bm H^{-1}_{i} := \left(  \mathbb{E} \big[ \myS''(\bar{r}, Y_t)\mW_{it}\mW_{it}^\top \big]^{-1}  - 
		\begin{pmatrix}
			\mathbb{E} [\myS''(\bar{r}, Y_t)] ^{-1} & 0 \\ 0 & 0
		\end{pmatrix} \right)$,
		and 
		$\bm N_{i} \sim \mathcal{N}(0, \bm\Pi_{i})$
		with
		$\bm\Pi_{i} := \lim_{T \to \infty} \var \left(T^{-1/2} \sum_{t=1}^T \myS'(\mW_{it}^\top \bar{\mtheta}_i, Y_t) \mW_{it} \right)$.
	\end{thm}

	Theorem~\ref{thm:MCBandDSC0} shows for \emph{smooth} scoring functions that when the population quantities  $\MCBp_{i}$ or $\DSCp_{i}$ attain their boundary value of zero, the non-negative test statistics $T\,\widehat{\MCB}_{iT}$ and $T\,\widehat{\DSC}_{iT}$ (see  Theorem~\ref{thm:PositivComponenets}) converge to generalized $\chi^2$-distributions.
	
	We treat the boundary cases differently for \emph{non-smooth} scoring functions, such as those used for quantile forecasts. 
	As the proof of Theorem~\ref{thm:MCBandDSC0} illustrates that $\MCBp_i=0$ implies $X_{it} = \mW_{it}^\top \bar \mtheta_i$, i.e., $\mtheta_i = (0,1,0,\dots)^\top$, we test the case $\MCBp_i=0$ through the quantile-specific MZ-type test of \citet{Gaglianone2011} using misspecification-robust covariance estimation. 
	We adapt this test to the case $\DSCp_i=0$ (which implies $\bar r = \mW_{it}^\top \bar \mtheta_i$, i.e., $\mtheta_i = (\bar r,0,0,\dots)^\top$) by testing whether all slope coefficients in the MZ regression are equal to zero. We pursue this alternative approach as extending Theorem~\ref{thm:MCBandDSC0} to \emph{non-smooth} scoring functions via the proof strategy of Theorem~\ref{thm:joint_nomal_mcb_dsc} would necessitate establishing stochastic equicontinuity of $\sqrt{T} \, \nu_{iT} (\mtheta) = \sum_{t=1}^T \myS( \mW_{it}^\top \mtheta, Y_t) - \E \big[ \myS (\mW_{it}^\top \mtheta , Y_t ) \big]$. 
	To the best of our knowledge, no established method exists to derive stochastic equicontinuity  of the $\sqrt{T}$-scaled version of $\nu_{iT}(\mtheta)$ in this setting.

	Recall that the asymptotic variance in the Gaussian case of  Theorem~\ref{thm:joint_nomal_mcb_dsc} is unaffected by the estimation error in  $\widehat{r}_T$ and $\widehat{\mtheta}_{iT}$; it depends solely on the sampling uncertainty from approximating expected scores by empirical averages, as  captured by $\mOmega_T$ in~\eqref{eqn:OmegaT}. 
	In contrast, the limiting distributions in Theorem~\ref{thm:MCBandDSC0} depend  exclusively on the estimation effects, as summarized by  $\bm\Upsilon_{i}$, $\bm H_{i}$, and $\bm\Pi_{i}$, which correspond to the  population counterparts of the asymptotic expansion in~\eqref{eq:comb_mestimators}. 
	Technically, this distinction arises because the asymptotic distribution in  Theorem~\ref{thm:joint_nomal_mcb_dsc} is driven by \emph{first-order} Taylor  expansion terms, whereas the distributions in Theorem~\ref{thm:MCBandDSC0} are  governed by \emph{second-order} terms; see the proofs for details.
	
	Feasible inference requires consistent estimators of the misspecification-robust population matrices  $\mOmega_T$, $\bm{\Pi}_{i}$, $\bm\Upsilon_{i}^{-1}$, and $\bm{H}_{i}^{-1}$ so that the asymptotic results in  Theorems~\ref{thm:joint_nomal_mcb_dsc} and \ref{thm:MCBandDSC0} remain valid when replacing population quantities by their estimates.
	While $\bm\Upsilon_{i}^{-1}$ and $\bm H^{-1}_{i}$ admit consistent estimation via sample averages, we use HAC estimators for the long-run covariance matrices $\mOmega_T$ and $\bm{\Pi}_{i}$ to accommodate serial dependence \citep{NW_1987,A_1991}.  
	Consistency follows by standard arguments, see e.g., \citet[Chapter~6.4]{W_2001} for smooth and \citet{GY_2024} for non-smooth scoring functions.

	\subsection{Testable hypotheses}
	\label{sec:Testing}
	
	For ``one-sample'' hypotheses of the form $\MCBp_i = c$ or $\DSCp_i = c$ for  some $c \ge 0$, the appropriate asymptotic distribution---Gaussian or  generalized $\chi^2$---can be selected directly based on the value of $c$. 
	Under ``two-sample'' hypotheses such as $\MCBp_1 = \MCBp_2$ or 
	$\DSCp_1 = \DSCp_2$, however, the \emph{population} values 
	$\MCBp_{i}$ and $\DSCp_{i}$ may be either zero or strictly positive, so it is \emph{a priori} unclear whether the Gaussian or the generalized $\chi^2$ approximation applies. 
	A valid testing procedure must therefore remain robust to both  possibilities. We achieve this by employing a hybrid of intersection--union and union--intersection testing principles \citep{Berger1997likelihood}.
	
	For example, consider the hypothesis $\H_\MCB :  \MCBp_{1} = \MCBp_{2}$, which we can rewrite as
	\begin{align}
		\H_\MCB  
		& = \left\{ \MCBp_{1} = \MCBp_{2} \right\}   \nonumber  \\
		&= \left\{ \MCBp_{1} = \MCBp_{2} > 0 \right\} \cup \left\{ \MCBp_{1} = \MCBp_{2} = 0 \right\} \nonumber  \\
		&= \left\{ \MCBp_{1} = \MCBp_{2} > 0 \right\} \cup \left( \left\{ \MCBp_{1} = 0 \right\} \cap \left\{ \MCBp_{2} = 0 \right\} \right)   \nonumber  \\
		%&=:  \H^+ \cup \big( \H^0_{1} \cap  \H^0_{2} \big), \\
		&=:  \H^+_{\MCB} \cup \big( \H^0_{{\MCB}_1} \cap  \H^0_{{\MCB}_2} \big).
		\label{eqn:HypothesisCombination}
	\end{align}
	For the hypothesis $\H^+_{\MCB}$, we employ the \emph{selection vector} $\momega = (1,0,-1,0)^\top$ such that Theorem~\ref{thm:joint_nomal_mcb_dsc} yields
	\begin{align}
		\label{eqn:mcbnormal}
		\left( \momega^\top  \mathbf{\Xi} \mathbf{\Omega}_T  \mathbf{\Xi}^\top \momega \right)^{-1/2}
		\sqrt{T}  \left( \big( \MCBs_{1T} - \MCBs_{2T} \big) - \big( \MCBp_{1} - \MCBp_{2} \big) \right)
		\toD \mathcal{N}(0, 1),
	\end{align}
	which immediately gives an asymptotically valid $p$-value called $p^{+}_{\MCB}$.
	
	We propose tests for the hypothesis $\H^0_{{\MCB}_i}$ for both smooth and non-smooth scoring functions.
	In the smooth case, we use the test statistic $T \, \MCBs_{iT}$ and the asymptotic approximation from Theorem~\ref{thm:MCBandDSC0}.
	For a realized value of $T \, \MCBs_{iT}$, we compute the associated $p$-values $p^0_{{\MCB_i}} = \P \bigl( 0.5 \bm N_{i}^\top \bm\Upsilon_{i}^{-1} \bm N_{i} > T \, \MCBs_{iT} \bigr)$ by using the method of \citet{I_1961}, which uses an inversion formula to express the probability as a definite integral of the characteristic function that is solved via numerical integration by using the \texttt{R} package \texttt{CompQuadForm} \citep{CompQuadFormPackage}.
	For non-smooth scoring functions as the GPL score from \eqref{eqn:GPLLoss}, we test the hypothesis $\H^0_{{\MCB}_i}$ by using the MZ-type ``VQR'' test of \citet{Gaglianone2011} and denote its $p$-value by $p^0_{{\MCB_i}}$.
	
	Using \eqref{eqn:HypothesisCombination}, an asymptotically valid $p$-value for the combined null hypothesis $\H_{\MCB}$ is
	\begin{align}
		\label{eq:pvalcomb}
		&p_{\MCB} := \min \left(1, \max \left( p^{+}_{\MCB} , \, 2\min\left(p^0_{{\MCB_1}}, p^0_{{\MCB_2}}\right) \right) \right). 
	\end{align}
	
	Similarly, we obtain an asymptotically valid $p$-value, $p_{\DSC}$, for testing equal discrimination,  
	\begin{align}
		\label{eqn:DSCHypothesisCombination}
		\H_\DSC 
		= \left\{ \DSCp_{1} = \DSCp_{2} \right\}  
		= \left\{ \DSCp_{1} = \DSCp_{2} > 0 \right\} \cup \left( \left\{ \DSCp_{1} = 0 \right\} \cap \left\{ \DSCp_{2} = 0 \right\} \right),
		% = \H^+_{\DSC} \cup \big( \H^0_{{\DSC}_1} \cap  \H^0_{{\DSC}_2} \big). 
	\end{align}
	%$\H_\DSC = \left\{ \DSCp_{1} = \DSCp_{2} \right\}$, 
	by adapting  \eqref{eqn:HypothesisCombination} and \eqref{eq:pvalcomb}. 
	The $p$-value from the Gaussian approximation follows from \eqref{eqn:mcbnormal} by setting  $\momega = (0,1,0,-1)^\top$, while for smooth scoring functions, the corresponding generalized $\chi^2$-probability is computed analogously using the method of \citet{I_1961}.
	For non-smooth scoring functions, we test $\DSCp_i = 0$ through the MZ-type regression of the VQR test, but by simply testing the hypothesis of zero slopes. 
	
	A classical DM test for equal expected scores, $\H_\myS = \left\{ \Sp_{1} = \Sp_{2} \right\}$,  arises as a special case of Theorem \ref{thm:joint_nomal_mcb_dsc} for $\momega = (1,-1,-1,1)^\top$.
	Confidence intervals for the individual decomposition terms as well as their differences can be obtained by inverting the above tests.

	\subsection{Low-level inference conditions for mean and quantile forecasts}
	\label{sec:LowLevelConditions}
	
	We begin with mean forecasts, where $\Gamma$ corresponds to the expectation. 
	In this case, all strictly consistent scoring functions belong to the Bregman class given in \eqref{eqn:BregmanLoss}.
	For this class, we verify the high-level Assumptions~\ref{ass:Normality} and \ref{ass:AsymptoticsDSC0}---and hence the validity of Theorems~\ref{thm:joint_nomal_mcb_dsc} and \ref{thm:MCBandDSC0}---in a general time-series setting with potential model misspecification.
	
	\begin{assumption}
		\label{ass:Philowlevel}
		The following holds for $i=1,2$:
		\begin{enumerate}[label=(\roman*)]
			\item 
			\label{ass:itemlowlevePhicont}
			The function $\phi(z)$ in \eqref{eqn:BregmanLoss} is five times continuously differentiable  with $\phi''(z)>0$ for all $z \in \A=\O$.
			
			\item 
			\label{ass:itemlowlevePhimixing}
			$(Y_t, \mW^\top_{1t}, \mW^\top_{2t})_{t\in\N}$ is strong mixing of size $-(s-\tilde{k})/(\tilde{k}s)$ for some $s > \tilde k > k \ge 2$, and where $k$ is the dimensionality of $\mTheta$ from Assumption~\ref{ass:TrueLinearSpecification}.
			%$-\size/(\size-2)$, $\size>2$; see \citet[Definition 3.45]{W_2001}.
			
			\item  
			\label{ass:itemlowlevePhihessian}
			The matrix $\E \big[ \mW_{it}\mW_{it}^\top \big]$ is positive definite.

			\item 
			\label{ass:itemlowlevePhiCLTcramer}
			The Hessian matrix $\E \left[ \phi''(\mW_{it}^\top\bar\mtheta_i)\mW_{it}\mW_{it}^\top + \phi'''(\mW_{it}^\top\bar\mtheta_i)\mW_{it}\mW_{it}^\top(\mW_{it}^\top\bar\mtheta_i - Y_t)  \right]$
			is invertible and 
			the matrix $\mathbf{V}_T := \var \left( T^{-1/2} \sum_{t=1}^{T} \phi''(\mW_{it}^\top\bar\mtheta_i) \mW_{it}(\mW_{it}^\top\bar\mtheta_i - Y_t) \right)$ is uniformly positive definite for $T$ large enough.
			
			\item 
			\label{ass:itemlowlevePhiadditionalmoms} 
			The following moments exist and are finite for all $\mtheta \in \mTheta$, $t \in \N$, and $\size>2$ from condition~\ref{ass:itemlowlevePhimixing}:
			$\E \left[ \left\vert \phi(Y_{t}) \right\vert^{2\size} \right]$;\,
			$\E\left[ \left\vert \phi(\mW_{it}^\top\mtheta) \right\vert^{2\size} \right]$;\,
			$\mathbb{E} \Big[\left\vert\phi'\left(\mW_{it}^\top \mtheta\right)  \left(Y_t- \mW_{it}^\top \mtheta \right)\right\vert^s\Big]$; \\
			$\mathbb{E}\Big[ \sup\limits_{\mtheta\in\mTheta} \left\| \phi''\left(\mW_{it}^\top\mtheta\right)\left( Y_t - \mW_{it}^\top \mtheta \right) \mW_{it} \right\|^\size \Big]$; 
			$\mathbb{E}\Big[ \sup\limits_{\mtheta\in\mTheta} \left\|  \phi'''(\mW_{it}^\top\mtheta)(\mW_{it}^\top\mtheta - Y_t)   \mW_{it}\mW_{it}^\top \right\| \Big]$; \\
			$\E\Big[ \sup\limits_{\mtheta\in\mTheta}  \left| \phi'''(\mW_{it}^\top\mtheta) \right|  \|\mW_{it}\|^3 \Big]$;
			$\E\Big[ \sup\limits_{\mtheta\in\mTheta}  \left| \phi''''(\mW_{it}^\top\mtheta) \right| \left| \mW_{it}^\top\mtheta - Y_t \right|  \|\mW_{it}\|^3 \Big]$; \\
			$\E\Big[ \sup\limits_{\mtheta\in\mTheta}  \left| \phi''''(\mW_{it}^\top\mtheta) \right|  \|\mW_{it}\|^4 \Big]$; 
			and 
			$\E\Big[\sup\limits_{\mtheta \in \mTheta} \left\vert  \phi'''''\left(\mW_{it}^\top\mtheta\right)  \left(\mW_{it}^\top\mtheta - Y_t\right) \|\mW_{it}\|^4   \right\vert \Big]$.
		\end{enumerate}
	\end{assumption}

	Assumption \ref{ass:Philowlevel} resembles standard time series conditions for mean regression.  
	Item~\ref{ass:itemlowlevePhicont} is a smoothness and convexity condition on the scoring function.
	For the SE scoring function, we often use $\A=\O =\R$, whereas the QLIKE score requires positive values such as $\A=\O = (0, \infty)$.
	Item~\ref{ass:itemlowlevePhimixing} is a classical mixing condition that ensures that suitable LLNs, CLTs and stochastic equicontinuity (using \citet{hansen1996stochastic}) apply.
	Items~\ref{ass:itemlowlevePhihessian} and \ref{ass:itemlowlevePhiCLTcramer} are standard in M-estimation \citep[Chapter 3.1]{NM_1994}. 
	The invertibility of the Hessian is required for the asymptotic theory under misspecification and is naturally guaranteed for the squared error loss or under correct model specification. 
	In either case, the latter term involving $\phi'''$ vanishes.
	Item~\ref{ass:itemlowlevePhiadditionalmoms} collects moment conditions, which simplify considerably for the common choices $\phi_{\textsf{SE}}(z) = z^2$ and $\phi_{\textsf{QL}}(z) = -\log(z)$
	of the SE and QLIKE scores in \eqref{eqn:SEscore} and \eqref{eqn:QLIKEscore}.
	
	\begin{assumption}
		\label{ass:itemlowlevePhicltvar}
		The matrix $\mOmega_T$ from \eqref{eqn:OmegaT} is uniformly positive definite for $T$ large enough. 
	\end{assumption}

	Positive definiteness of $\mOmega_T$ is stated in the separate Assumption~\ref{ass:itemlowlevePhicltvar} as it is required for the asymptotic normality result in Theorem~\ref{thm:joint_nomal_mcb_dsc}, but is violated in the cases $\MCBp_i = 0$ and $\DSCp_i =0$, and is as such not required for the verification of the conditions of Theorem~\ref{thm:MCBandDSC0}.
	The positive definiteness of $\mOmega_T$ essentially imposes non-collinearity conditions for the (sum over time of the) original, recalibrated, and reference scores.
	Violations arise, for instance, if the recalibrated score is a constant multiple of the original score for all $t \in \mathbb{N}$, such as if $\MCBp_i=0$.

	\begin{prop}
		\label{prop:verification_phi} 
		Consider a Bregman score from \eqref{eqn:BregmanLoss} for the decomposition, which is also used in \eqref{eqn:Mestimator}. Then:
		\begin{itemize}
			\item 
			Assumptions~\ref{ass:TrueLinearSpecification}, \ref{ass:Philowlevel}, and \ref{ass:itemlowlevePhicltvar} imply Assumption \ref{ass:Normality} such that Theorem \ref{thm:joint_nomal_mcb_dsc} holds; and
			
			\item  
			Under $\MCBp_{i} = 0$ or $\DSCp_{i} = 0$, 
			Assumptions~\ref{ass:TrueLinearSpecification},  \ref{ass:Philowlevel} and  imply Assumption \ref{ass:AsymptoticsDSC0} such that Theorem~\ref{thm:MCBandDSC0} holds.
		\end{itemize}
	\end{prop}

	We continue with the case where $\Gamma$ is the $\alpha$-quantile for some $\alpha \in (0,1)$, where all strictly consistent loss functions are given by the GPL class in \eqref{eqn:GPLLoss}.
	For a linear GPL score decomposition, we employ (generalized) quantile regression by combining the estimator \eqref{eqn:Mestimator} with the loss in \eqref{eqn:GPLLoss} and impose the following conditions.
	
	\begin{assumption}
		\label{ass:QuantilesLowLevel}
		The following holds for $i=1,2$:
		\begin{enumerate}[label=(\roman*)]
			\item The function $g$ from \eqref{eqn:GPLLoss} is twice continuously differentiable with $0 < g'(z) \le C < \infty$ for all $z \in \A=\O$.
			
			\item 
			\label{cond:QuantCondMixing}
			The sequence $(Y_t, \mW^\top_{1t}, \mW^\top_{2t})_{t \in \N}$ is strong mixing of size $-(s-\tilde{k})/(\tilde{k}s)$ for some  $s > \tilde k > k \ge 2$, and where $k$ is the dimensionality of $\mTheta$ from Assumption~\ref{ass:TrueLinearSpecification}.

			\item 
			\label{cond:QuantCondContDistribution}
			For all $t \in \mathbb{N}$, 
			$F_{it}$, the conditional distribution of $Y_t$ given $\mathcal{W}_{it}$, belongs to a class of distributions that possess a positive and continuous Lebesgue density $f_{it}$ with $0 < f_{it}(x) \le K < \infty$, uniformly on the support of $F_{it}$.
			
			\item 
			The matrix $\E \left[ \mW_{it} \mW_{it}^\top \right]$ is positive definite.
			
			\item 
			The matrix  $\mathbf{V}_T := \var \left( T^{-1/2} \sum_{t=1}^{T} g'(\mW_{it}^\top\bar\mtheta_i) 
			\big( \mathds{1}\{Y_t \le \mW_{it}^\top \bar\mtheta_i \} - \alpha \big) \mW_{it} \right)$  is uniformly positive definite for $T$ large enough. 
			
			\item 
			$\E \Big[ (\mW_{it} \mW_{it}^\top) \left\{ g''(\mW_{it}^\top \bar\mtheta_i) \big( F_{it}(\mW_{it}^\top \bar\mtheta_i) - \alpha \big) + g'(\mW_{it}^\top \bar\mtheta_i) f_{it}(\mW_{it}^\top \bar\mtheta_i)\right\} \Big]$ is invertible.
			
			\item  
			$\E|Y_t| < \infty$, $\E|g(Y_t)| < \infty$,
			$\E \big\Vert \mW_{it} \big\Vert^s < \infty$,
			$\E \big| g(\mW_{it}^\top\mtheta) - g(Y_t) \big|^s < \infty$, and
			there exist integrable functions $G$ and $G'$, $G_2'$, $G''$ such that
			$\Vert \mW_{it} g(\mW_{it}^\top \mtheta) \Vert \le G(\mW_{it})$,
			$\Vert \mW_{it} g'(\mW_{it}^\top \mtheta) \Vert \le G'(\mW_{it})$,
			$\Vert \mW_{it} \mW_{it}^\top g'(\mW_{it}^\top \mtheta) \Vert \le G_2'(\mW_{it})$, and
			$\Vert \mW_{it} \mW_{it}^\top g''(\mW_{it}^\top \mtheta) \Vert \le G''(\mW_{it})$, where $s$ is from condition \ref{cond:QuantCondMixing}.
			
			\item 
			\label{cond:QuantCondOmegaPD}
			The matrix $\mOmega_T$ from \eqref{eqn:OmegaT} is uniformly positive definite for $T$ large enough.
			
		\end{enumerate}
	\end{assumption}   
	
	The conditions of Assumption~\ref{ass:QuantilesLowLevel} are standard regularity conditions for (misspecified) quantile regression, but also cover the case of estimation through the more general GPL loss functions, similar to \citet{Komunjer2005}.
	The required moments of order $s > k$ are required for the derivation of stochastic equicontinuity for strong mixing processes through \citet{hansen1996stochastic}. 
	As we consider low dimensionalities of $\mTheta \subset \R^k$ such as $k=2$ in our applications, this is not overly restrictive.
	Condition (vi) is required under model misspecification, and holds naturally if either the tick loss is used (implying $g''=0$), or under correct specification (implying $F_{it}(\mW_{it}^\top \bar\mtheta_i) = \alpha$).
	
	Compared to Assumption~\ref{ass:Philowlevel}, the conditions in Assumption~\ref{ass:QuantilesLowLevel} require the conditional distribution to be absolutely continuous, as specified in item~\ref{cond:QuantCondContDistribution}; this is a standard requirement for quantile regression frameworks.
	Furthermore, item~\ref{cond:QuantCondOmegaPD} is included here because the subsequent proposition specifically verifies the conditions for the asymptotic normality result established in Theorem~\ref{thm:joint_nomal_mcb_dsc}, but does not treat the case of Theorem~\ref{thm:MCBandDSC0}.
	
	\begin{prop}
		\label{prop:Quantiles_Verification} 
		Consider a GPL score from \eqref{eqn:GPLLoss} for the decomposition, that is also used in \eqref{eqn:Mestimator}. Then, Assumptions~\ref{ass:TrueLinearSpecification} and \ref{ass:QuantilesLowLevel} imply Assumption \ref{ass:Normality} such that Theorem \ref{thm:joint_nomal_mcb_dsc} holds.
	\end{prop}
	
	Proposition~\ref{prop:Quantiles_Verification} verifies the high-level conditions of Theorem~\ref{thm:joint_nomal_mcb_dsc} through standard conditions from (time series) quantile regression. 
	As discussed above, Theorem~\ref{thm:MCBandDSC0} only applies to smooth loss functions, and hence, the cases $\MCBp_{i} = 0$ and $\DSCp_{i} = 0$ are captured through the quantile-specific MZ-test of \citet{Gaglianone2011}.

	\subsection{Relation to Value-at-Risk backtests}
	\label{sec:VaRBacktestsTheory}
	
	VaR backtesting is a fundamental pillar of financial risk management, used to assess whether a model accurately predicts the frequency and magnitude of potential losses. 
	Traditional backtests such as by \citet{K_1995} and the traffic light system of \citet{B_1996, B_2019} typically focus on unconditional counts for the ``violations'' or ``hits'', where realized losses exceed the VaR predictions as captured through the quantile identification functions, $\myV_{it} = \myV(X_{it}, Y_t) = \mathds{1}\{Y_t \le X_{it}\} - \alpha$, from \eqref{eqn:CondCalibrationIDFunction}. 
	These unconditional tests, however, remain insensitive to the clustering of violations---a phenomenon particularly critical during financial crises, when consecutive days of unpredicted losses can lead to rapid capital depletion.
	
	\begin{table}[tb]
		\caption{Summary of prominent VaR backtests, including the underlying hypothesis, methodology and relation to conditional calibration in \eqref{eqn:CondCalibrationIDFunction} through $\mathcal{G}_{it} = \sigma(\bm{G}_{it})$.
			We use the shorthand $\myV_{it-l}:=\myV(X_{it-l}, Y_{t-l})$ for $l \ge 0$.
		}
		\label{tab:BacktestTheory}
		\centering
		\resizebox{\textwidth}{!}{%
			\begin{tabular}{lllll}
				\hline
				Reference & Label &
				Covariates/Instruments $\bm{G}_{it}$ &
				Method & Hypothesis $\H_0$ \\
				\hline
				\addlinespace
				
				\multicolumn{3}{l}{\underline{Two-sided backtests}} & \\  
				
				\citet{K_1995} & UC & $\bm{G}_{it}=1$ & $\E\left[\myV(X_{it}, Y_t) \mid \bm{G}_{it} \right]= \bm{G}_{it}^\top\bar{\bm{\beta}}_i$ & $\bar{\bm{\beta}}_i = \bm{0}$\\
				
				\citet{C_1998} & CC & $\bm{G}_{it}=\left(1,\myV_{it-1}\right)^\top$  & $\E\left[\myV(X_{it}, Y_t) \mid \bm{G}_{it} \right]= \bm{G}_{it}^\top\bar{\bm{\beta}}_i$ & $\bar{\bm{\beta}}_i = \bm{0}$ \\
				
				\citet{EM_2004} & DQ &
				$\bm{G}_{it}=\left(1, \myV_{it-1}, \dots, \myV_{it-4}\right)^\top$ &  $\E\left[\myV(X_{it}, Y_t) \mid \bm{G}_{it} \right]= \bm{G}_{it}^\top\bar{\bm{\beta}}_i$ & $\bar{\bm{\beta}}_i = \bm{0}$\\
				
				\citet{KMP_2006}& DQX &
				$\bm{G}_{it}=\left(1,X_{it},\myV_{it-1}, \dots, \myV_{it-4}\right)^\top$ &   $\E\left[\myV(X_{it}, Y_t) \mid \bm{G}_{it} \right]= \bm{G}_{it}^\top\bar{\bm{\beta}}_i$ & $\bar{\bm{\beta}}_i = \bm{0}$ \\  
				
				\citet{Gaglianone2011} & VQR & $\bm{G}_{it}=\left(1,X_{it}\right)^\top$ &
				$Q_\alpha(Y_t \mid \bm{G}_{it}) = \bm{G}_{it}^\top \bar{\bm{\beta}}_i$ & $ \bar{\bm{\beta}}_i = (0,1)^\top$ \\ \addlinespace
				
				\multicolumn{3}{l}{\underline{One-sided backtests}} & \\
				\citet{B_1996} & Basel & $\bm{G}_{it} = 1$ &  Binomial CDF & $\E[\myV_{it}]\leq0$\\ 
				\citet{NZ_2017} & NZ & $\bm{G}_{it}=(1,X_{it})^\top$  & Moment Condition & $\E[\bm{G}_{it}^\top\myV_{it}]\leq0$\\[0.1cm]
				\hline
			\end{tabular}
		}
	\end{table}
	
	Following the seminal work of \citet{C_1998}, a plethora of \emph{conditional} backtests has been proposed over the past decades. 
	See e.g., \citet{campbell2007review} and \citet{nieto2016frontiers} for comprehensive reviews on VaR backtesting.
	Most of these methods test hypotheses closely related to the conditional calibration property in \eqref{eqn:CondCalibrationIDFunction} by utilizing varying information sets $\mathcal{G}_{it} = \sigma(\bm{G}_{it})$.
	Table~\ref{tab:BacktestTheory} provides a non-exhaustive list of prominent VaR backtests, illustrating the intrinsic connection between their choice of covariates (or instruments) $\bm{G}_{it}$ and our recalibration variables $\mW_{it}$. 
	Notably, the equivalence between \eqref{eqn:ConditionalCalibration} and \eqref{eqn:CondCalibrationIDFunction} permits specifications based on either quantile regressions for $Y_t$, or mean regressions for the respective identification function.
	
	As a canonical case of assessing auto-calibration, the quantile-specific adaptation of the MZ backtest by \citet{Gaglianone2011} utilizes the same recalibration technique as our linear score decomposition when $\mW_{it}^\top = (1,X_{it})$.
	By matching the recalibration variables $\mW_{it}$ with the covariates $\bm{G}_{it}$, both the score decompositions and the underlying inference can be tailored to the specific notion of conditional calibration underlying the other backtests presented in Table~\ref{tab:BacktestTheory}.

	\section{Simulations}
	\label{sec:sim}
	
	We now assess the finite-sample performance of the proposed tests from Section~\ref{sec:Testing}. 
	Section~\ref{sec:sim_para} describes the simulation setup and Section~\ref{sec:sim_results} reports the results.

	\subsection{Simulation setup}
	\label{sec:sim_para}
	
	To empirically validate our tests for equal miscalibration, with null hypothesis $\H_{\MCB}: \MCBp_{1} = \MCBp_{2}$, or for equal discrimination, with $\H_{\DSC}: \DSCp_{1} = \DSCp_{2}$, between competing mean or quantile forecasts, we use the following data-generating process (DGP) for the realizations,
	\begin{align}
		\label{eqn:SimProcess}
		Y_t =  \mathbf{V}_t^\top \bm{\gamma} + \varepsilon_{Y,t}, \quad\text{with}\quad \mathbf{V}_t=(K_t,L_t,M_t,N_t)^\top  \quad\text{and}\quad \bm{\gamma}=(\gamma_K,\gamma_L,\gamma_M,\gamma_N)^\top,
	\end{align} 
	with $\varepsilon_{Y,t} \stackrel{\text{iid}}{\sim} \mathcal{N}(0,1)$.
	The predictor variables in $\mathbf{V}_t$ follow independent autoregressive AR(1) processes $K_t = \beta K_{t-1} + \varepsilon_{K,t},\  L_t = \beta L_{t-1} + \varepsilon_{L,t},\ M_t = \beta M_{t-1} + \varepsilon_{M,t},\ N_t = \beta N_{t-1} + \varepsilon_{N,t}$, with fixed $\beta = 0.25$ and $ \varepsilon_{K,t}, \varepsilon_{L,t}, \varepsilon_{M,t}, \varepsilon_{N,t} \overset{\text{iid}}{\sim} \mathcal{N}(0,1)$, which are all mutually independent.
	
	Forecaster $X_{1t}$ has exclusive access to $K_t$ and forecaster $X_{2t}$ to $L_t$, while both forecasters share information on $M_t$ and $N_t$ as formalized by $\mathcal{I}_{1t} = \sigma \left\{K_{t}, M_{t}, N_{t} \right\}$ and $\mathcal{I}_{2t} = \sigma \left\{ L_{t}, M_{t}, N_{t} \right\}$.
	Notice that our setup requires both ``joint'' predictors $M_{t}$ and $N_{t}$ in order to construct non-collinear forecasts with equal $\MCBp$ and/or $\DSCp$ values.

	Based on the distinct information sets $\mathcal{I}_{1t}$ and $\mathcal{I}_{2t}$, the forecasters issue \emph{mean} forecasts
	\begin{align}
		\label{eqn:SimFcast1}
		X_{1t} &=  \delta_0  + \mathbf{V}_t^\top \bm{\delta}, \qquad\text{with}\qquad \bm{\delta}=(\delta_K,0,\delta_M,\delta_N)^\top \in \mathbb{R}^4, \\
		\label{eqn:SimFcast2}
		X_{2t} &=  \xi_0  + \mathbf{V}_t^\top \bm{\xi}, \qquad\text{with}\qquad \bm{\xi}=(0,\xi_L,\xi_M,\xi_N)^\top \in \mathbb{R}^4, 
	\end{align}
	and $\delta_0, \xi_0 \in \mathbb{R}$.
	The imposed zeros in $\bm{\delta}$ and $\bm{\xi}$ reflect the absence of the corresponding predictor variables in the information sets of the respective forecaster.
	The coefficients $(\delta_0, \bm{\delta}^\top)$ and $(\xi_0, \bm{\xi}^\top)$ and their relation to $\bm{\gamma}$ allows to generate forecasts with different calibration and discrimination properties, as formalized in the following proposition for the squared error score from \eqref{eqn:SEscore}.
	
	\begin{prop}
		\label{prop:sim}
		For realizations $Y_t$ and forecasts $X_{1t}$ and $X_{2t}$ following \eqref{eqn:SimProcess}--\eqref{eqn:SimFcast2} with $\bm{\delta} \not= 0$ and $\bm{\xi} \not= 0$, and given the recalibration information $\mathcal{W}_{it} = \sigma\{\mW_{it}\}$, with $\mW_{it}=(1,X_{it})^\top$, $i=1,2$, 
		the population squared error score decomposition of the forecasts $X_{1t}$ is
		\begin{align}
			\label{eqn:MCBDSCSimFormula}
			\MCBp_1= \delta_0^2 + \varsigma \cdot \frac{(\bm{\delta}^\top \bm{\delta} - \bm{\delta}^\top \bm{\gamma})^2}{\bm{\delta}^\top \bm{\delta}}, \qquad 
			\DSCp_1= \varsigma \cdot \frac{(\bm{\delta}^\top \bm{\gamma})^2}{\bm{\delta}^\top \bm{\delta}},  \qquad 
			\UNCp= \varsigma \cdot \bm{\gamma}^\top \bm{\gamma} + 1,
		\end{align}
		where $\varsigma=\frac{1}{1-\beta^2}$ is the variance of the zero-mean AR(1) processes $K_t,L_t,M_t,N_t$. For $\MCBp_2$ and $\DSCp_2$ interchange $\delta_0^2$ and $\bm{\delta}$ with  $\xi_0^2$ and $\bm{\xi}$, respectively. 
	\end{prop}

	The expressions in Proposition~\ref{prop:sim} illustrate that for non-constant forecasts, $\MCBp_1=0$ holds if and only if $\bm{\delta}^\top \bm{\delta} = \bm{\delta}^\top \bm{\gamma}$ and $\delta_0=0$.
	Hence, perfect calibration, $\MCBp_1=0$, can be achieved even if the forecaster omits relevant variables. 
	However, satisfying the condition $\bm{\delta}^\top \bm{\delta} = \bm{\delta}^\top \bm{\gamma}$ generally requires the forecaster to correctly weigh the signals they do include. 
	Absent a coincidental combination of perfectly offsetting coefficient errors, using an incorrect parameter value or incorporating an irrelevant predictor (e.g., $\delta_{M} \neq 0$ when $\gamma_{M}=0$) will violate this equality and cause miscalibration.
	
	The $\DSCp_1$ component is unaffected by the intercept $\delta_0$ that does not carry any information.
	$\DSCp_1$ is positive whenever $\bm{\delta}^\top \bm{\gamma} \neq 0$, that is, whenever $X_{1t}$ reacts to the signal in $Y_t$, even incorrectly such as in the wrong direction.
	Moreover, by Cauchy-Schwarz, $\frac{(\bm{\delta}^\top \bm{\gamma})^2}{\bm{\delta}^\top \bm{\delta}} \leq \bm{\gamma}^\top\bm{\gamma}$, such that the unconstrained theoretical upper bound for discrimination is $\varsigma \boldsymbol{\gamma}^{\top} \boldsymbol{\gamma}$.
	This global bound would be attained if and only if a forecaster could choose parameters $\boldsymbol{\delta}$ proportionally to the full signal $\boldsymbol{\gamma}$, i.e., $\boldsymbol{\delta}=c \boldsymbol{\gamma}$ for some $c \in \mathbb{R}$
	as then, $\frac{(\bm{\delta}^\top \bm{\gamma})^2}{\bm{\delta}^\top \bm{\delta}} = \frac{c^2 (\bm{\gamma}^\top \bm{\gamma})^2}{c^2 \bm{\gamma}^\top \bm{\gamma}} = \bm{\gamma}^\top \bm{\gamma}$.
	Hence, maximal discrimination can arise even for miscalibrated forecasts. 
	If only some relevant variables are observed, the best achievable discrimination is obtained when the available predictors are used correctly.
	Analogous conclusions hold for the second forecaster $X_{2t}$ by replacing $\bm{\delta}$ with $\bm{\xi}$.
	
	\begin{table}[tbp]
		\centering
		\caption{Twelve \emph{mean forecasting} specifications of the parameters $\bm{\gamma}, \delta_0, \bm{\delta}, \xi_0, \bm{\xi}$ in \eqref{eqn:SimProcess}--\eqref{eqn:SimFcast2} based on a flexible $k \in \{0, 0.05, \ldots, 0.45, 0.5\}$.
			All setups guarantee non-collinear forecasts and generate different relations of the $\MCBp$ and $\DSCp$ values as given in the numbered rows and the last column termed ``Implications''.}
		\begin{adjustbox}{width=\textwidth}
			\begin{tabular}{c  c c c c | c c c c c | c c c c c | c } 
				\toprule
				% table head
				& 	\multicolumn{4}{c}{Realization} & \multicolumn{5}{c}{Forecaster $X_{1t}$} & \multicolumn{5}{c}{Forecaster $X_{2t}$} & \\
				
				& $\gamma_K$ & $\gamma_L$ & $\gamma_M$ & $\gamma_N$ & $\delta_0$ & $\delta_K$ & $\delta_L$ & $\delta_M$ & $\delta_N$ & $\xi_0$ & $\xi_K$ & $\xi_L$ & $\xi_M$ & $\xi_N$ & Implications  \\
				\hline  \\ 
				
				\multicolumn{16}{l}{ \bf  \underline{(A) Parameterization of  $\MCBp$ difference}} \\ \\

				% part 1a
				\multicolumn{16}{l}{\bf {1) $\MCBp_{1}(k) = \MCBp_{2}(k) = \frac{8k^2}{15}$}} \\ \addlinespace
				
				& 0 & $\frac{1}{4}$ & $\frac{1}{4}$  & 0  
				& 0 & 0 & 0 & $\frac{k}{\sqrt{2}} + \frac{1}{4} $ & 0 
				& 0 & 0 & $\frac{k}{{2}} + \frac{1}{4} $ &  $\frac{k}{{2}} + \frac{1}{4} $ &0 
				&  $0 < \DSCp_{1} < \DSCp_{2}$ \\
				
				% part 1b
				& $\frac{1}{4}$  & $\frac{1}{4}$ & $\frac{1}{4}$  &  0
				& 0 & $\frac{k}{{2}} + \frac{1}{4}$  & 0 & $\frac{k}{{2}} + \frac{1}{4}$ & 0
				& 0 & 0 & $\frac{k}{{2}} + \frac{1}{4}$ & $\frac{k}{2} + \frac{1}{4}$ & 0 
				& $0 < \DSCp_{1} =\DSCp_{2}$ 	
				\\  \\

				% part 2a
				\multicolumn{16}{l}{ \bf{2) $\MCBp_{1} = 0$ {and} $\MCBp_{2}(k) = \frac{8k^2}{15}$}} \\ \addlinespace
				& 0 & $\frac{1}{4}$ & $\frac{1}{4}$  & 0  
				& 0 & 0  & 0 &$\frac{1}{4}$ & 0 
				& 0 & 0 & $\frac{k}{2} + \frac{1}{4}$ & $\frac{k}{2} + \frac{1}{4}$ & 0 
				&$0 < \DSCp_{1} < \DSCp_{2}$ \\
				
				% part 2b
				& $\frac{1}{4}$   & $\frac{1}{4}$ & $\frac{1}{4}$  & 0
				& 0 & $\frac{1}{4}$ & 0 & $ \frac{1}{4} $ & 0 
				&0 & 0 & $\frac{k}{2} + \frac{1}{4}$  & $\frac{k}{2} + \frac{1}{4}$  & 0 
				& $0 < \DSCp_{1} = \DSCp_{2}$ 
				\\ 	   \\ 
				
				% part 3a
				\multicolumn{16}{l}{\bf{3) ${\MCBp}_{1}(k) = 2 \, {\MCBp}_{2}(k) = \frac{16 k^2}{15}$}} \\ \addlinespace
				
				& 0 & $\frac{1}{4}$ & $\frac{1}{4}$  & 0  
				& 0 & 0 & 0 &$k+\frac{1}{4}$ & 0  
				& 0 & 0 &  $\frac{k}{2} + \frac{1}{4}$ & $\frac{k}{2} + \frac{1}{4}$ & 0 
				& $0 < \DSCp_{1} < \DSCp_{2}$ \\
				
				% part 3b
				& $\frac{1}{4}$  & $\frac{1}{4}$ & $\frac{1}{4}$  & 0
				& 0 & $\frac{k}{\sqrt{2}} + \frac{1}{4}$  &  0 & $\frac{k}{\sqrt{2}} + \frac{1}{4}$ & 0
				& 0 & 0 &  $\frac{k}{{2}} + \frac{1}{4}$ & $\frac{k}{2} + \frac{1}{4}$ & 0 
				& $0 < \DSCp_{1} =\DSCp_{2}$ 	\\     
				
				\\ 
				\midrule
				\\
				
				% DSC o
				\multicolumn{16}{l}{ \bf  \underline{(B) Parameterization of  $\DSCp$ difference}} \\ \\
				
				\multicolumn{16}{l}{\bf {4) $\DSCp_{1}(k) = \DSCp_{2}(k) = \frac{8 k^2}{15}$}} \\ \addlinespace
				
				% part 4a
				& 0 & 0 & 0 & $k$ 
				& 0 & $\frac{k}{2} + \frac{1}{4\sqrt{2} }$ & 0 & 0  & $\frac{k}{2} +    \frac{1}{4\sqrt{2} }$ 
				& 0 & 0  & $\frac{k}{2} + \frac{1}{4}$  & 0 & $\frac{k}{2} + \frac{1}{4}$ 
				& $0 < \MCBp_{1} < \MCBp_{2}$ \\
				
				% part 4b
				& 0 & 0 & 0 & $k$ 
				& 0 & $\frac{k}{2}  + \frac{1}{4}$ & 0 & 0 & $\frac{k}{2} + \frac{1}{4}$ 
				& 0 & 0  & $\frac{k}{2} + \frac{1}{4}$  & 0 & $\frac{k}{2} + \frac{1}{4}$ 
				& $0 < \MCBp_{1} = \MCBp_{2}$ \\ \\

				% part 5a
				\multicolumn{16}{l}{\bf {5) $\DSCp_{1} = 0$ {and} $\DSCp_{2}(k) = \frac{8 k^2}{15}$} } \\ \addlinespace
				
				& 0 & 0 & 0 & $k$ 
				& 0 & 0 & 0 & $\frac{1}{4}$ & 0  
				& 0 & 0  & $\frac{k}{2} + \frac{1}{4}$  & 0 & $\frac{k}{2} + \frac{1}{4}$ 
				& $0 < \MCBp_{1} < \MCBp_{2}$ \\

				% part 5b
				& 0 & 0 & 0 & $k$ 
				& 0  & $ \frac{1}{4} $ & 0  & $\frac{1}{4}$ & 0 
				& 0 & 0  & $\frac{k}{2} + \frac{1}{4}$  & 0 & $\frac{k}{2} + \frac{1}{4}$ 
				& $0 < \MCBp_{1} = \MCBp_{2} $ \\    \\ 
				
				% part 6a
				\multicolumn{16}{l}{\bf {6) ${\DSCp}_{1}(k) = 2 \, {\DSCp}_{2}(k) = \frac{16 k^2}{15}$} } \\ \addlinespace
				
				& 0 & 0 & 0 & $k$ 
				& 0 & 0 & 0 & 0 & $k + \frac{1}{4}$ 
				& 0 & 0  & $\frac{k}{2} + \frac{1}{4}$  & 0 & $\frac{k}{2} + \frac{1}{4}$ 
				& $0 < \MCBp_{1} < \MCBp_{2}$ \\
				
				% part 5b
				& 0 & 0 & 0 & $k$ 
				& $\frac{1}{\sqrt{15}}$ & 0 & 0  & 0 & $k+ \frac{1}{4}$ 
				& 0 & 0  & $\frac{k}{2} + \frac{1}{4}$  & 0 & $\frac{k}{2} + \frac{1}{4}$ 
				& $0 < \MCBp_{1} =\MCBp_{2}$ \\
				\addlinespace
				\bottomrule
			\end{tabular}
		\end{adjustbox}
		\label{tab:main}
	\end{table}

	The closed-form expressions in Proposition~\ref{prop:sim} enable diverse specifications of $\MCBp_i$ and $\DSCp_i$, allowing us to study the size and power of our proposed tests across different settings by appropriately choosing the parameters $\bm{\gamma}$, $(\delta_0, \bm{\delta}^\top)$, and $(\xi_0, \bm{\xi}^\top)$.
	In detail, Table~\ref{tab:main} specifies twelve different parameterizations that result in illustrative settings for $\MCBp_i$ and $\DSCp_i$ as given in the facet labels of Figure~\ref{fig:sim_param_rr}.
	To analyze test size and power, we vary a parameter $k>0$ that affects $\boldsymbol{\gamma}, (\delta_0, \boldsymbol{\delta}^\top)$, and $(\xi_0, \boldsymbol{\xi}^\top)$, which continuously shifts the DGPs either along or away from their respective null hypotheses.
	The specifications of Table~\ref{tab:main} are chosen to guarantee that either $\MCBp_2(k)$ or $\DSCp_2(k)$ depend on $k$ and equal $\frac{8k^2}{15}$ for all $k \ge 0$.
	Figure~\ref{fig:TrueMCBDSC_MeanDPG} in Appendix~\ref{sec:AddFigs} plots the population values $\MCBp_i$ and $\DSCp_i$, $i=1,2$  as functions of $k$ for all twelve setups.
	
	For example, we generate the scenario $\MCBp_{1} = 0$ and $\MCBp_{2}(k) = \tfrac{8k^2}{15}$ while maintaining $0 < \DSCp_1 < \DSCp_2$ independent of $k$ by setting $\bm\gamma^\top = (0, \tfrac{1}{4}, \tfrac{1}{4}, 0)$ as well as  $(\delta_0, \bm\delta^\top) = (0, 0, 0, \tfrac{1}{4}, 0)$ and $(\xi_0, \bm\xi^\top) = (0,0, \tfrac{1}{4} + \tfrac{k}{2}, \tfrac{1}{4} + \tfrac{k}{2}, 0)$.
	Here, the realizations are only driven by $L_t$ and $M_t$ with equal magnitudes, $\gamma_L = \gamma_M$.
	The first forecaster $X_{1t}$ does not have access to $L_t$, but she correctly uses $M_t$ through $\delta_M = \gamma_M$ such that her forecasts are perfectly calibrated, $\MCBp_1 = 0$.
	The sole access to $M_t$ however limits her discrimination to $\DSCp_1 = \tfrac{1}{15}$.
	By contrast, the second forecaster $X_{2t}$ uses both relevant variables $L_t$ and $M_t$, and hence achieves a higher discrimination $\DSCp_2 = \tfrac{2}{15} > \DSCp_1$ for any $k \ge 0$.
	However, for values $k > 0$, she uses the information incorrectly with $\bm{\xi} \not= \bm{\gamma}$ such that $\MCBp_2(k) = \tfrac{8k^2}{15}$ increases with $k$.
	
	\begin{table}[tbp]
		\centering
		\caption{Four \emph{quantile forecasting} specifications of the parameters $\bm{\gamma}, \delta_0, \bm{\delta}, \xi_0, \bm{\xi}$ in \eqref{eqn:SimProcess} and \eqref{eqn:SimFcastq} based on a flexible $k \in \{0, 0.05, \ldots, 0.95, 1\}$.
			All setups guarantee non-collinear forecasts and generate different relations of the $\MCBp$ and $\DSCp$ values as given in the numbered rows and the last column termed ``Implications''.
			The values of $\xi_0(k,\alpha)$ are chosen numerically such that  $\MCBp_{1}(k) = \MCBp_{2}(k)$ holds; see Appendix~\ref{app:QuantileSimDGP} and Figure~\ref{fig:xi0_choice} for details.
		}
		\label{tab:SimQuantiles}
		\begin{adjustbox}{width=\textwidth}
			\begin{tabular}{c  c c c c | c c c c c | c c c c c | c } 
				\toprule
				% table head
				& 	\multicolumn{4}{c}{Realization} & \multicolumn{5}{c}{Forecaster $X_{1t}^\alpha$} & \multicolumn{5}{c}{Forecaster $X_{2t}^\alpha$} & \\
				
				& $\gamma_K$ & $\gamma_L$ & $\gamma_M$ & $\gamma_N$ & $\delta_0$ & $\delta_K$ & $\delta_L$ & $\delta_M$ & $\delta_N$ & $\xi_0$ & $\xi_K$ & $\xi_L$ & $\xi_M$ & $\xi_N$ & Implications  \\
				\hline  \\ 
				
				\multicolumn{16}{l}{ \bf  \underline{(A) Parameterization of  $\MCBp$ difference}} \\ \\
				
				% part 1b
				\multicolumn{16}{l}{\bf {1) $\MCBp_{1}(k) = \MCBp_{2}(k)$ increasing with $k \ge 0$, starting from $0$}} \\ \addlinespace
				& $\frac{1}{4}$  & $\frac{1}{4}$ & $\frac{1}{4}$  &  0
				& $(\sqrt{\tfrac{16}{15}} -1) z_\alpha$ & $\frac{k}{{2}} + \frac{1}{4}$  & 0 & $\frac{k}{{2}} + \frac{1}{4}$ & 0
				& $(\sqrt{\tfrac{16}{15}} -1) z_\alpha$ & 0 & $\frac{k}{{2}} + \frac{1}{4}$ & $\frac{k}{2} + \frac{1}{4}$ & 0 
				& $0 < \DSCp_{1} =\DSCp_{2}$ 	
				\\  \\ 
				
				% part 2b
				\multicolumn{16}{l}{\bf {2) $\MCBp_{1} = 0$ and $\MCBp_{2}(k)$ increasing with $k \ge 0$, starting from $0$}} \\ \addlinespace
				& $\frac{1}{4}$   & $\frac{1}{4}$ & $\frac{1}{4}$  & 0
				& $(\sqrt{\tfrac{16}{15}} -1) z_\alpha$ & $\frac{1}{4}$ & 0 & $ \frac{1}{4} $ & 0 
				& $(\sqrt{\tfrac{16}{15}} -1) z_\alpha$ & 0 & $\frac{k}{2} + \frac{1}{4}$  & $\frac{k}{2} + \frac{1}{4}$  & 0 
				& $0 < \DSCp_{1} = \DSCp_{2}$ 
				\\ 	   
				\\ 
				\midrule
				\\
				
				% DSC o
				\multicolumn{16}{l}{ \bf  \underline{(B) Parameterization of  $\DSCp$ difference}} \\ \\
				
				% part 4b
				\multicolumn{16}{l}{\bf {4) $\DSCp_{1}(k) = \DSCp_{2}(k)$ increasing with $k \ge 0$, starting from $0$}} \\ \addlinespace
				& 0 & 0 & 0 & $k$ 
				& $(\sqrt{1+ \tfrac{8}{15} k^2} -1) z_\alpha$  & $\frac{k}{2}  + \frac{1}{4}$ & 0 & 0 & $\frac{k}{2} + \frac{1}{4}$ 
				& $(\sqrt{1+ \tfrac{8}{15} k^2} -1) z_\alpha$  & 0  & $\frac{k}{2} + \frac{1}{4}$  & 0 & $\frac{k}{2} + \frac{1}{4}$ 
				& $0 < \MCBp_{1}(k) = \MCBp_{2}(k)$ \\ \\

				% part 5b
				\multicolumn{16}{l}{\bf {5) $\DSCp_{1} = 0$ and $\DSCp_{2}(k)$ increasing with $k \ge 0$, starting from $0$}} \\ \addlinespace
				& 0 & 0 & 0 & $k$ 
				& $ \frac{1}{2} $  & $ \frac{1}{4} $ & 0  & $\frac{1}{4}$ & 0 
				& $\xi_0(k,\alpha)$ & 0  & $\frac{k}{2} + \frac{1}{4}$  & 0 & $\frac{k}{2} + \frac{1}{4}$ 
				&$0 < \MCBp_{1}(k) = \MCBp_{2}(k)$ \\ 
				\addlinespace
				\bottomrule
			\end{tabular}
		\end{adjustbox}
	\end{table}
	
	For the evaluation of our tests for \emph{quantile forecasts}, we use the same DGP in \eqref{eqn:SimProcess}.
	As the information in $\mathbf{V}_t$ only affects the conditional mean of $Y_t$ in  \eqref{eqn:SimProcess}, we generate the quantile forecasts at level $\alpha \in (0,1)$ as
	\begin{align}
		\label{eqn:SimFcastq}
		X_{1t}^\alpha =  \mathbf{V}_t^\top \bm{\delta} + \delta_0  + z_\alpha, 
		\qquad \text{and} \qquad 
		X_{2t}^\alpha = \mathbf{V}_t^\top \bm{\xi} + \xi_0  +  z_\alpha,
	\end{align}
	using standard normal quantiles, denoted by $z_\alpha = \Phi^{-1}(\alpha)$.
	
	We deliberately use the pure location process \eqref{eqn:SimProcess} for quantile forecasts, allowing us to explicitly control the population miscalibration and discrimination values for the check loss in Proposition~\ref{prop:sim_q}, albeit these closed-form expressions become much more complicated as for the squared error.
	Due to this additional complication, we only consider four informative parameterizations as given in Table~\ref{tab:SimQuantiles}.
	For these, Appendix~\ref{app:QuantileSimDGP} provides the closed-form expressions for $\MCBp_i$ and $\DSCp_i$ as explicit functions of $k, \alpha$ and the parameters of the DGP.
	Notably, for the last parametrization, (approximate) equality of the miscalibration values is achieved by a numerical choice for the values of $\xi_0 = \xi_0(k,\alpha)$; see Figure~\ref{fig:xi0_choice} and the text of Appendix~\ref{app:QuantileSimDGP}.

	\subsection{Simulation results}
	\label{sec:sim_results}
	
	\begin{figure}[tb]
		\includegraphics[width=\textwidth]{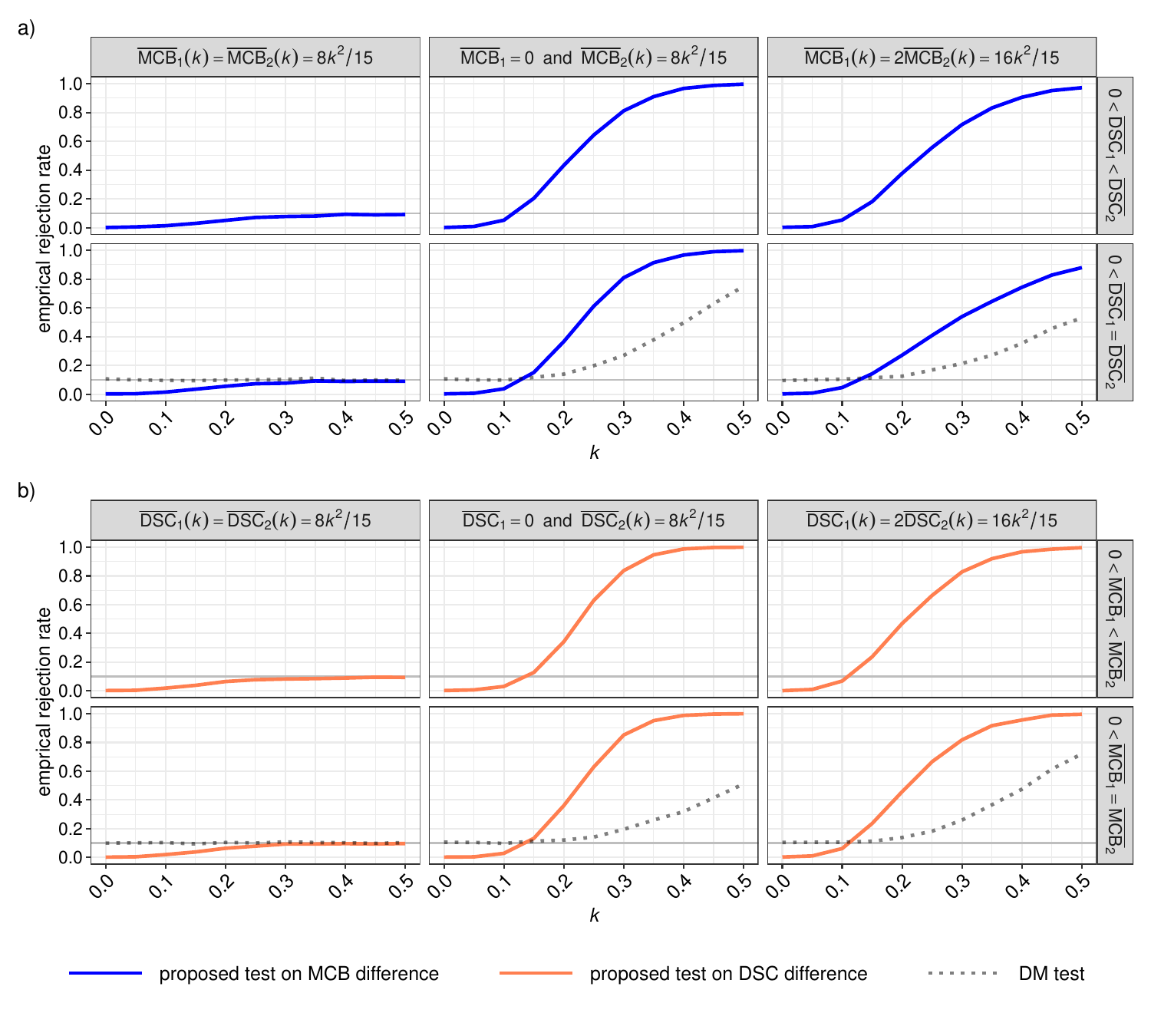}
		\caption{Empirical rejection rates for mean forecasts and the squared error score for the proposed tests of equal miscalibration (in blue) in the upper panel a), equal discrimination (in orange) in the lower panel b), and of the DM test of equal predictive performance (in gray) in the lower plots rows of both panels.
			We generate the data according to \eqref{eqn:SimProcess}--\eqref{eqn:SimFcast2} with the twelve parameterizations of Table~\ref{tab:main}, which depend on the parameter $k \ge 0$ that is displayed on the $x$-axes.
			We use $T = 500$ and a nominal level of $10\%$.}
		\label{fig:sim_param_rr}
	\end{figure}
	
	All results are based on 5000 Monte Carlo replications, a nominal significance level of $10\%$ and $T=500$. 
	We estimate all long-run covariance matrices by the HAC estimator of \citet{A_1991}, adapted to quantile regressions by \citet{GY_2024}.
	For the implementation, we follow the default \texttt{vcovHAC} procedure from the \texttt{R} package \texttt{sandwich} of \citet{Z_2004}.

	For \emph{mean forecasts} and the associated squared error score, Figure~\ref{fig:sim_param_rr} shows the rejection rates of the test for equal miscalibration $\H_\MCB:  \MCBp_{1} = \MCBp_{2}$ in the upper panel in blue color, and of the test for equal discrimination $\H_\DSC:  \DSCp_{1} = \DSCp_{2}$ in orange color in the lower panel.
	For comparison, we also show rejection rates of the DM test of equal predictive power $\H_{\myS}: \Sp_{1} = \Sp_{2}$ with gray and dotted lines. 
	We only show these rejection rates when the underlying population score difference coincides with the respective difference in $\MCBp$ or $\DSCp$ values.
	E.g., in the lower row of panel a), we ensure $\DSCp_{1} = \DSCp_{2}$ such that $\Sp_{1} - \Sp_{2} = \MCBp_1 - \MCBp_2$.
	
	\begin{figure}[tbp]
		\centering
		\includegraphics[width=\textwidth]{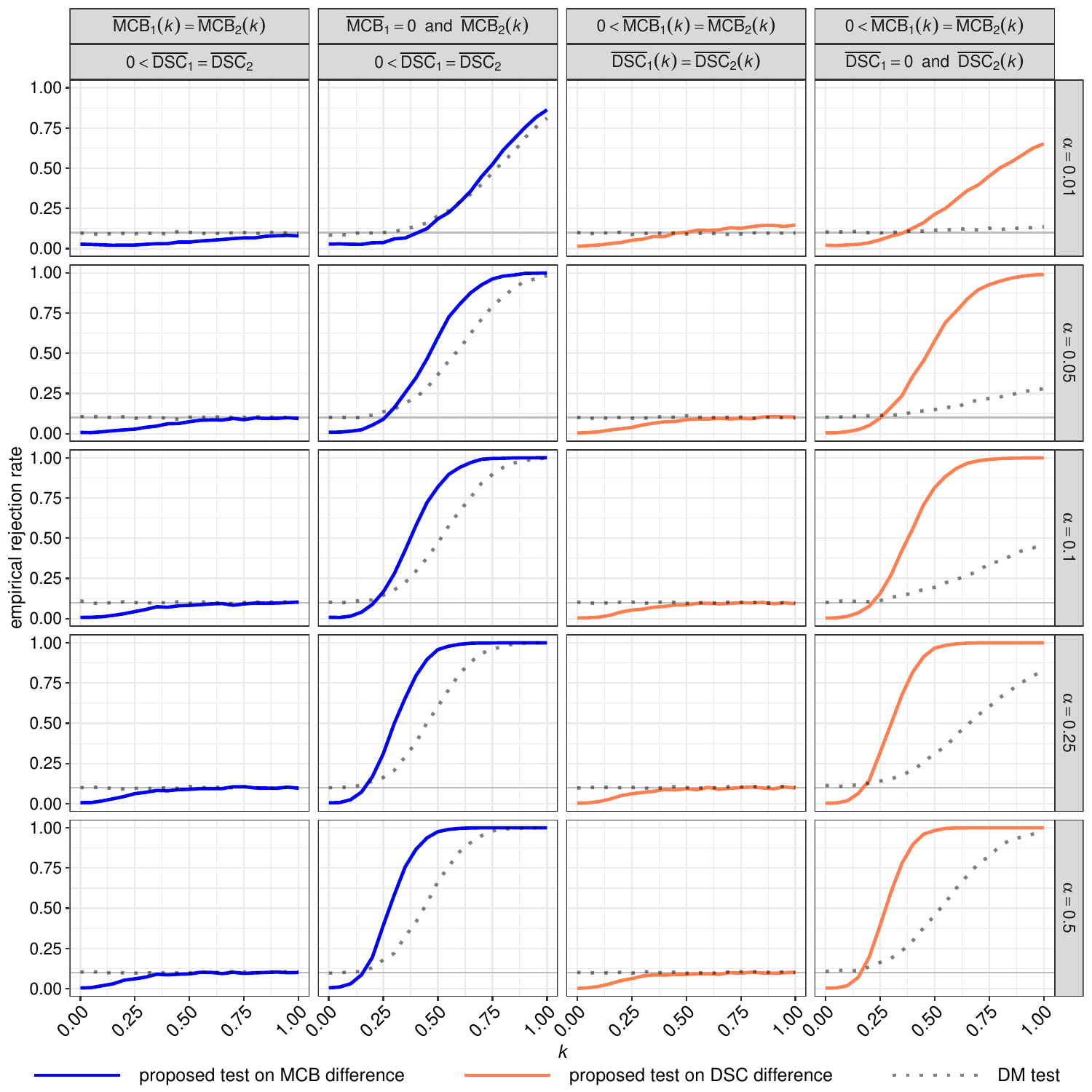}
		\caption{Empirical rejection rates for $\alpha$-quantile forecasts and the check loss for the proposed tests of equal miscalibration (in blue) in the left two columns, equal discrimination (in orange) in the right two columns, and of the DM test of equal predictive performance (in gray) in all subplots.
			We generate the data according to \eqref{eqn:SimProcess} and \eqref{eqn:SimFcastq} with the four parameterizations of Table~\ref{tab:SimQuantiles}, which depend on the parameter $k \ge 0$ that is displayed on the $x$-axes.
			We use $T = 500$ and a nominal level of $10\%$.}
		\label{fig:sim_q}
	\end{figure}

	Throughout the four plots in the left column, as well as for $k=0$ in all twelve plots, the data is generated under the respective null hypotheses of equal miscalibration in panel a), and equal discrimination in panel b).
	Here, we find correct or conservative rejection rates below the nominal level of $10\%$ for our tests, which is to be expected given the Bonferroni-type corrections from Section~\ref{sec:Testing}.
	In the left panel, our tests approach the correct size for larger values of $k$ as e.g., for the $\MCB$-test, miscalibration of both forecasts increases with $k$ such that $2 \min \{ p^0_{\MCB_1}, p^0_{\MCB_2} \}$ becomes much smaller than $p^{+}_{\MCB}$, which operates under its null hypothesis and hence generates asymptotically exact size control.
	
	The middle and right plot columns show that power increases with $k$ across all subplots, illustrating the robustness of the proposed component tests to zero or positive values of the remaining population $\MCBp_i$ and $\DSCp_i$ components. 
	In the scenarios where score differences are driven exclusively by differences in $\MCBp$ (upper panel) or $\DSCp$ (lower panel), the corresponding component test rejects more frequently than the DM test. 
	While the underlying null hypotheses differ and the rejection frequencies are therefore not directly comparable, this finding reflects that the DM test is affected by variability from both score components, whereas the component tests isolate the relevant source of predictive differences. 
	This increased sensitivity is also reflected in the results of the empirical applications in Section~\ref{sec:appl}.

	Figure~\ref{fig:sim_q} displays rejection rates for quantile forecasts at levels $\alpha \in \{0.01, 0.05, 0.1, 0.25, 0.5\}$ and within the four most informative parameterizations that allow for comparison with the DM test.
	These four scenarios correspond to the lower left and middle plots in panels a) and b) from Figure~\ref{fig:sim_param_rr}.
	The results for median forecasts with $\alpha = 0.5$ are comparable to the mean forecasting results from Figure~\ref{fig:sim_param_rr}: The tests are conservative under the null, develop power with $k > 0$ and despite their conservative size\footnote{In unreported simulations, we find that the elevated type I error of the DSC test in the third column for $\alpha = 0.01$ does not increase with $k$ and declines as the sample size grows. This pattern suggests that the distortion is attributable to finite-sample effects.}, our tests reject more frequency than the DM test.
	For more extreme quantile levels, power naturally declines.
	Importantly, especially for the DSC test, this difference in rejection frequencies becomes even more pronounced for extreme quantile levels.

	\section{Applications}
	\label{sec:appl}
	
	Here, we apply the proposed tests for equal miscalibration and discrimination in two central forecasting domains: inflation rates in Section~\ref{sec:appl_infl} and financial risks in Section~\ref{sec:appl_vola}.
	Replication code and details on the data are available at \href{https://github.com/marius-cp/SDI_replications}{https://github.com/marius-cp/SDI\_replications}.

	\subsection{Survey forecasts of U.S.\ annual CPI inflation}
	\label{sec:appl_infl}
	
	Inflation is renowned for being difficult to forecast, and survey forecasts often outperform sophisticated econometric models \citep{FW_2013}. 
	We compare two popular one-year-ahead survey forecasts for quarterly U.S.\ consumer price index (CPI) inflation: the Survey of Professional Forecasters (SPF) of the Federal Reserve Bank of Philadelphia, and the Michigan Survey of Consumers.\footnote{Data sources: SPF forecasts (\texttt{Individual\_CPI.xlsx}) from \href{https://www.philadelphiafed.org/surveys-and-data/cpi-spf}{Philadelphia Fed}; Michigan forecasts (Table 32) from their \href{https://data.sca.isr.umich.edu/tables.php}{data site}; and CPI data (\texttt{cpiQvMd.xlsx}) from \href{https://www.philadelphiafed.org/surveys-and-data/real-time-data-research/cpi}{Philadelphia Fed}.}
	The SPF is considered an expert forecast, as its panelists are professionals with extensive macroeconomic forecasting experience. 
	In contrast, the Michigan survey randomly samples U.S.\ cell phone numbers, without targeting professionals. 
	One might therefore expect the experts to outperform consumers, yet overall forecast comparisons are often inconclusive as e.g., in the empirical applications of \citet{EGJK_2016} and \citet{P_2020}.
	
	The aim of this application is to demonstrate that, even when inference on the overall scores is inconclusive, the improved test power of the individual $\MCB$ and $\DSC$ decomposition terms illustrated in Section~\ref{sec:sim} can produce significant results with clear and interpretable insights while being applicable in general multi-step ahead forecast environments.
	
	\begin{figure}[tb]
		\centering
		\includegraphics[width=\textwidth]{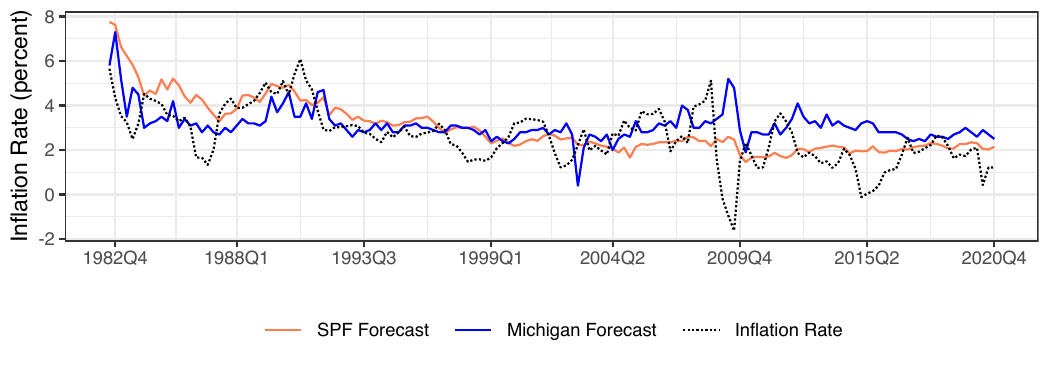}
		\caption{Quarterly year-over-year CPI inflation rate together with the SPF and Michigan forecasts. For details, see the text in Section~\ref{sec:appl_infl}.} 
		\label{fig:appl_infl_ts}
	\end{figure}

	We replicate the analyses of \citet{EGJK_2016} and \citet{P_2020}, comparing the predictive performance of the SPF and Michigan survey forecasts, and extend their approach by applying our novel tests for equal forecast calibration and discrimination.
	Our evaluation sample covers $T = 154$ quarters, from 1982:Q3 to 2020:Q4, thereby excluding the high-inflation period following the COVID-19 pandemic.
	Figure~\ref{fig:appl_infl_ts} displays the forecasts alongside the corresponding realizations.\footnote{The validity of our inference methods could be generalized to non-stationary data at the expense of a more involved notation as long as suitable CLTs and LLNs in Assumption~\ref{ass:Normality} hold.}
	
	Michigan survey respondents report their expectations of the percentage change in prices over the next 12 months. To ensure comparability, we follow \citet{EGJK_2016} and \citet{P_2020} and construct a corresponding SPF measure by averaging the one- to four-quarter-ahead consensus forecasts (where the consensus is defined as the median across respondents). Averaging is appropriate because SPF forecasts are stated in terms of annualized quarter-over-quarter inflation rates.
	As the evaluation target $Y_t$, we use the four-quarter logarithmic growth rate of the CPI, based on the 2023:Q3 vintage, over the year following the forecast date.

	\begin{figure}[tb]
		\centering
		\includegraphics[width=\textwidth]{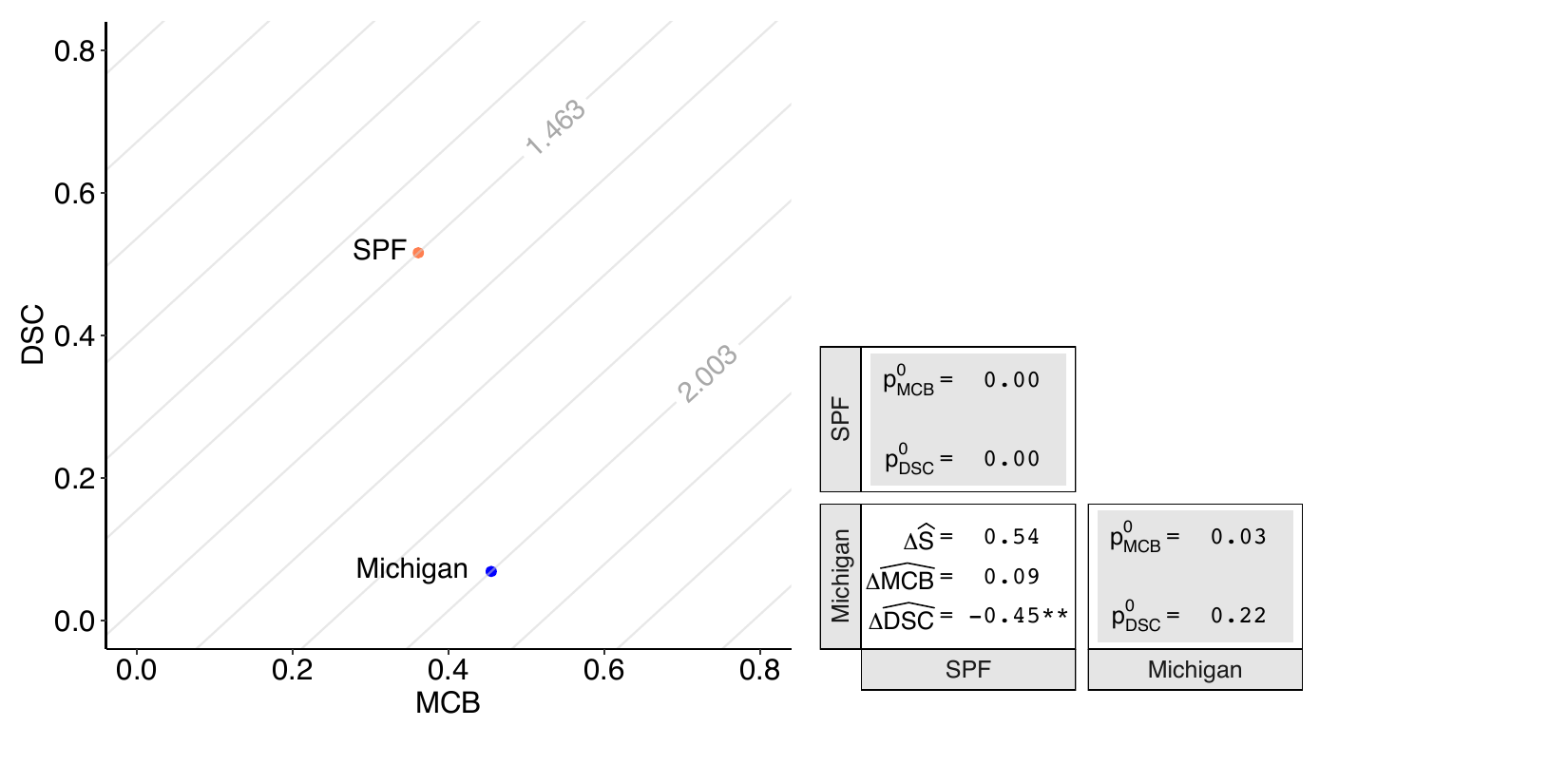}
		\caption{Inflation forecast evaluation results for the SPF and Michigan forecasts using the squared error scoring function. 
			The left ``$\MCB$--$\DSC$'' plots show the average score as iso-lines together with its decomposition into miscalibration and discrimination components for the four competing forecasts. 
			The gray boxes in the right panels display the $p$-values of tests for zero miscalibration and discrimination, respectively.
			The white box displays the average score, miscalibration and discrimination differences, together with 0 to 3 stars indicating significance at the 10\%, 5\%, and 1\% levels.}
		\label{fig:appl_infl}
	\end{figure}

	The left panel of Figure~\ref{fig:appl_infl} reports the average squared error score together with its decomposition into miscalibration and discrimination components for the two competing forecasts. 
	Calibration appears broadly similar across the two surveys. 
	In contrast, the SPF forecasts display a markedly higher degree of discrimination---a difference that is not visible in the purely scoring-function-based analyses of \citet{EGJK_2016} and \citet{P_2020}.
	The right panel corroborates their main finding: the overall difference in predictive performance is statistically insignificant, with a $p$-value of~0.284. 
	
	% pval equal MCB 0.811 
	Consistent with the visual impression from the left panel, the difference in miscalibration is not statistically significant. However, the SPF forecasts exhibit significantly stronger discrimination, with a $p$-value of~0.037. For the Michigan survey forecasts, we cannot even reject the null hypothesis of no discrimination.
	These findings suggest that professional forecasters’ predictions contain more information about future inflation outcomes, which is plausible given their field-specific expertise and experience. 
	
	Figure~\ref{fig:MZdiagnosticsInflation} in Appendix~\ref{sec:AddFigs} reports diagnostic checks for the linear recalibration regression and reveals no evidence of misspecification, supporting the use of linear recalibration for inference on score decompositions.
	
	From a methodological perspective, this application underscores the utility of our proposed inference framework. 
	Beyond partitioning forecast performance into interpretable economic components, our tests can exhibit superior statistical power compared to the DM test in these settings. 
	Furthermore, the framework is applicable in general time-series environments, including multi-step-ahead forecasting horizons.

	\subsection{Financial volatility and Value-at-Risk forecast performance}
	\label{sec:appl_vola}
	
	Here, we extend the motivational example from Section~\ref{sec:motex} and consider one-day-ahead forecasts of the variance and the $\alpha$-quantiles ($\alpha$-VaR), with $\alpha \in \{1\%, 5\%\}$, of financial log returns $R_t$.
	This application demonstrates the value of inference based on score decompositions, as it helps reconcile the partly contradictory findings that often arise between traditional backtests reviewed in Section~\ref{sec:VaRBacktestsTheory} and scoring-function-based forecast evaluations in applied work such as in Section~\ref{sec:motex}.
	
	Given the standard assumption $\E[R_t] = 0$ such that $\var(R_t) = \E[R_t^2]$, the variance forecasts can be treated as forecasts for the expectation of $R_t^2$, or as introduced below, the expectation of an unbiased RV estimator.
	Hence, we follow \citet{P_2011} and evaluate the forecasts with the squared error and QLIKE loss functions that are special cases of the Bregman class~\eqref{eqn:BregmanLoss}, considered in Proposition~\ref{prop:verification_phi}.
	For the VaR forecasts, we use the check loss from \eqref{eqn:GPLLoss}, as covered by Proposition~\ref{prop:Quantiles_Verification}.
	
	For $R_t$, we use logarithmic returns of the S\&P 500 E-mini futures, a highly liquid instrument central to practical portfolio management, which is widely used in the financial variance forecasting literature \citep{C_2009, LPS_2015, BHHP_2018}.
	E-mini contracts trade almost continuously from Sunday evening to Friday afternoon, with brief daily interruptions from 3:15–3:30 p.m.\ and 4:00–5:00 p.m.\ CT.
	A trading day is defined as the interval spanning from 5:00 p.m.\ CT of the preceding calendar day to 4:59 p.m.\ CT.
	We use intraday price observations from the data provider Refinitiv (ticker: ESc1), spanning the period from January 1, 2000, to January 1, 2022. 
	From this, we compute the daily open-to-close return $R_t$ and the 5-minute realized variance, $\RV_t := \sum_{m=1}^M r_{tm}^2$, where \(r_{tm}\) is the \(m\)-th intraday return on day \(t\) and \(M\) denotes the number of equally spaced observations per trading day.
	
	We employ several standard models to generate one-step-ahead variance forecasts. 
	We first consider the standard GARCH(1,1) model with Gaussian innovations of \citet{B_1986}.
	Second, we employ an asymmetric extension, the GJR-GARCH model of \citet{GJR_1993} with Student’s $t$ residuals, given by
	\begin{align} 
		\label{eqn:gjr}
		\begin{split}
			R_t &= \sqrt{\sigma_t^2} \, \varepsilon_t, 
			\qquad \text{with} \qquad 
			\varepsilon_t \overset{\text{iid}}{\sim} t_\nu, \\
			\sigma_t^2 &= \alpha_0 + \alpha_1 \sigma_{t-1}^2 + \alpha_2 R_{t-1}^2 + \alpha_3 R_{t-1}^2 \mathds{1}\{R_{t-1} < 0 \}.
		\end{split}
	\end{align}
	
	We further use two models that use high-frequency information through the RV estimator.
	The third model is the Heterogeneous Autoregressive (HAR) model proposed by \citet{C_2009},
	\begin{align*}
		{\operatorname{RV}_t^{1/2}}  = \beta_0 +  \beta_1 {\operatorname{RV}_{t-1}^{1/2}}
		+  \beta_2 \overline{\operatorname{RV}}_{{w},t-1}^{1/2} 
		+  \beta_3 \overline{\operatorname{RV}}_{{m},t-1}^{1/2} 
		+ \varepsilon_t,  
		\qquad \text{with} \qquad 
		\varepsilon_t \overset{\text{iid}}{\sim}  \mathcal{N}(0,\sigma^2),
	\end{align*} 
	where $\overline{\operatorname{RV}}_{w,t}^{1/2} := \frac{1}{5}\sum_{j=0}^4 {\operatorname{RV}_{t-j}^{1/2}}$ and  $\overline{\operatorname{RV}}_{m,t}^{1/2} := \frac{1}{22}\sum_{j=0}^{21} {\operatorname{RV}}_{t-j}^{1/2}$ are weekly and monthly averages, respectively.
	Fourth, we employ the MIDAS model of \citet{GSV_2006, GKZ_2016}, where ${\operatorname{RV}_t^{1/2}}$ is predicted based on a weighted average of its lagged values,
	\begin{align*}
		\operatorname{RV}_{t}^{1/2} = \phi_0 + \sum_{l=1}^{20} w_l \operatorname{RV}_{t-l}^{1/2} + \varepsilon_t, \quad \text{with} \quad \varepsilon_t \overset{\text{iid}}{\sim} \mathcal{N}(0, \sigma^2),
	\end{align*}
	where the MIDAS weights $w_l$ are parameterized through the exponential Almon polynomial.
	
	We construct VaR forecasts for these four models by multiplying the square root of the respective variance forecast with the $\alpha$-quantile of the imposed Gaussian residual distribution for all models but the GJR-GARCH, where we use a $t$-distribution.
	For quantile forecasts, we additionally use the classical Historical Simulation (HS) method, that uses the empirical $\alpha$-quantile of the preceding 250 trading days as a quantile forecast.
	
	For the first four models, we estimate the parameters by maximum likelihood using the imposed residual distributions, based on a fixed estimation window from January 1, 2000, to December 31, 2007, corresponding to approximately the first third of the sample.
	The subsequent period from January 1, 2008 to January 1, 2022 serves as  evaluation period, yielding $T=3617$ one-step ahead forecasts and corresponding realizations.
	Figure~\ref{fig:RiskApplicationTimeSeriesPlot} in Appendix~\ref{sec:AddFigs} shows the variance and VaR forecasts together with their evaluation targets.
	
	%%%%%%% Variance results

	\begin{figure}[tbp]
		\centering
		\includegraphics[width=\textwidth]{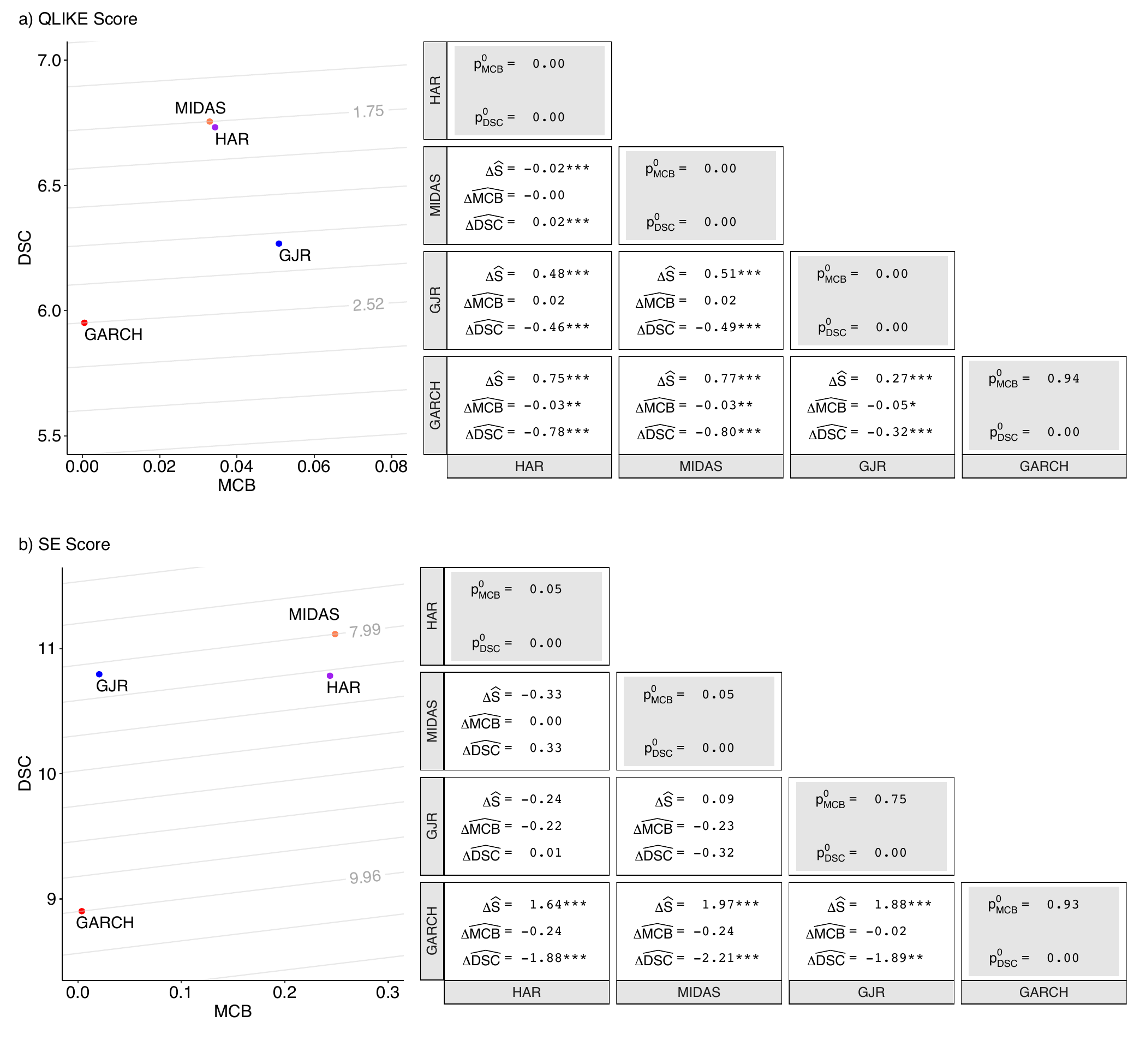}
		\caption{Variance forecast evaluation results for the E-mini futures using a) the QLIKE and b) the squared error scoring functions.
			The left ``$\MCB$--$\DSC$'' plots show the average score as iso-lines together with its decomposition into miscalibration and discrimination components for the four competing forecasts. 
			The gray boxes in the right panels display the $p$-values of tests for zero miscalibration and discrimination, respectively.
			The white boxes display the average score, miscalibration and discrimination differences, together with 0 to 3 stars indicating significance at the 10\%, 5\%, and 1\% levels.}
		\label{appl:vola}
	\end{figure}

	Figure~\ref{appl:vola} presents the evaluation of variance forecasts under the QLIKE and SE scoring functions. The left panels, following \citet{DGJV_2024}, plot the estimated miscalibration and discrimination of the four forecasting methods along the axes. Gray iso-lines represent average score levels, with observations closer to the upper-left corner indicating superior performance. 
	We find that the RV-based forecasting methods exhibit markedly stronger discrimination, which translates into better overall predictive accuracy as reflected in the average scores.
	
	The right panel of Figure \ref{appl:vola} reports the corresponding test results:
	Gray-shaded boxes display $p$-values for tests of zero miscalibration and zero discrimination. While all models show significantly positive discrimination, the null of correct calibration is rejected more frequently for the RV-based models than for the two GARCH-type specifications.
	
	The remaining off-diagonal panels present pairwise differences in average scores ($\Delta \widehat{\myS}$), miscalibration ($\Delta \MCBs$), and discrimination ($\Delta \DSCs$), with significance indicated by 0 to 3 stars corresponding to the 10\%, 5\%, and 1\% levels. 
	Focusing on the lower column, which compares the GARCH model to the three alternatives, most differences are highly significant. 
	In particular, the inferior performance of the GARCH model is primarily driven by its significantly weaker discrimination, consistent with the informational advantage of RV-based models that exploit high-frequency data.

	%%%%%%% VaR results

	\begin{table}
		\caption{Backtesting results for $1\%$ and $5\%$ VaR forecasts. The first panel displays the $p$-values of the backtests summarized in Table~\ref{tab:BacktestTheory}.
			The column ``Basel'' also displays the colors of their traffic light system.
			The last column reports the relative frequencies of hits, i.e., how often the realizations exceed the VaR forecasts.}
		\label{tab:BacktestResults}    
		\centering
		\begin{tabular}[t]{llrlrrrrrr}
			\toprule
			\multicolumn{2}{c}{ } & \multicolumn{7}{c}{p-values backtests} & \\
			\cmidrule(l{3pt}r{3pt}){3-9} 
			& Model & UC & Basel & CC & NZ & VQR & DQ & DQX & hit freq. \\
			\midrule
			\multirow{5}{*}{$1\%$-VaR} & HS & 0.137 & \cellcolor{Goldenrod!20}{\textcolor{black}{0.024}} & 0.149 & 0.124 & 0.321 & 0.002 & 0.004 & $1.3\%$\\
			& GARCH & 0.001 & \cellcolor{Red!15}{\textcolor{black}{0.000}} & 0.002 & 0.000 & 0.006 & 0.023 & 0.016 & $1.8\%$\\
			& GJR & 0.000 & \cellcolor{Red!15}{\textcolor{black}{0.000}} & 0.000 & 0.000 & 0.000 & 0.002 & 0.001 & $2.0\%$\\
			& HAR & 0.000 & \cellcolor{Red!15}{\textcolor{black}{0.000}} & 0.000 & 0.000 & 0.003 & 0.000 & 0.000 & $2.3\%$\\
			& MIDAS & 0.000 & \cellcolor{Red!15}{\textcolor{black}{0.000}} & 0.000 & 0.000 & 0.000 & 0.000 & 0.000 & $2.3\%$\\
			\addlinespace\hline\addlinespace
			
			\addlinespace
			\multirow{5}{*}{$5\%$-VaR} & HS & 0.700 & \cellcolor{ForestGreen!15}{\textcolor{black}{0.303}} & 0.006 & 0.416 & 0.013 & 0.000 & 0.000 & 5.2\%\\
			& GARCH & 0.266 & \cellcolor{ForestGreen!15}{\textcolor{black}{0.117}} & 0.543 & 0.140 & 0.452 & 0.066 & 0.052 & 5.4\%\\
			& GJR & 0.161 & \cellcolor{ForestGreen!15}{\textcolor{black}{0.069}} & 0.313 & 0.073 & 0.360 & 0.207 & 0.534 & 5.5\%\\
			& HAR & 0.298 & \cellcolor{ForestGreen!15}{\textcolor{black}{0.132}} & 0.588 & 0.159 & 0.518 & 0.190 & 0.132 & 5.4\%\\
			& MIDAS & 0.266 & \cellcolor{ForestGreen!15}{\textcolor{black}{0.117}} & 0.544 & 0.145 & 0.568 & 0.183 & 0.134 & 5.4\%\\
			\bottomrule
		\end{tabular}
	\end{table}

	Turning to the VaR evaluation results, Table~\ref{tab:BacktestResults} reports the $p$-values of seven standard VaR backtests---described in Section~\ref{sec:VaRBacktestsTheory} and Table~\ref{tab:BacktestTheory}---along with the unconditional hit frequency in the final column.
	Strikingly, particularly at the 1\% level, the simple HS model appears to be the best calibrated, whereas all competing models are consistently rejected across the backtests. 
	Although evidence at the 5\% level and for conditional calibration backtests is weak or inconclusive, the 1\% level remains the most relevant for risk management \citep[p.~81]{B_2019}.
	
	\begin{figure}[tbp]
		\centering
		\includegraphics[width=\textwidth]{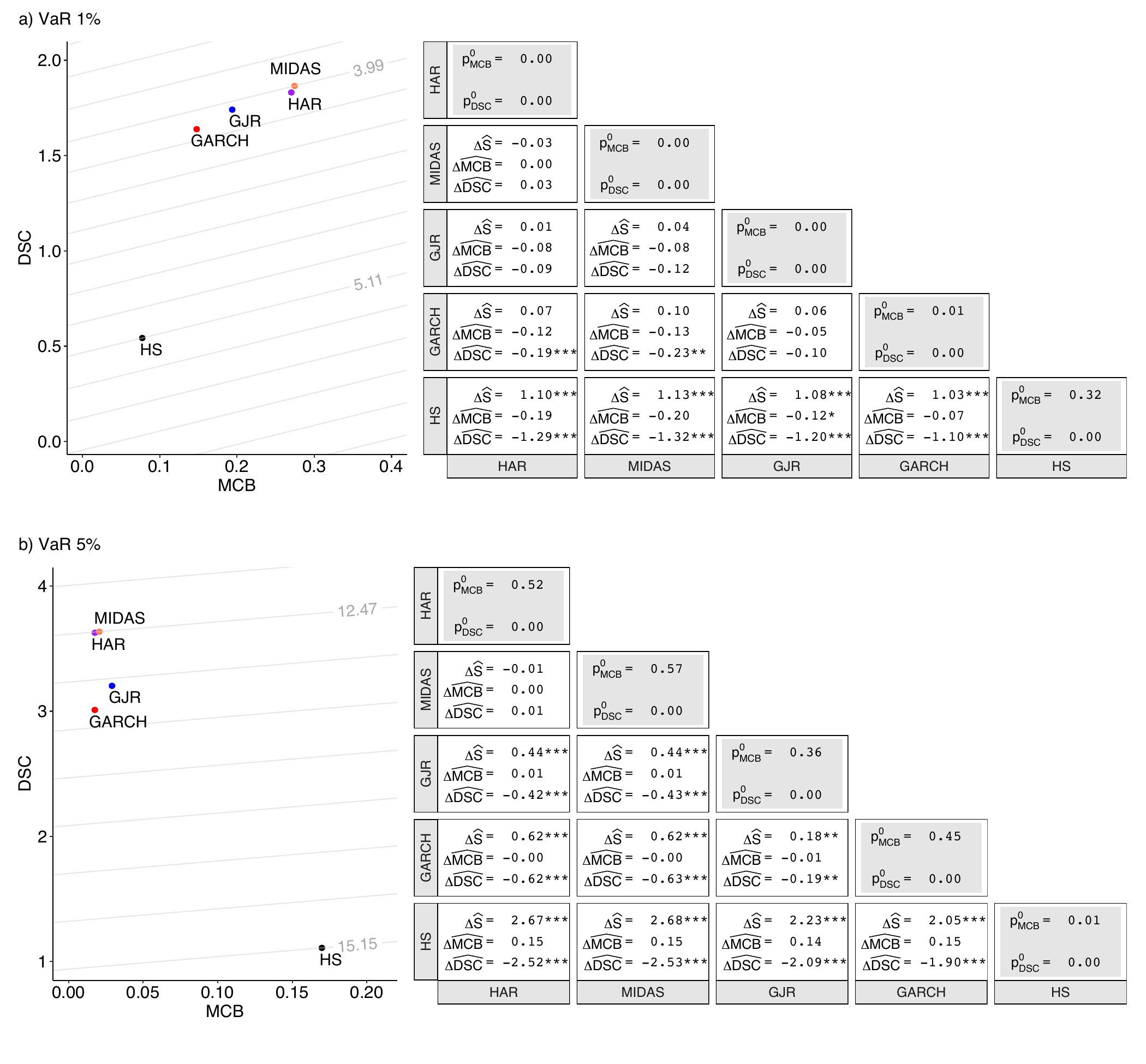}
		\caption{VaR forecast evaluation results for the E-mini futures at quantile levels a) $1\%$ and b) $5\%$ for the check loss function and its decomposition terms.
			For details, see the caption of Figure~\ref{appl:vola}.}
		\label{fig:appl_var}
	\end{figure}
	
	Score decompositions with their associated inference help explain the seemingly superior performance of the HS model. Figure~\ref{fig:appl_var} presents these decompositions, analogous to Figure~\ref{appl:vola}, using the check loss function for 1\% and 5\% VaR forecasts.
	
	At the 1\% level, the $\MCB$–$\DSC$ plot in Figure~\ref{fig:appl_var} shows that, although the HS model exhibits the best calibration, it has by far the weakest discrimination. 
	This results in clearly inferior predictive performance according to the check loss, as further confirmed by the strongly significant results in the right panel.
	These findings resolve the seemingly contradictory results highlighted in the motivational example from Section~\ref{sec:motex} and underscore the risk of evaluating forecasts solely on calibration. 
	As such, they point to a potentially serious shortcoming in current banking regulation.
	
	At the 5\% level, the HS model shows markedly poorer \emph{conditional} calibration, as reflected both in the backtest results in Table~\ref{tab:BacktestResults} and in the score decomposition in Figure~\ref{fig:appl_var}. 
	The relative ordering of the other four models, however, is similar to that observed at the 1\% level, with generally stronger significance, which can be attributed to lower noise at less extreme probability levels.
	
	A general pattern evident in the $\MCB$–$\DSC$ plot at the 1\% level---and, to a lesser extent, at the 5\% level and for the variance forecasts in Figure~\ref{appl:vola}---is that models incorporating more information can convert it into higher discrimination. 
	However, this often comes at the cost of increased miscalibration.
	Notably, achieving auto-calibration conditional on $\mW_{it} = (1, X_{it})$ is relatively easy when forecasts $X_{it}$ are nearly constant, but much more difficult for highly discriminating forecasts that vary substantially over time. 
	This observation calls into question the common practice of evaluating risk measure forecasts primarily through MZ-type backtests \citep{Gaglianone2011, Guler2017, BD_2020}, which focus on auto-calibration.
	
	Figures~\ref{fig:MZdiagnosticsVola}--\ref{fig:MZdiagnosticsVaR} in Appendix~\ref{sec:AddFigs} present diagnostic checks for the linear recalibration regressions underlying the score decompositions of the variance and VaR forecasts. 
	The pointwise confidence bands largely cover the null specification, providing no evidence against correct model specification. 
	This supports the use of linear recalibration for score decompositions in these applications, combining robustness and adequate fit with the ability to conduct inference, which is essential in the settings considered.

	\section{Conclusion}
	\label{sec:conclusion}
	
	This paper has established a comprehensive inferential framework for score decompositions, enabling a more granular assessment of predictive performance through the lenses of miscalibration, discrimination, and uncertainty. 
	By utilizing a linear recalibration technique, we have shown that it is possible to derive misspecification-robust inference for a wide range of point forecasts---including those governed by non-smooth loss functions---while maintaining attractive finite sample non-negativity conditions and a formal bridge to the classical \citet{MZ_1969} regression. 
	Our results demonstrate that this decomposition approach enriches the informational content of traditional predictive ability tests.
	Through our empirical applications, we have highlighted how these methods can reveal significant latent differences in forecast discrimination and, crucially, identify systemic deficiencies in financial risk backtesting that current banking regulations fail to capture.
	
	The linear recalibration specification is exact under the null hypotheses of perfect calibration or zero discrimination and provides a stable and flexible approximation more generally. 
	It naturally accommodates additional recalibration variables, facilitating extensions to, for example, cross-calibration \citep{SZ_2017} and links to VaR backtesting and macroeconomic survey forecasts \citep{CG_2015}. 
	Compared with nonparametric alternatives, the linear specification is also less susceptible to overfitting, particularly in the small-to-moderate sample sizes considered here. Moreover, the appendix reports no empirical evidence against the linearity assumption in our applications.
	
	Nonetheless, particularly in data-rich environments, semiparametric extensions that estimate the recalibration curve nonparametrically---using kernel, spline, or isotonic regression---represent promising avenues for future research. 
	Such developments could build on the seminal work of \citet{newey1994asymptotic} regarding two-step semiparametric inference and more recent advances in isotonic plug-in estimation, such as \citet{Xu_IsoPlugIn_2022}.
	
	We treat the possibly model-based forecasts as the finalized output of a specific sequence of modeling choices, including parameter estimation. 
	Consequently, we formulate our hypotheses based on the realized quality of these forecasts rather than on the properties of the underlying models themselves. 
	In contrast, the framework of \citet{West1996asymptotic} assesses whether two forecasting \emph{models} based on their pseudo-true parameters would exhibit equal predictive accuracy, hence requiring that the test statistic's standardization explicitly accounts for the asymptotic noise introduced by parameter estimation.
	This extension can be integrated into our framework by redefining the population $\MCBp$ and $\DSCp$ components in a model-based environment and replacing the central limit theorem in Assumption~\ref{ass:Normality}~\ref{item:CLT} with the asymptotic approximations of \citet{West1996asymptotic}, or with the numerous refinements and generalizations proposed in subsequent work.
	
	Generalizations of our framework to other elicitable functionals, such as expectiles, are straightforward. 
	However, functionals that are merely jointly or multi-objective elicitable following \citet{FZ_2016} and \citet{FisslerHoga2024}---such as the variance (when the mean is unknown), the Expected Shortfall, or the systemic risk measures Co-Value-at-Risk and Marginal Expected Shortfall---require additional care.
	In these cases, the underlying recalibration regressions necessitate the estimation of conditional models for the associated nuisance components, such as the mean or the VaR \citep{patton2019dynamic, BD_2020, dimitriadis2025regressions}.
	The propagation of estimation risk from these nuisance quantities into the final score decomposition represents a non-trivial extension of the current inferential framework.
	
	In conclusion, the contributions of this paper extend beyond their immediate applications in econometrics and finance. The principles of score decomposition and linear recalibration constitute a versatile toolkit for any domain where predictive accuracy is paramount, from meteorology to machine learning. The framework proposed herein represents a fundamental shift in perspective: from merely asking \emph{which} forecast is better to diagnosing \emph{why} it is better. As the complexity of our models and the richness of our data continue to grow, this diagnostic approach will be essential for robust and reliable decision-making.

	\section*{Acknowledgments}
	
	Timo Dimitriadis acknowledges funding by the German Research Foundation (DFG) through the projects 502572912 and 568876076.
	We gratefully acknowledge the Hohenheim Datalab (DALAHO) for providing access to Refinitiv TickHistory.
	We thank all seminar and conference participants who provided feedback on earlier versions of this paper. 
	In particular, we are grateful to Sam Allen, Tillman Gneiting, Daniel Gutknecht, Alexander Jordan, Robert Jung, Fabian Krüger, Andrew Patton, Winfried Pohlmeier, Johannes Resin, Karsten Schweikert, and Johanna Ziegel for their valuable comments and suggestions.

	\section*{Declaration of generative AI use}
	
	During the preparation of this work, the authors used the AI tools ChatGPT and Gemini for language editing, code generation, and assistance in setting up the simulation process for quantile forecasts described in Appendix~\ref{app:QuantileSimDGP}.
	After using these tools, the authors reviewed and edited the content as needed and take full responsibility for the content of the article.

	%=============================================================
	%=============================================================
	% biblio
	%=============================================================
	%=============================================================
	
	\small 
	\singlespacing
	\setlength{\bibsep}{3pt} 
	
	\bibliographystyle{apalike}
	\addcontentsline{toc}{section}{References} 
	\bibliography{biblio}

	\pagebreak 
	\appendix
	
	\huge 
	\noindent
	\textbf{Supplementary Material}
	\normalsize
	\bigskip 
	\bigskip 
	
	\numberwithin{figure}{section}
	\numberwithin{table}{section}

	\noindent 
	The supplementary material contains all proofs in Appendix~\ref{sec:Proofs},
	a discussion of related score decompositions in Appendix~\ref{sec:connecttolit},
	details on the quantile simulation DGP in Appendix~\ref{app:QuantileSimDGP},
	and additional figures for the simulations and applications in Appendix~\ref{sec:AddFigs}.
	The references refer to the bibliography provided in the main paper.

	\section{Proofs}
	\label{sec:Proofs}
	
	\begin{proof}[\bf Proof of Theorem \ref{thm:PositivComponenets}.]
		For item \ref{item:MCBs>0}, notice that by the definition of the M-estimator in \eqref{eqn:Mestimator}, we have that
		\begin{align*}
			\MCBs_{iT} 
			= \Ss_{iT} - \Ss^{\mathsf{C}}_{iT}
			&= \frac{1}{T} \sum_{t=1}^T \myS \big(X_{it} , Y_t \big) - \frac{1}{T} \sum_{t=1}^T \myS \big(\mW_{it}^\top \widehat{\mtheta}_{iT} , Y_t \big) \\
			&= \frac{1}{T} \sum_{t=1}^T \myS \big(X_{it} , Y_t \big) -  \underset{\mtheta \in \mTheta}{\min} \; \frac{1}{T} \sum_{t=1}^T \myS \big(\mW_{it}^\top \mtheta , Y_t \big) \ge 0,
		\end{align*}
		since $\mW_{it}$ is assumed to contain an intercept and the forecast $X_{it}$.
		If $X_{it} = \mW_{it}^\top \widehat{\mtheta}_{iT}$ almost surely, $\MCBs_{iT} = 0$ is immediate.
		
		For item \ref{item:MCBs=0}, let $X_{it} \not= \mW_{it}^\top \widehat{\mtheta}_{iT}$ for some $t \in \{1,\dots,T\}$ with positive probability. 
		Then, we get that
		\begin{align*}
			\frac{1}{T} \sum_{t=1}^T \myS \big(X_{it} , Y_t \big) 
			> \frac{1}{T} \sum_{t=1}^T \myS \big( \mW_{it}^\top \widehat{\mtheta}_{iT} , Y_t \big) 
			= \underset{\mtheta \in \mTheta}{\min} \; \frac{1}{T} \sum_{t=1}^T \myS \big(\mW_{it}^\top \mtheta , Y_t \big),
		\end{align*}
		with positive probability, where the inequality holds as the latter minimum is assumed to be unique and $(0,1,0,\dots)^\top \in \mTheta$, for which $\mW_{it}^\top (0,1,0,\dots)^\top = X_{it}$.
		Thus, $\MCBs_{iT} > 0$ with positive probability follows.
		
		For item \ref{item:DSCs>0}, first notice that the empirical distribution of the sample $Y_1, \dots, Y_T$ is contained in $\mathcal{P}$ as the convex combination of Dirac measures at $Y_1, \dots, Y_T$, such that the strict consistency of $\myS$ implies that the unconditional functional $\widehat{r}_{T}$ coincides with the M-estimator  $\underset{r \in \mathsf{A}}{\argmin} \; \frac{1}{T} \sum_{t=1}^T \myS \big(r, Y_t \big)$.
		Hence, we have that 
		\begin{align*}
			\DSCs_{iT} 
			= \Ss^{\mathsf{R}}_{iT} - \Ss^{\mathsf{C}}_{iT}
			&= \frac{1}{T} \sum_{t=1}^T \myS \big(\widehat{r}_{T} , Y_t \big) - \frac{1}{T} \sum_{t=1}^T \myS \big(\mW_{it}^\top \widehat{\mtheta}_{iT} , Y_t \big) \\
			&= \underset{r \in \mathsf{A}}{\min} \; \frac{1}{T} \sum_{t=1}^T \myS \big(r, Y_t \big) 
			- \underset{\mtheta \in \mTheta}{\min} \; \frac{1}{T} \sum_{t=1}^T \myS \big(\mW_{it}^\top \mtheta , Y_t \big) \ge 0,
		\end{align*}
		as $\mW_{it}$ is assumed to contain an intercept, which nests the first minimization over all $r \in \mathsf{A}$.
		If $\widehat{r}_T = \mW_{it}^\top \widehat{\mtheta}_{iT}$ almost surely, $\DSCs_{iT} = 0$ is immediate.
		
		For item \ref{item:DSCs=0}, assume that $\widehat{r}_{T} \not= \mW_{it}^\top \widehat{\mtheta}_{iT}$ for some $t \in \{1,\dots,T\}$ with positive probability. 
		Then,
		\begin{align*}
			\underset{r \in \mathsf{A}}{\min} \; \frac{1}{T} \sum_{t=1}^T \myS \big(r, Y_t \big) = \frac{1}{T} \sum_{t=1}^T \myS \big(\widehat{r}_{T}, Y_t \big) 
			> \frac{1}{T} \sum_{t=1}^T \myS \big( \mW_{it}^\top \widehat{\mtheta}_{iT} , Y_t \big) 
			= \underset{\mtheta \in \mTheta}{\min} \; \frac{1}{T} \sum_{t=1}^T \myS \big(\mW_{it}^\top \mtheta , Y_t \big),
		\end{align*}
		with positive probability, where the strict inequality holds as the latter minimum is assumed to be unique and  $(\widehat{r}_T,0,\dots)^\top \in \mTheta$, for which $\mW_{it}^\top (\widehat{r}_T,0,\dots)^\top = \widehat{r}_T$. 
		Thus, $\DSCs_{iT} > 0$ with positive probability follows, which concludes the proof.
	\end{proof}

	\begin{proof}[\bf Proof Theorem \ref{thm:joint_nomal_mcb_dsc}] 
		We start by decomposing
		\begin{align*}
			\MCBs_{iT} - \MCBp_{i} = A_{iT} - B_{iT} - C_{iT} 
			\qquad \text{ and } \qquad 
			{\DSCs}_{iT} - \DSCp_{i} = - B_{iT} - C_{iT} + D_{T} + E_{T} 
		\end{align*}
		where 
		\begin{align*}
			A_{iT} &:= \frac{1}{T} \sum_{t=1}^T a_{it} 
			:= \frac{1}{T} \sum_{t=1}^T \myS(X_{it},Y_t) - \E [ \myS(X_{it},Y_t)] \\
			B_{iT} &:= \frac{1}{T} \sum_{t=1}^T b_{itT} 
			:=\frac{1}{T}  \sum_{t=1}^T \myS (\mW_{it}^\top \widehat{\mtheta}_{iT} , Y_t ) - \myS( \mW_{it}^\top \bar\mtheta_i,Y_t) \\
			C_{iT} &:= \frac{1}{T} \sum_{t=1}^T c_{it} 
			:=\frac{1}{T}  \sum_{t=1}^T \myS( \mW_{it}^\top \bar\mtheta_i,Y_t) -  \E [ \myS(\mW_{it}^\top \bar\mtheta_i, Y_t)] \\
			D_{T} &:= \frac{1}{T} \sum_{t=1}^T d_{tT} 
			:= \frac{1}{T} \sum_{t=1}^T \myS(\hat r_T, Y_t) - \myS(\bar r,Y_t)  \\
			E_{T} &:= \frac{1}{T} \sum_{t=1}^T e_{t} 
			:=\frac{1}{T} \sum_{t=1}^T \myS(\bar r,Y_t) - \E[ \myS(\bar r,Y_t)].
		\end{align*}
		Here, $A_{iT}$, $C_{iT}$ and $E_T$ capture the difference between sample and population versions of the scores of the original forecasts, the (population) recalibrated and the (population) reference forecasts, respectively.
		The terms $B_{iT}$ and $D_T$ capture the estimation effect of the recalibrated and reference forecast, respectively.
		
		We start to show that the estimation effects captured in $\sqrt{T} B_{iT}$ (and similarly, in $\sqrt{T} D_T$) vanish in probability.
		For this,  we define for all $\mtheta \in \mTheta$,
		\begin{align}
			\label{eqn:DefQs}
			\widehat{Q}_{iT}(\mtheta) := \frac{1}{T} \sum_{t=1}^T \myS( \mW_{it}^\top \mtheta, Y_t)
			\qquad \text{ and } \qquad 
			\overline{Q}_{iT}(\mtheta) := \frac{1}{T} \sum_{t=1}^T \E \big[ \myS( \mW_{it}^\top \mtheta, Y_t) \big],
		\end{align}
		such that $\nu_{iT}(\mtheta) = \sqrt{T} \big( \widehat{Q}_{iT}(\mtheta) - \overline{Q}_{iT}(\mtheta) \big)$ for $\nu_{iT}(\mtheta)$ defined in~\eqref{eqn:EmpProcessSE}.
		
		Then, we employ a mean-value theorem in the third equality below to get for some (possibly random) $\widetilde{\mtheta}_{iT}$ that is on the line between $\widehat{\mtheta}_{iT}$ and $\bar{\mtheta}_{i}$ that
		\begin{align}
			\begin{aligned}
				\label{eqn:hatQ_orders}
				\sqrt{T} B_{iT} 
				&= \sqrt{T} \widehat{Q}_{iT}(\widehat{\mtheta}_{iT}) - \sqrt{T} \widehat{Q}_{iT}(\bar{\mtheta}_{i}) \\
				&= \sqrt{T} \left( \overline{Q}_{iT}(\widehat{\mtheta}_{iT}) - \overline{Q}_{iT}(\bar{\mtheta}_{i}) \right) + \left( \nu_{iT}(\widehat{\mtheta}_{iT}) - \nu_{iT}(\bar{\mtheta}_{i}) \right) \\
				&= \nabla_{\mtheta} \overline{Q}_{iT}(\widetilde{\mtheta}_{iT}) \cdot \sqrt{T} \big(\widehat{\mtheta}_{iT} -  \bar\mtheta_i \big)  + \left( \nu_{iT}(\widehat{\mtheta}_{iT}) - \nu_{iT}(\bar{\mtheta}_{i}) \right) \\
				&= o_\P(1) \cdot \mathcal{O}_\P(1) +  o_\P(1) =  o_\P(1).
			\end{aligned}
		\end{align}
		We now provide more details on the derivation of the order terms in the last line of \eqref{eqn:hatQ_orders}.
		For the last summand, we use the stochastic equicontinuity imposed in Assumption~\ref{ass:Normality}~\ref{item:SE}:
		For any $\eta > 0$ and $\varepsilon > 0$, there exists a $\delta > 0$ such that
		\begin{align*}
			&\underset{T \to \infty}{\lim \sup} \; \P \left( \big| \nu_{iT}(\widehat{\mtheta}_{iT}) - \nu_{iT}(\bar{\mtheta}_{i}) \big| > \eta \right) \\
			&\qquad \le \underset{T \to \infty}{\lim \sup} \; \P \left( \big| \nu_{iT}(\widehat{\mtheta}_{iT}) - \nu_{iT}(\bar{\mtheta}_{i}) \big| > \eta, \;\; \Vert \widehat{\mtheta}_{iT} -  \bar\mtheta_i \Vert \le \delta \right)
			+ \underset{T \to \infty}{\lim \sup} \; \P \left(  \Vert \widehat{\mtheta}_{iT} -  \bar\mtheta_i \Vert > \delta \right) \\
			&\qquad \le \underset{T \to \infty}{\lim \sup} \; \P \left( \sup_{ \{\mtheta \in \mTheta: \, \Vert \mtheta - \bar \mtheta_i \Vert \le \delta\} }    \big| \nu_{iT}(\mtheta) - \nu_{iT}(\bar \mtheta_i) \big| > \eta \right)  + o(1) \\
			&\qquad < \varepsilon,
		\end{align*}
		where we have used that $\widehat{\mtheta}_{iT}$ is a consistent estimator for $\bar{\mtheta}_{i}$.
		
		For the first order term in the last line of \eqref{eqn:hatQ_orders}, notice that  $\widetilde{\mtheta}_{iT} \toP  \bar{\mtheta}_{i}$ as $\widetilde{\mtheta}_{iT}$ is between $\bar{\mtheta}_{i}$ and the consistent estimator $\widehat{\mtheta}_{iT}$.
		As by Assumption~\ref{ass:Normality}~\ref{item:ContExpScore}, the map $\mtheta \mapsto g(\mtheta) = \nabla_{\mtheta} \E \big[ \myS( \mW_{it}^\top \mtheta, Y_t) \big]$ is continuous.
		The Continuous Mapping Theorem hence implies that
		\begin{align*}
			\nabla_{\mtheta} \overline{Q}_{iT}(\widetilde{\mtheta}_{iT}) 
			= g(\widetilde{\mtheta}_{iT})
			= \nabla_{\mtheta} \E \big[ \myS( \mW_{it}^\top \bar{\mtheta}_{i}, Y_t) \big]  + o_\P(1)
			= o_\P(1),
		\end{align*}
		as 
		$\nabla_{\mtheta} \E \big[ \myS( \mW_{it}^\top \bar{\mtheta}_{i}, Y_t) \big] = 0$ by \eqref{eqn:PseudoTrue} in Assumption~\ref{ass:TrueLinearSpecification}.
		An equivalent reasoning shows that $\sqrt{T} D_{T} = o_\P(1)$ (by simply setting $\mW_{it} = 1$ in \eqref{eqn:DefQs} and the following.)
		
		Hence, using $\sqrt{T} B_{iT} = o_\P(1)$ and $\sqrt{T} D_{T} = o_\P(1)$, we can write
		\begin{align*}
			\sqrt{T}
			\begin{pmatrix}
				{\MCBs}_{1T} - \MCBp_{1} \\
				{\DSCs}_{1T} - \DSCp_{1} \\
				{\MCBs}_{2T} - \MCBp_{2} \\
				{\DSCs}_{2T} - \DSCp_{2}
			\end{pmatrix}  
			=
			\begin{pmatrix}
				1 & -1  & 0 &  0  & 0 \\
				0 & -1  & 0 &  0  & 1 \\
				0 &  0  & 1 & -1  & 0  \\
				0 &  0  & 0 & -1  & 1  
			\end{pmatrix}
			T^{-1/2} \sum_{t=1}^T \big(a_{1t}, c_{1t}, a_{2t}, c_{2t}, e_{t} \big)^\top + o_\P(1),
		\end{align*}
		and as
		\begin{align*}
			T^{-1/2} \mOmega_T^{-1/2} \sum_{t=1}^T \big(a_{1t}, c_{1t}, a_{2t}, c_{2t}, e_{t} \big)^\top 
			\toD \mathcal{N}(0, I_5)
		\end{align*}
		satisfies a CLT by assumption, the result of the theorem follows.
	\end{proof}

	\begin{proof}[\bf Proof of Theorem \ref{thm:MCBandDSC0}]
		We start to show the distribution of $T\, \MCBs_{iT}$ under $\MCBp_{i} = 0$.
		For this, we have that
		\begin{align}
			\label{eqn:MCB0_equality}
			\E \big[ \myS(X_{it}, Y_t) \big] =  \E \big[ \myS(\mW_{it}^\top \bar \mtheta_i, Y_t)  \big] \qquad \forall t \in \N.
		\end{align}
		Hence, by the uniqueness in \eqref{eqn:PseudoTrue}, this implies that $X_{it} = \mW_{it}^\top \bar \mtheta_i$ a.s. (e.g., with similar arguments as in the proof of part \ref{item:MCBs=0} of Theorem \ref{thm:PositivComponenets}).
		
		Then, by using the definitions from \eqref{eqn:DefQs} and arguments as in \eqref{eqn:hatQ_orders}, we get that
		\begin{align}
			\begin{aligned}
				\label{eqn:BiT_Expansion}
				T \left( \widehat{\MCB_i}_T - \MCBp_{i}  \right)
				&=  \sum_{t=1}^T \left( \myS \big( X_{it}, Y_t \big) - \myS \big( \mW_{it}^\top \widehat{\mtheta}_{iT}, Y_t \big) \right) \\
				&= - \sum_{t=1}^T \left( \myS \big( \mW_{it}^\top \widehat{\mtheta}_{iT}, Y_t \big) - \myS \big( \mW_{it}^\top \bar{\mtheta}_{i}, Y_t \big) \right) \\
				% &= \widehat{Q}_{iT}(\widehat{\mtheta}_{iT}) - \widehat{Q}_{iT}(\bar{\mtheta}_{i}) \\
				&= - \sum_{t=1}^T \myS'(\mW_{it}^\top \bar \mtheta_i, Y_t)  \big( \mW_{it}^\top \widehat{\mtheta}_{iT}  - \mW_{it}^\top \bar \mtheta_i \big) \\
				&\qquad - \frac{1}{2} \sum_{t=1}^T  \myS''(\mW_{it}^\top \bar \mtheta_i, Y_t)  \big(  \mW_{it}^\top \widehat{\mtheta}_{iT} -  \mW_{it}^\top \bar \mtheta_i \big)^2 \\
				&\qquad - \frac{1}{6} \sum_{t=1}^T \myS'''(\mW_{it}^\top  \widetilde{\mtheta}_{iT}, Y_t)  \big( \mW_{it}^\top \widehat{\mtheta}_{iT} - \mW_{it}^\top \bar \mtheta_i \big)^3 \\
				&=: - T B_{iT}^{(1)} - T B_{iT}^{(2)} - T B_{iT}^{(3)},
			\end{aligned} 
		\end{align}
		for some random variable $\widetilde{\mtheta}_{iT}$ between $\widehat{\mtheta}_{iT}$ and $\bar \mtheta_i$ by a second-order Taylor expansion and a mean-value theorem.
		
		In the following, we show that $T B_{iT}^{(3)} = o_\P(1)$, whereas $T B_{iT}^{(1)} = \mathcal{O}_\P(1)$ and $T B_{iT}^{(2)} = \mathcal{O}_\P(1)$ contribute to the asymptotic distribution of $T \big( \widehat{\MCB}_{iT} - \MCBp_{i}  \big)$.
		
		By Assumption~\ref{ass:AsymptoticsDSC0}~\ref{item:BahadurMestJoint}, equation~\eqref{eq:MCB0_mestimators}, it holds that
		\begin{align}
			\label{eq:MestExpansion}
			\sqrt{T} \big(\widehat{\mtheta}_{iT} - \bar \mtheta_i \big) 
			= - \left( \frac{1}{T} \sum_{t=1}^T \myS''(\mW_{it}^\top \bar \mtheta_i, Y_t)  \mW_{it} \mW_{it}^\top\right)^{-1}  \left( T^{-1/2} \sum_{t=1}^T \myS'(\mW_{it}^\top \bar \mtheta_i, Y_t) \mW_{it} \right) + o_\P(1) .
		\end{align}
		Hence, for $B_{iT}^{(1)}$, by using \eqref{eq:MestExpansion}, we get that
		\begin{align*}
			- T B_{iT}^{(1)} 
			&= - T^{-1/2} \sum_{t=1}^T \myS'(\mW_{it}^\top \bar \mtheta_i, Y_t) \mW_{it}^\top \times  \sqrt{T} \big(  \widehat{\mtheta}_{iT} - \bar \mtheta_i \big) \\
			&= \left( T^{-1/2} \sum_{t=1}^T \myS'(\mW_{it}^\top \bar \mtheta_i, Y_t) \mW_{it} \right)^\top \times \left( \frac{1}{T} \sum_{t=1}^T \myS''(\mW_{it}^\top \bar \mtheta_i, Y_t)  \mW_{it} \mW_{it}^\top \right)^{-1}  \\
			&\qquad \times \left( T^{-1/2} \sum_{t=1}^T \myS'(\mW_{it}^\top \bar \mtheta_i, Y_t) \mW_{it} \right) + o_\P(1).
		\end{align*}
		
		\allowdisplaybreaks
		Similarly, for $B_{iT}^{(2)}$, using \eqref{eq:MestExpansion} twice, we get that
		\begin{align*}
			- T B_{iT}^{(2)} 
			&= - \frac{1}{2} \sum_{t=1}^T  \myS''(\mW_{it}^\top \bar \mtheta_i, Y_t)  \big(  \mW_{it}^\top \widehat{\mtheta}_{iT} -  \mW_{it}^\top \bar \mtheta_i \big)^2 \\
			&= - \frac{1}{2} \big( \sqrt{T} ( \widehat{\mtheta}_{iT} - \bar \mtheta_i) \big)^\top \times \left( \frac{1}{T} \sum_{t=1}^T \myS''(\mW_{it}^\top \bar \mtheta_i, Y_t)  \mW_{it} \mW_{it}^\top \right) \times  \big( \sqrt{T} ( \widehat{\mtheta}_{iT} - \bar \mtheta_i) \big) \\
			&= - \frac{1}{2} \left( T^{-1/2} \sum_{t=1}^T \myS'(\mW_{it}^\top \bar \mtheta_i, Y_t) \mW_{it} \right)^\top \times \left( \frac{1}{T} \sum_{t=1}^T \myS''(\mW_{it}^\top \bar \mtheta_i, Y_t)  \mW_{it} \mW_{it}^\top\right)^{-1}  \\
			&\qquad \times \left( T^{-1/2} \sum_{t=1}^T \myS'(\mW_{it}^\top \bar \mtheta_i, Y_t) \mW_{it} \right)  + o_\P(1).
		\end{align*}
		
		Finally, for the absolute value of $T B_{iT}^{(3)}$, we get that
		\begin{align*}
			\big| T B_{iT}^{(3)} \big| &=  \left| \frac{1}{6} \sum_{t=1}^T \myS'''(\mW_{it}^\top  \widetilde{\mtheta}_{iT}, Y_t)  \big( \mW_{it}^\top \widehat{\mtheta}_{iT} - \mW_{it}^\top \bar \mtheta_i \big)^3  \right| \\
			&\le \frac{1}{6} \sum_{t=1}^T \big| \myS'''(\mW_{it}^\top  \widetilde{\mtheta}_{iT}, Y_t) \big|  \times \big|  \mW_{it}^\top \widehat{\mtheta}_{iT} - \mW_{it}^\top \bar \mtheta_i \big|^3  \\
			&\le \frac{1}{6} T^{-1/2} \frac{1}{T} \sum_{t=1}^T \big| \myS'''(\mW_{it}^\top  \widetilde{\mtheta}_{iT}, Y_t) \big| \Vert \mW_{it} \Vert^3 \times \Vert \sqrt{T} (\widehat{\mtheta}_{iT} - \bar \mtheta_i ) \Vert^3 \\
			&= {o}_\P(1) \times \mathcal{O}_\P(1) =  {o}_\P(1),
		\end{align*}
		\sloppy
		where we now provide details for the first ${o}_\P(1)$ term.
		For $h_T(\mtheta) := \frac{1}{T} \sum_{t=1}^T \big| \myS'''(\mW_{it}^\top \mtheta, Y_t) \big| \Vert \mW_{it} \Vert^3$, we have that for any $\varepsilon > 0$ small enough
		\begin{align}
			\begin{aligned}
				\label{eqn:ProofOP1CMT}
				&\P \left( T^{-1/2} h_T(\widetilde{\mtheta}_{iT}) > \varepsilon \right) \\
				&\qquad\le \P \left( T^{-1/2} \big| h_T(\widetilde{\mtheta}_{iT}) - h_T(\bar{\mtheta}_{i}) \big| > \varepsilon/2 \right) 
				+ \P \left( T^{-1/2} h_T(\bar{\mtheta}_{i}) > \varepsilon/2 \right) \\
				&\qquad\le \P \left(  T^{-1/2} \sup_{\Vert \mtheta - \bar\mtheta_i \Vert \le \delta}  \big| h_T(\mtheta ) - h_T(\bar{\mtheta}_{i}) \big| > \varepsilon/2 \right) 
				+ \P \left( T^{-1/2} h_T(\bar{\mtheta}_{i}) > \varepsilon/2 \right) 
				+ \P \left( \Vert \widetilde{\mtheta}_{iT} - \bar \mtheta_i \Vert > \delta \right) \\ 
				&\qquad= o(1),
			\end{aligned}
		\end{align}
		where the first two terms above are $o(1)$ by Assumption~\ref{ass:AsymptoticsDSC0}~\ref{item:SmoothSDSC} and the third is $o(1)$ by the consistency of $\widetilde{\mtheta}_{iT}$.

		Combining the previous three terms, we get that
		\begin{align*}
			T \big( \widehat{\MCB}_{iT} - {\MCBp_{i}}  \big)  	
			&= \frac{1}{2} \left( T^{-1/2} \sum_{t=1}^T \myS'(\mW_{it}^\top \bar \mtheta_i, Y_t) \mW_{it} \right)^\top \times \left( \frac{1}{T} \sum_{t=1}^T \myS''(\mW_{it}^\top \bar \mtheta_i, Y_t)  \mW_{it} \mW_{it}^\top \right)^{-1}  \\
			&\qquad \times \left( T^{-1/2} \sum_{t=1}^T \myS'(\mW_{it}^\top \bar \mtheta_i, Y_t) \mW_{it} \right) + o_\P(1).
		\end{align*}
		Then, given a Gaussian random variable $\bm N_{i} \sim \mathcal{N}(0, \bm\Pi_{i})$, the quantity of interest $T \big( \widehat{\MCB}_{iT} - {\MCBp_{i}} \big)$ is approximated in distribution by
		\begin{align*}
			T \left( \widehat{\MCB}_{iT} - \MCBp_{i}  \right)
			\toD \frac{1}{2} \bm N_{i}^\top \bm\Upsilon_{i}^{-1} \bm N_{i},
		\end{align*}
		which concludes the first part of the proof.

		We continue with the second part, the asymptotic distribution of $T \, \widehat{\DSC}_{iT}$ if $\DSCp_{i} = 0$.
		We then have that
		\begin{align*}
			\E \big[ \myS(\bar r, Y_t) \big] =  \E \big[ \myS(\mW_{it}^\top \bar \mtheta_i, Y_t)  \big] \qquad \forall t \in \N.
		\end{align*}
		By the uniqueness  in \eqref{eqn:PseudoTrue}, we obtain $\bar r = \mW_{it}^\top \bar \mtheta_i$ a.s.\ for all $t \in \N$, which will be used below.
		Hence, we have that
		\begin{align*}
			T \big( {\DSCs}_{iT} - \DSCp_{i} \big)
			&= \sum_{t=1}^T \big( \myS(\widehat{r}_T, Y_t) -  \myS(\bar r,Y_t) \big) +    \sum_{t=1}^T \left(\myS( \mW_{it}^\top \bar\mtheta_i, Y_t)  - \myS (\mW_{it}^\top \widehat{\mtheta}_{iT} , Y_t ) \right) \\
			& =  T D_{T} - T B_{iT}.
		\end{align*}
		A third order Taylor expansion of $T D_{T}$ gives that for some $\widetilde{r}_{T}$ on the line between $\widehat{r}_T$ and $\bar{r}$,
		\begin{align*}
			T D_T 
			&= \sum_{t=1}^T \myS (\widehat{r}_T, Y_t ) - \myS(\bar{r},Y_t) \\ 
			&= \sum_{t=1}^T \myS'(\bar{r}, Y_t)  \big( \widehat{r}_T  - \bar{r} \big)  
			+ \frac{1}{2} \sum_{t=1}^T \myS''(\bar{r}, Y_t)  \big( \widehat{r}_T - \bar{r} \big)^2 
			+ \frac{1}{6} \sum_{t=1}^T \myS'''(\widetilde{r}_{T}, Y_t)  \big( \widehat{r}_T - \bar{r} \big)^3 \\
			&=:	T D_{T}^{(1)} + T D_{T}^{(2)} + T D_{T}^{(3)}. 
		\end{align*}
		Together with the expansion of $B_{iT}$ in \eqref{eqn:BiT_Expansion}, this implies
		\begin{align}
			\label{eqn:DSCdecompProof}
			T \big( {\DSCs}_{iT} - \DSCp_{i} \big) 
			&= T \big( D_{T}^{(1)}  - B_{iT}^{(1)} \big) +  T \big( D_{T}^{(2)}  - B_{iT}^{(2)} \big) +  T \big( D_{T}^{(3)}  - B_{iT}^{(3)} \big).
		\end{align}  
		
		For the vector consisting of the first terms in \eqref{eqn:DSCdecompProof}, using \eqref{eq:comb_mestimators} in Assumption~\ref{ass:AsymptoticsDSC0}~\ref{item:BahadurMestJoint}, we get
		\allowdisplaybreaks
		\begin{align} 
			T \big( D_{T}^{(1)}  - B_{iT}^{(1)} \big)
			&= \begin{pmatrix} 1 & -1 \end{pmatrix}  \times 
			\begin{pmatrix}
				T D_{T}^{(1)}  \\
				T B_{iT}^{(1)}
			\end{pmatrix} \nonumber \\
			&= \begin{pmatrix} 1 & -1 \end{pmatrix}  \times 	
			\sum_{t=1}^T 
			\begin{pmatrix}
				\myS'(\bar{r}, Y_t) (\widehat{r}_T -\bar{r}) \\
				\myS'(\bar{r}, Y_t) \mW_{it}^\top (\widehat{\mtheta}_{iT} -\bar{\mtheta}_{i}) \\
			\end{pmatrix}\nonumber \\
			% new equation
			&=  \begin{pmatrix} 1 & -1 \end{pmatrix} 	 \times 
			T^{-1/2} \sum_{t=1}^T \begin{pmatrix}
				\myS'(\bar r,Y_t) & \bm{0} \\
				0 & \myS'(\bar r,Y_t)\mW_{it}^\top 
			\end{pmatrix} \times 	\sqrt{T} 
			\begin{pmatrix}
				\widehat{r}_T -\bar{r} \\
				\widehat{\mtheta}_{iT} -\bar{\mtheta}_{i} \\
			\end{pmatrix}\nonumber\\
			% new equation
			&= - \begin{pmatrix} 1 & -1 \end{pmatrix} \times 
			T^{-1/2} \sum_{t=1}^T 
			\begin{pmatrix}
				\myS'(\bar r,Y_t) & \bm{0} \\
				0 & \myS'(\bar r,Y_t)\mW_{it}^\top
			\end{pmatrix} 	
			\nonumber \\ 
			&\qquad    \times		
			\begin{pmatrix}
				\left( \frac{1}{T} \sum_{t=1}^T \myS''( \bar r, Y_t) \right)^{-1}  + o_\P(1) & \boldsymbol{0} \\
				\multicolumn{2}{c}{\left( \frac{1}{T} \sum_{t=1}^T \myS''(\bar r, Y_t)   \mW_{it} \mW_{it}^\top\right)^{-1} + o_\P(1)}
			\end{pmatrix} 
			\times T^{-1/2} \sum_{t=1}^T \myS'(\bar r, Y_t) \mW_{it} \nonumber \\
			% new equation
			&= - T^{-1/2} \sum_{t=1}^T 
			\begin{pmatrix}
				\myS'(\bar r,Y_t), & - \myS'(\bar{r},Y_t)\mW_{it}^\top
			\end{pmatrix} 	
			\nonumber \\ 
			&\qquad    \times		
			\begin{pmatrix}
				\left( \frac{1}{T} \sum_{t=1}^T \myS''( \bar r, Y_t) \right)^{-1}  & \boldsymbol{0} \\
				\multicolumn{2}{c}{\left( \frac{1}{T} \sum_{t=1}^T \myS''(\bar r, Y_t)   \mW_{it} \mW_{it}^\top\right)^{-1}}
			\end{pmatrix} 
			\times T^{-1/2} \sum_{t=1}^T \myS'(\bar r, Y_t) \mW_{it} + o_\P(1) \nonumber \\
			% new equation
			&= \left( T^{-1/2} \sum_{t=1}^T \myS'(\bar r, Y_t) \mW_{it}^\top \right) \\
			&\qquad \times \left( \left(  \frac{1}{T} \sum_{t=1}^T \myS''(\bar r, Y_t)\mW_{it}\mW_{it}^\top \right)^{-1}  -\begin{pmatrix}
				\left(  \frac{1}{T} \sum_{t=1}^T \myS''(\bar{r}, Y_t)\right)^{-1} & \boldsymbol{0} \\ \boldsymbol{0} & \boldsymbol{0}
			\end{pmatrix} \right) \nonumber \\
			&\qquad \times \left( T^{-1/2} \sum_{t=1}^T \myS'(\bar r, Y_t) \mW_{it} \right) + o_\P(1). \nonumber
		\end{align}
		In the last step, we have used the fact that for any (column vector) $\boldsymbol{a} \in \mathbb{R}^k$ and matrix $\boldsymbol{B} \in \mathbb{R}^{k \times k}$ with first/upper-left entries  $a_1, B_{11} \in \mathbb{R}$, respectively, it holds that
		\begin{align}
			\label{eqn:FormulaWeirdOuterProduct}
			- \begin{pmatrix}
				a_1 & -\boldsymbol{a}^\top
			\end{pmatrix}
			\times 
			\begin{pmatrix}
				B_{11} \;\; \boldsymbol{0} \\
				\mathbf{B}
			\end{pmatrix}
			\times \boldsymbol{a}
			= 
			\boldsymbol{a}^\top 
			\times 
			\left( \mathbf{B} - 
			\begin{pmatrix}
				B_{11} & \boldsymbol{0} \\
				\boldsymbol{0} & \boldsymbol{0}
			\end{pmatrix}
			\right)
			\times \boldsymbol{a},
		\end{align}
		where the $\boldsymbol{0}$ entries have the appropriate dimensions.
		This reformulation allows to express the term above as a quadratic form.
		
		For the second terms in \eqref{eqn:DSCdecompProof}, similar derivations yield that
		\begin{align*} 
			T &\big( D_{T}^{(2)}  - B_{iT}^{(2)} \big) \\
			&= \begin{pmatrix} 1 & -1 \end{pmatrix}  \times 
			\begin{pmatrix}
				T D_{T}^{(2)}  \\
				T B_{iT}^{(2)}
			\end{pmatrix} \\
			&= \frac{1}{2} \begin{pmatrix} 1 & -1 \end{pmatrix}  \times  \sum_{t=1}
			\begin{pmatrix}
				\myS''(\bar{r}, Y_t) \big( \widehat{r}_T-\bar{r} \big)^2 \\
				\myS''(\bar{r}, Y_t)  \big( \mW_{it}^\top\widehat{\mtheta}_{iT}-\mW_{it}^\top\bar{\mtheta}_i \big)^2
			\end{pmatrix} \nonumber \\
			&= \frac{1}{2} \begin{pmatrix} 1 & -1 \end{pmatrix}  \times 
			\begin{pmatrix}
				\left(\sqrt{T} \left(\widehat{r}_T - \bar{r}\right) \right) \times \left( \frac{1}{T} \sum_{t=1}^{T} \myS'' (\bar{r}, Y_t)\right) 	\times \left(\sqrt{T} \left(\widehat{r}_T - \bar{r}\right)\right) \\
				\left(\sqrt{T} \left(  \widehat{\mtheta}_{iT} - \bar{\mtheta}_i  \right)\right)^\top \times  \left( \frac{1}{T} \sum_{t=1}^T  \myS''(\bar{r}, Y_t) \mW_{it}\mW_{it}^\top \right) \times 	\left(\sqrt{T} \left(  \widehat{\mtheta}_{iT} - \bar{\mtheta}_i  \right)\right)
			\end{pmatrix} \\
			&= 
			\frac{1}{2} \begin{pmatrix} 1 & -1 \end{pmatrix}  
			\times 
			\begin{pmatrix}
				\sqrt{T} \big(\widehat{r}_T - \bar{r}\big) & \boldsymbol{0} \\
				0 & \sqrt{T} \big(\widehat{\mtheta}_{iT} - \bar{\mtheta}_i  \big)^\top
			\end{pmatrix}
			\\
			&\qquad \times 
			\begin{pmatrix}
				\frac{1}{T} \sum_{t=1}^{T} \myS'' (\bar{r}, Y_t) & \boldsymbol{0} \\
				\boldsymbol{0} & \frac{1}{T} \sum_{t=1}^T  \myS''(\bar{r}, Y_t) \mW_{it}\mW_{it}^\top
			\end{pmatrix}
			\times 
			\begin{pmatrix}
				\sqrt{T} \big(\widehat{r}_T - \bar{r}\big) \\
				\sqrt{T} \big(\widehat{\mtheta}_{iT} - \bar{\mtheta}_i  \big)
			\end{pmatrix}  \\
			% new equation
			&= 
			\frac{1}{2} 
			\times 
			\begin{pmatrix}
				\sqrt{T} \big( \widehat{r}_T - \bar{r}\big) \\
				\sqrt{T} \big( \widehat{\mtheta}_{iT} - \bar{\mtheta}_i  \big)
			\end{pmatrix}^\top
			\times 
			\begin{pmatrix}
				\frac{1}{T} \sum_{t=1}^{T} \myS'' (\bar{r}, Y_t) & \boldsymbol{0} \\
				\boldsymbol{0} & - \frac{1}{T} \sum_{t=1}^T  \myS''(\bar{r}, Y_t) \mW_{it}\mW_{it}^\top
			\end{pmatrix}
			\times 
			\begin{pmatrix}
				\sqrt{T} \big(\widehat{r}_T - \bar{r}\big) \\
				\sqrt{T} \big(\widehat{\mtheta}_{iT} - \bar{\mtheta}_i  \big)
			\end{pmatrix} \\
			% new equation
			&= 
			\frac{1}{2} 
			T^{-1/2} \sum_{t=1}^T \myS'(\bar r, Y_t) \mW_{it}^\top
			\times 
			\begin{pmatrix}
				\left( \frac{1}{T} \sum_{t=1}^T \myS''( \bar r, Y_t) \right)^{-1}   & \boldsymbol{0} \\
				\multicolumn{2}{c}{\left( \frac{1}{T} \sum_{t=1}^T \myS''(\bar{r}, Y_t)   \mW_{it} \mW_{it}^\top\right)^{-1}}
			\end{pmatrix}^\top
			\\ 
			&\qquad \times 
			\begin{pmatrix}
				\frac{1}{T} \sum_{t=1}^{T} \myS'' (\bar{r}, Y_t) & \boldsymbol{0} \\
				\boldsymbol{0} & - \frac{1}{T} \sum_{t=1}^T  \myS''(\bar{r}, Y_t) \mW_{it}\mW_{it}^\top
			\end{pmatrix} \\
			&\qquad \times 
			\begin{pmatrix}
				\left( \frac{1}{T} \sum_{t=1}^T \myS''( \bar r, Y_t) \right)^{-1}   & \boldsymbol{0} \\
				\multicolumn{2}{c}{\left( \frac{1}{T} \sum_{t=1}^T \myS''(\bar{r}, Y_t)   \mW_{it} \mW_{it}^\top\right)^{-1}}
			\end{pmatrix} 
			\times 
			T^{-1/2} \sum_{t=1}^T \myS'(\bar r, Y_t) \mW_{it}  + o_\P(1) \\
			% new equation
			&= - \frac{1}{2} \left( T^{-1/2} \sum_{t=1}^T \myS'(\bar r, Y_t) \mW_{it}^\top \right) \\
			&\qquad \times \left( \left(  \frac{1}{T} \sum_{t=1}^T \myS''(\bar{r}, Y_t)\mW_{it}\mW_{it}^\top \right)^{-1}  -\begin{pmatrix}
				\left(  \frac{1}{T} \sum_{t=1}^T \myS''(\bar{r}, Y_t)\right)^{-1} & \boldsymbol{0} \\ \boldsymbol{0} & \boldsymbol{0}
			\end{pmatrix} \right) \nonumber \\
			&\qquad \times \left( T^{-1/2} \sum_{t=1}^T \myS'(\bar r, Y_t) \mW_{it}^\top \right) + o_\P(1). \nonumber
		\end{align*}
		In the penultimate equality, we have plugged in the (joint) asymptotic expansion of the estimators from \eqref{eq:comb_mestimators}.
		In the last equality above, we have used the fact that for a (column) vector  ${\boldsymbol{b}} = (b_{1}, \; \boldsymbol{0})^\top \in \mathbb{R}^{k}$ with first entry $b_{1} \in \mathbb{R}$, $b_{1} \not= 0$ and an invertible matrix $\boldsymbol{B} \in \mathbb{R}^{k \times k}$, we have that 
		\begin{align*}
			%\label{eqn:FormulaWeirdOuterProduct2}
			\begin{pmatrix}
				{\boldsymbol{b}} & \boldsymbol{B}
			\end{pmatrix}
			\times 
			\begin{pmatrix}
				{b}_{1}^{-1} & \boldsymbol{0} \\
				\boldsymbol{0} & - \boldsymbol{B}^{-1}
			\end{pmatrix}
			\times
			\begin{pmatrix}
				{\boldsymbol{b}^\top} \\ \boldsymbol{B}
			\end{pmatrix}
			= 
			\begin{pmatrix}
				{b}_{1}^{-1} {\boldsymbol{b}} & -\boldsymbol{I}_k
			\end{pmatrix}
			\times
			\begin{pmatrix}
				{\boldsymbol{b}}^\top \\ \boldsymbol{B}
			\end{pmatrix}
			= 
			{b}_{1}^{-1}  {\boldsymbol{b}} {\boldsymbol{b}}^\top - \boldsymbol{B}= 
			\begin{pmatrix}
				{b}_{1} & \boldsymbol{0} \\
				\boldsymbol{0} & \boldsymbol{0}
			\end{pmatrix} - \boldsymbol{B}.
		\end{align*}
		
		For the third term in \eqref{eqn:DSCdecompProof}, we get that
		\begin{align}
			\left| \begin{pmatrix} T D_{T}^{(3)} \\ T B_{iT}^{(3)} \end{pmatrix}  \right| 
			&= \dfrac{1}{6} 
			\begin{pmatrix}
				\left|   \sum_{t=1}^T \myS''' (\widetilde{r}_{T}, Y_t) (\widehat{r}_T - \bar{r} )^3 \right| \\
				\left| \sum_{t=1}^T \myS''' (\mW_{it}^\top  \widetilde{\mtheta}_{iT}, Y_t) (\mW_{it}^\top \widehat{\mtheta}_{iT} - \mW_{it}^\top \bar{\mtheta}_i )^3 \right| 
			\end{pmatrix} \nonumber \\
			&\leq \dfrac{1}{6} 
			\begin{pmatrix}
				\sum_{t=1}^T \left|  \myS'''(\widetilde{r}_{T}, Y_t) \right| \cdot \left|  \widehat{r}_t - \bar{r} \right|^3 \\
				\sum_{t=1}^T \left|  \myS'''(\mW_{it}^\top \widetilde{\mtheta}_{iT}, Y_t) \right| \cdot \left|  \mW_{it}^\top \widehat{\mtheta}_{iT} - \mW_{it}^\top \bar{\mtheta}_i \right|^3 \\
			\end{pmatrix} \nonumber \\
			&\leq \dfrac{1}{6} T^{-1/2} 
			\left\{ \dfrac{1}{T} \sum_{t=1}^T
			\begin{pmatrix}
				| \myS'''(\widetilde{r}_{T}, Y_t) | & 0   \\
				0 & | \myS'''(\widetilde{\mtheta}_{iT}, Y_t) | \cdot \Vert \mW_{it}^\top \Vert^3 
			\end{pmatrix} \right\}  
			\begin{pmatrix}
				\big( \sqrt{T}  \Vert \widehat{r}_t - \bar{r} \Vert \big)^3 \\
				\big( \sqrt{T}  \big\Vert \widehat{\mtheta}_{iT} - \bar{\mtheta}_i \big\Vert \big)^3
			\end{pmatrix} \nonumber \\
			&= o_{\P}(1)  \times \begin{pmatrix}
				\mathcal{O}_{\P}(1) \\   \mathcal{O}_{\P}(1)
			\end{pmatrix} \nonumber \\
			&= o_\P(1),
		\end{align}
		where the proof of the $o_\P(1)$-term follows analogously to \eqref{eqn:ProofOP1CMT}.
		Notice for this that as $\mW_{it}$ contains a constant by assumption,    the procedure in  \eqref{eqn:ProofOP1CMT} together with Assumption~\ref{ass:AsymptoticsDSC0}~\ref{item:SmoothSDSC} also implies that $T^{-1/2} \frac{1}{T} \sum_{t=1}^T |\myS'''(\widetilde{r}_{T}, Y_t) | = o_{\P}(1)$.
		
		Hence, from the above three derivations together with \eqref{eqn:DSCdecompProof}, we obtain 
		\begin{align*}
			T \left( {\DSCs}_{iT} - \DSCp_{i} \right) 
			&= \frac{1}{2} \left( T^{-1/2} \sum_{t=1}^T \myS'(\bar r, Y_t) \mW_{it}^\top \right) \\ 
			&\qquad \times \left( \left(  \frac{1}{T} \sum_{t=1}^T \myS''(\bar{r}, Y_t)\mW_{it}\mW_{it}^\top \right)^{-1}  -\begin{pmatrix}
				\left(  \frac{1}{T} \sum_{t=1}^T \myS''(\bar{r}, Y_t)\right)^{-1} & \boldsymbol{0} \\ \boldsymbol{0} & \boldsymbol{0}
			\end{pmatrix} \right) \\
			&\qquad \times \left( T^{-1/2} \sum_{t=1}^T \myS'(\bar r, Y_t) \mW_{it} \right) + o_\P(1), 
		\end{align*}
		such that the claim of the theorem follows by applying a CLT to the outer terms and a law of large numbers to the inner one.
	\end{proof}

	\begin{proof}[\bf Proof of Proposition \ref{prop:verification_phi}] 
		
		After laying out some preliminaries, we first show the statement on Assumption~\ref{ass:Normality}, followed by the second statement on the verification of Assumption~\ref{ass:AsymptoticsDSC0}.
		
		Notice that $\phi(Y_t)$  and $\phi(\mW_{it}^\top\mtheta)$ are measurable functions of $\a$-mixing processes of size $-\size/(\size-2),\size>2$, and thus itself mixing of the same size by \citet[Theorem 3.49]{W_2001} together with Assumption \ref{ass:Philowlevel}\ref{ass:itemlowlevePhimixing}. Next, consider the objective functions
		\begin{align*}
			\widehat{Q}_{iT}(\mtheta) &:= \frac{1}{T} \sum_{t=1}^T \myS(\mW_{it}^\top \mtheta,Y_t)  
			\qquad \text{ and } \qquad \\
			\overline{Q}_{iT}(\mtheta) &:= \E \left[ \widehat{Q}_{iT}(\mtheta) \right] = \E \left[\frac{1}{T} \sum_{t=1}^T  \myS(\mW_{it}^\top \mtheta,Y_t) \right]
			= \E \left[\myS(\mW_{it}^\top \mtheta,Y_t) \right]. 
		\end{align*}
		The first and second derivatives of the objective functions with respect to $\mtheta$ are
		\begin{align*}
			\nabla_{\mtheta}\widehat{Q}_{iT}(\mtheta) &:= \frac{1}{T} \sum_{t=1}^T   \phi''(\mW_{it}^\top\mtheta) \mW_{it}(\mW_{it}^\top\mtheta - Y_t), \\
			\nabla_{\mtheta}\overline{Q}_{iT}(\mtheta) &:= \E\left[ \nabla_{\mtheta}\widehat{Q}_{iT}(\mtheta) \right] = \E \left[ \phi''(\mW_{it}^\top\mtheta) \mW_{it}(\mW_{it}^\top\mtheta - Y_t)  \right] \\
			\nabla_{\mtheta\mtheta}\widehat{Q}_{iT}(\mtheta) &:= \frac{1}{T} \sum_{t=1}^T \phi'''(\mW_{it}^\top\mtheta)\mW_{it}\mW_{it}^\top(\mW_{it}^\top\mtheta - Y_t) + \phi''(\mW_{it}^\top\mtheta)\mW_{it}\mW_{it}^\top, \\
			\nabla_{\mtheta\mtheta}\overline{Q}_{iT}(\mtheta) &:= \E\left[ \nabla_{\mtheta\mtheta}\widehat{Q}_{iT}(\mtheta) \right] = \E \left[   \phi'''(\mW_{it}^\top\mtheta)\mW_{it}\mW_{it}^\top(\mW_{it}^\top\mtheta - Y_t) + \phi''(\mW_{it}^\top\mtheta)\mW_{it}\mW_{it}^\top\right].
		\end{align*}
		Notice, when working with the derivatives of the objective function  at the population level, we interchange differentiation and expectation, i.e., $\E\left[ \nabla_{\mtheta}\widehat{Q}_{iT}(\mtheta) \right]=  \nabla_{\mtheta}\E\left[\widehat{Q}_{iT}(\mtheta) \right]$,  using a measure-theoretic version of the Leibniz rule,  which requires the derivative to exist and to be absolutely bounded  by an integrable function independent of $\mtheta$.
		For $\nabla_{\mtheta}\overline{Q}_{iT}(\mtheta) $ and $\nabla_{\mtheta\mtheta}\overline{Q}_{iT}(\mtheta) $, this is ensured by 
		\begin{align*}
			d_1(\mW_{it}, Y_t) &= \sup_{\mtheta\in\mTheta} \left\Vert  \phi''(\mW_{it}^\top\mtheta) \mW_{it}(\mW_{it}^\top\mtheta - Y_t)  \right\Vert \quad\text{and} \\
			d_2(\mW_{it}, Y_t)  &= \sup_{\mtheta\in\mTheta} \left\Vert  \phi'''(\mW_{it}^\top\mtheta)\mW_{it}\mW_{it}^\top(\mW_{it}^\top\mtheta - Y_t) + \phi''(\mW_{it}^\top\mtheta)\mW_{it}\mW_{it}^\top \right\Vert ,
		\end{align*}
		together with the moment conditions in Assumption~\ref{ass:Philowlevel}\ref{ass:itemlowlevePhiadditionalmoms}.

		Note, throughout the proof, the $L_s$ vector norm is denoted by $\| \cdot \|_s$ for $s > 0$. In particular, for $s = 2$, we denote the Euclidean norm by $\| \cdot \| := \| \cdot \|_2$.
		We proceed to verify the high-level assumptions from Section \ref{sec:estimation}, following the order in which they are presented.

		\bigskip
		\noindent
		\textbf{Verification of Assumption~\ref{ass:Normality}:} \\
		For Assumption \ref{ass:Normality}\ref{item:EstOP1}, we start to show consistency and convergence of the M-estimators using Theorems 2.1 and 3.1 of \citet{NM_1994}. This directly implies the statement in  high level Assumption \ref{ass:Normality}\ref{item:EstOP1}. 
		
		We start by confirming conditions (i)--(iv) of \citet[Theorem 2.1]{NM_1994} for the parameter $\mtheta$. 
		Condition (i) of \citet[Theorem 2.1]{NM_1994} holds as $\bar\mtheta_i$ is the unique minimizer of $\overline{Q}_{iT}(\mtheta)$ by Assumption~\ref{ass:TrueLinearSpecification}.
		
		Condition (ii) of \citet[Theorem 2.1]{NM_1994}, i.e., compactness of $\mTheta$, holds by imposing Assumption \ref{ass:TrueLinearSpecification} and condition (iii) requires continuity of the objective function $\overline{Q}_{iT}(\mtheta)$ which is fulfilled as per Assumption \ref{ass:Philowlevel}\ref{ass:itemlowlevePhicont}.
		
		Condition (iv) of \citet[Theorem 2.1]{NM_1994} requires $ \sup_{ \mtheta \in \mTheta}  \left| \widehat{Q}_{iT}(\mtheta) -  \overline{Q}_{iT}(\mtheta)\right| = o_\P(1)$ which follows from \citet[Theorem 21.9]{Davidson1994} given that
		\begin{enumerate}
			\item[(a)] $\widehat{Q}_{iT}(\mtheta) -  \overline{Q}_{iT}(\mtheta) = o_\P(1)$ for all $\mtheta \in \mTheta$, and 
			\item[(b)] $\big\{ \widehat{Q}_{iT}(\cdot) -  \overline{Q}_{iT}(\cdot) \big\}$ is stochastically equicontinuous.
		\end{enumerate} 
		We first prove the point-wise law of large numbers for the objective function in (a).
		From \citet[Theorem 3.49]{W_2001}, it follows that $ \myS(\mW_{it}^\top \mtheta,Y_t)$ is stationary and $\alpha$-mixing of size $-\size/(\size-2),\; \size >2,$ by Assumptions \ref{ass:TrueLinearSpecification} and \ref{ass:Philowlevel}\ref{ass:itemlowlevePhimixing}. 
		Further, as of \citet[Proposition 3.44]{W_2001}, $ \myS(\mW_{it}^\top \mtheta,Y_t)  $ is ergodic because it is a measurable function of $(Y_t, \mW_{it}^\top)$. 
		Since 
		\begin{align}
			\E\left\vert \myS(\mW_{it}^\top \mtheta,Y_t) \right\vert &=   \E\left\vert  \phi(Y_t) - \phi(\mW_{it}^\top \mtheta) - \phi'(\mW_{it}^\top \mtheta)(Y_t-\mW_{it}^\top \mtheta) \right\vert \nonumber \\
			\label{eqn:phiobjectivetrianlge}
			&\leq   \E\left\vert\phi\left(Y_t\right) \right\vert +   \E\left\vert\phi\left(\mW_{it}^\top \mtheta\right) \right\vert +   \E \left\vert\phi' \left(\mW_{it}^\top \mtheta\right)  \left(Y_t- \mW_{it}^\top \mtheta \right)\right\vert < \infty, 
		\end{align}
		the ergodicity Theorem 3.34 in \citet{W_2001} is verified and the pointwise law of large numbers in (a) follows.

		Concerning condition (b) of \citet[Theorem 21.9]{Davidson1994}, we show that 
		\[
		\left| \frac{1}{T} \sum_{t=1}^T \myS\left(\mW_{it}^\top \mtheta,Y_t\right) -  \myS\left(\mW_{it}^\top \widetilde\mtheta,Y_t\right)  \right| \leq B_T \left\Vert \mtheta - \widetilde\mtheta \right\Vert \quad\text{for all ~} \mtheta,\widetilde{\mtheta}\in \mTheta, 
		\]
		where $B_T>0$ is a stochastic sequence independent of $\mtheta$ and $\widetilde{\mtheta}$, for which we have $B_T=\mathcal{O}_\P(1)$. 
		By the mean value theorem, which can be applied due to convexity of $\mTheta$ as per Assumption \ref{ass:TrueLinearSpecification}, together with triangle inequality, we get for some mean value $\mtheta^*$ on the line joining $\widetilde\mtheta$ and $\mtheta$, that
		\begin{align} \label{eqn:SEphi}
			\begin{split}
				&\left| \frac{1}{T}\sum_{t=1}^T \myS\left(\mW_{it}^\top \mtheta,Y_t\right) -  \frac{1}{T}\sum_{t=1}^T \myS\left(\mW_{it}^\top \widetilde\mtheta,Y_t\right) \right| 
				= 
				\left|  \frac{1}{T}\sum_{t=1}^T \phi''(\mW_{it}^\top\mtheta^*) \mW_{it}^\top (Y_t-\mW_{it}^\top\mtheta^*) \left( \mtheta - \widetilde\mtheta\right)\right| 
				\\
				&\leq \frac{1}{T}\sum_{t=1}^T \left\vert\phi''(\mW_{it}^\top\mtheta^*) \right\vert \cdot \left\Vert\mW_{it}\right\Vert \cdot \left| \mW^\top_{it}\mtheta^*-Y_t \right| \cdot \left\Vert \mtheta - \widetilde\mtheta\right\Vert \\
				&\leq B_T \left\Vert \mtheta - \widetilde\mtheta\right\Vert,
			\end{split}
		\end{align}
		where
		\begin{align}
			B_T :=  \frac{1}{T}\sum_{t=1}^T \sup_{\mtheta \in \mTheta}\left\vert\phi''(\mW_{it}^\top\mtheta) \right\vert \left\Vert\mW_{it}\right\Vert \cdot \left| \mW_{it}^\top \mtheta - Y_t \right|.
		\end{align}
		Then, $B_T = \mathcal{O}_\P(1)$ follows from \citet[Theorem 12.11]{Davidson1994} as
		$\sup_{T\in\N} \mathbb{E}[B_T] < \infty$ is implied by Assumption \ref{ass:Philowlevel}\ref{ass:itemlowlevePhiadditionalmoms}.
		It follows that $\big\{ \widehat{Q}_{iT}(\cdot) -  \overline{Q}_{iT}(\cdot) \big\}$ is stochastically equicontinuous.
		Hence, we can apply Theorem 2.1 of \citet{NM_1994} to conclude that the M-estimator is consistent, i.e.\ $\widehat{\mtheta}_{iT} \toP \bar{\mtheta}_i$.

		Secondly, we check conditions (i)--(v) of \citet[Theorem 3.1]{NM_1994}.  
		Since $\bar\mtheta_i\in\interior(\mTheta)$ by Assumption \ref{ass:TrueLinearSpecification}, condition (i) holds. 
		By Assumption \ref{ass:Philowlevel}\ref{ass:itemlowlevePhicont}, the objective function $\widehat{Q}_{iT}(\mtheta)$ is twice continuously differentiable, ensuring that condition (ii) is fulfilled.
		To verify condition (iii), we have to show that a central limit theorem holds for $\sqrt{T} \nabla_{\mtheta}\widehat{Q}_{iT}(\bar\mtheta_{i})$. 
		By the Cramér–Wold device, this reduces to show univariate convergence, i.e.,
		\begin{align} 
			\label{eqn:Phiunivconv}
			T^{-1/2} \mlambda^\top \mathbf{V}_T^{-1/2} \sum_{t=1}^T \phi''\left(\mW_{it}^\top\bar\mtheta_i\right) \mW_{it}\left( Y_t  - \mW^\top_{it}\bar\mtheta_i \right) \toD \mathcal{N}(0,1), \qquad T\to\infty,
		\end{align} for any $\mlambda\in\R^k$ satisfying $\mlambda^\top\mlambda=1$. 
		We show this by verifying the required conditions of the CLT in \citet[Theorem 5.20]{W_2001}.
		Note that by \citet[Theorem 3.49]{W_2001} and Assumption \ref{ass:Philowlevel}\ref{ass:itemlowlevePhimixing}, the sequences $\left\{ \phi''\left(\mW_{it}^\top\bar\mtheta_i\right) \mW_{it}\left( Y_t  - \mW^\top_{it}\bar\mtheta_i \right) \right\}$ and $\left\{ \mlambda^\top \mathbf{V}_T^{-1/2} \phi''\left(\mW_{it}^\top\bar\mtheta_i\right) \mW_{it}\left( Y_t  - \mW^\top_{it}\bar\mtheta_i \right) \right\}$ are $\alpha$-mixing of size $-\size/(\size-2)$, $\size>2$. 
		Since $\mtheta=\bar\mtheta_i$ uniquely minimizes the population objective,
		\begin{align*}
			\bm 0 
			= \nabla_{\mtheta}\overline{Q}_{iT}(\bar\mtheta_i) &= \E\left[ \frac{1}{T}\sum_{t=1}^T\phi''(\mW_{it}^\top\bar\mtheta_i) \mW_{it}(\mW_{it}^\top\bar\mtheta_i - Y_t)\right],
			% = \bm 0,
		\end{align*} 
		which further implies that $\E\left[ \mlambda^\top \mathbf{V}_T^{-1/2} \phi''\left(\mW_{it}^\top\bar\mtheta_i\right) \mW_{it}\left( Y_t  - \mW^\top_{it}\bar\mtheta_i \right) \right] = \bm{0}$.
		By sub-multiplicativity of norms and Assumption \ref{ass:Philowlevel}\ref{ass:itemlowlevePhiadditionalmoms}, we get that 
		\begin{align*}
			\E\left[ \left| \mlambda^\top \mathbf{V}_T^{-1/2}  \phi''\left(\mW_{it}^\top\bar\mtheta_i\right)\mW_{it}\left( Y_t  - \mW^\top_{it}\bar\mtheta_i \right) \right|^{\size}\right] 
			&\leq \|\mlambda^\top \mathbf{V}_T^{-1/2}\|^\size \E\left[ \left\|  \phi''\left(\mW_{it}^\top\bar\mtheta_i\right)\mW_{it}\left( Y_t - \mW_{it}^\top \bar\mtheta_i \right) \right\|^\size \right]
		\end{align*} 
		is finite for $\size>2,\;\forall t$.
		Therefore, given Assumption \ref{ass:Philowlevel}\ref{ass:itemlowlevePhiCLTcramer}, the univariate convergence result in \eqref{eqn:Phiunivconv} follows. Further, by the Cramér–Wold device, the conditions of the univariate central limit theorem for mixing series hold for all linear combinations, leading to the desired result
		\begin{align}\label{eqn:phiCLT}
			T^{-1/2} \mathbf{V}_T^{-1/2} \sum_{t=1}^T \phi''\left(\mW_{it}^\top\bar\mtheta_i\right) \mW_{it}\left( Y_t  - \mW^\top_{it}\bar\mtheta_i \right) = -\mathbf{V}_T^{-1/2} \sqrt{T}\nabla_{\mtheta}\widehat{Q}_{iT}(\bar\mtheta_i)\toD  \mathcal{N}(0,I_k).
		\end{align}

		For the fourth condition of \citet[Theorem  3.1]{NM_1994}, we use \citet[Lemma 2.4]{NM_1994} to show uniform convergence in probability for the Hessian, i.e.
		\begin{align}\label{eqn:phiLLNhesse}
			\sup_{\mtheta \in \mTheta} \left\Vert  \nabla_{\mtheta\mtheta}\widehat{Q}_{iT}(\mtheta) -  \nabla_{\mtheta\mtheta}\overline{Q}_{iT}(\mtheta)  \right\Vert    \toP0.
		\end{align} 
		Since we consider stationary random variables, a compact $\mTheta$ by Assumption \ref{ass:TrueLinearSpecification}, and the objective function is continuous for all $\mtheta\in\mTheta$ by Assumption \ref{ass:Philowlevel}\ref{ass:itemlowlevePhicont}, we are left to show that there exists a dominating function $d_3(\mW_{it}, Y_t) \geq \left\Vert \nabla_{\mtheta\mtheta}\widehat{Q}_{iT}(\mtheta)  \right\Vert$ for all $\mtheta\in\mTheta$, with $\E\left[d_3(\mW_{it}, Y_t) \right] < \infty$, which holds for
		\begin{align}
			\begin{split}
				d_3(\mW_{it}, Y_t) &:=
				\sup_{\mtheta \in \mTheta}\left( \left\vert \phi'''\left(\mW_{it}^\top\mtheta\right) \right\vert \cdot \left\Vert \mW_{it}\right\Vert^2 \cdot\left\vert \mW_{it}^\top\mtheta - Y_t\right\vert\right) + \sup_{\mtheta \in \mTheta} \left\vert \phi''\left(\mW_{it}^\top\mtheta\right) \right\vert  \cdot \left\Vert \mW_{it}\right\Vert^2,
			\end{split}
		\end{align}
		whose expectation is finite by Assumption \ref{ass:Philowlevel}\ref{ass:itemlowlevePhiadditionalmoms}.
		Hence, the desired convergence result in \eqref{eqn:phiLLNhesse} follows.
		
		\sloppy
		Moreover, the non-singular expected Hessian, $\nabla_{\mtheta\mtheta}\overline{Q}_{iT}(\bar\mtheta_i) =  \E \left[   \phi'''(\mW_{it}^\top\bar\mtheta_i)\mW_{it}\mW_{it}^\top(\mW_{it}^\top\bar\mtheta_i - Y_t) + \phi''(\mW_{it}^\top\bar\mtheta_i)\mW_{it}\mW_{it}^\top\right]$, required in condition (v) of \citet[Theorem  3.1]{NM_1994} holds by Assumption~\ref{ass:Philowlevel}\ref{ass:itemlowlevePhiCLTcramer}.
		We conclude the verification of Assumption \ref{ass:Normality}\ref{item:EstOP1} by noting that the above results generalize to the case where the reference forecast is considered, which corresponds to the setting $\mW_{it}=1$.

		We continue to verify the high-level Assumption \ref{ass:Normality}\ref{item:CLT}, which invokes a CLT for the vector of scores $T^{-1/2} \sum_{t=1}^T \big(a_{1t}, c_{1t}, a_{2t}, c_{2t}, e_{t} \big)$. 
		For this, we use the Cramér-Wold Device under which the problem reduces to showing that 
		\begin{align}
			\label{eqn:Phiscorecrameruniv}
			T^{-1/2} \mlambda^\top  \mOmega_T^{-1/2} \sum_{t=1}^T \big(a_{1t}, c_{1t}, a_{2t}, c_{2t}, e_{t} \big)^\top \toD \mathcal{N}(0, 1),
		\end{align}
		for all $\mlambda\in\R^5$ with $\mlambda^\top \mlambda=1$. 
		
		Since the scores are measurable functions of the $\a$-mixing variables by Assumption \ref{ass:Philowlevel}\ref{ass:itemlowlevePhimixing}, it follows from \citet[Theorem 3.49]{W_2001} that the components of the vector $\big(a_{1t}, c_{1t}, a_{2t}, c_{2t}, e_{t} \big)$ are also \(\alpha\)-mixing of size $-\size/(\size-2)$ for $\size>2$. 
		Thus, we can apply the standard central limit theorem for mixing variables in \citet[Theorem 5.20]{W_2001}. 
		To do so, note that $a_{it}, c_{it}, e_t$, $i=1,2$,  are defined as difference in scores between the sample and population values. Hence these terms are centered around zero by construction and thus $\E\left[  \mlambda^\top \big(a_{1t}, c_{1t}, a_{2t}, c_{2t}, e_{t} \big)^\top\right]=0$. 
		It further holds that
		\begin{align}
			\E\left\vert \mlambda^\top \mOmega_T^{-1/2}\big(a_{1t}, c_{1t}, a_{2t}, c_{2t}, e_{t} \big)^\top \right\vert^\size \leq \left\Vert \mlambda^\top \mOmega_T^{-1/2} \right\Vert^\size \E\left\Vert \big(a_{1t}, c_{1t}, a_{2t}, c_{2t}, e_{t} \big) \right\Vert^\size < \infty, 
		\end{align}
		which validates the required moment conditions for \citet[Theorem 5.20]{W_2001}. 
		
		Since  $\left\Vert \mlambda^\top \mOmega_T^{-1/2} \right\Vert^\size$ is finite by Assumption \ref{ass:itemlowlevePhicltvar}, we particularly analyze the vector $(a_{1t}, c_{1t}, a_{2t}, c_{2t}, e_{t})$  and note that by linearity of expectations
		\begin{align}\label{eqn:Phiscorecrinequ}
			\E \left\|  \big(a_{1t}, c_{1t}, a_{2t}, c_{2t}, e_{t} \big)^\top \right\|_{\size}^\size =  \E|a_{1t}|^\size + \E|c_{1t}|^\size +   \E|a_{2t}|^\size + \E|c_{2t}|^\size + \E|e_{t}|^\size < \infty,
		\end{align}
		provided that each $\E|a_{it}|^\size, \E|c_{it}|^\size , \E|e_{t}|^\size < \infty$. To verify this, we detail the bound for $\E|c_{it}|^\size$ and show that it is finite by applying the $c_\size$-inequality: 
		\begin{align}
			\E|c_{it}|^\size &= \E\left|\myS(\mW_{it}^\top\mtheta,Y_t) - \E [ \myS(\mW_{it}^\top\mtheta,Y_t)]\right|^\size \nonumber \\ \label{eqn:Phiscore_c}
			&\leq 2^{\size-1} \bigg( \E\left[ |\myS(\mW_{it}^\top\mtheta,Y_t)|^{\size}\right] + \left| \E\left[ \myS(\mW_{it}^\top\mtheta,Y_t)\right]\right|^\size\bigg) < \infty.
		\end{align}
		Finiteness of the sum in \eqref{eqn:Phiscore_c} follows from the moment conditions in Assumptions \ref{ass:Philowlevel}\ref{ass:itemlowlevePhiadditionalmoms} by again applying the $c_\size$-inequality, i.e.,
		\begin{align*}%\label{eqn:Phiscorecrinequ2}
			\E\left|\myS(\mW_{it}^\top\mtheta,Y_t) \right|^\size \leq 3^{2\size-1} \left( \E \left[ \left\vert\phi(Y_t) \right\vert^\size\right] +  \E \left[ \left\vert\phi(\mW_{it}^\top\mtheta) \right\vert^\size\right] + \E \left[ \left\vert\phi'(\mW_{it}^\top\mtheta)\left( Y_t- \mW_{it}^\top\mtheta\right) \right\vert^\size\right]\right) < \infty.
		\end{align*}
		The same reasoning applies to $\E|a_{it}|^\size$ and $\E|e_{t}|^\size$. Since all norms on a finite-dimensional space are equivalent, there exists a constant  $K>0$  such that, for the vector $(a_{1t}, c_{1t}, a_{2t}, c_{2t}, e_{t})$, the Euclidean norm is at most $K$ times the $L_\size$-norm, i.e., we can write $\left\|  \big(a_{1t}, c_{1t}, a_{2t}, c_{2t}, e_{t} \big)^\top \right\| \leq K\left\|  \big(a_{1t}, c_{1t}, a_{2t}, c_{2t}, e_{t} \big)^\top \right\|_{\size}$. 
		Raising this expression to the $\size$-th power, taking expectations, and using  \eqref{eqn:Phiscorecrinequ}, we get that
		\begin{align*}
			\E\left\|  \big(a_{1t}, c_{1t}, a_{2t}, c_{2t}, e_{t} \big)^\top \right\|^\size \leq  K^\size \E\left\|  \big(a_{1t}, c_{1t}, a_{2t}, c_{2t}, e_{t} \big)^\top \right\|_\size^\size \leq K^\size \Big[ \E|a_{1t}|^\size + \ldots +   \E|e_{t}|^\size \Big] < \infty,
		\end{align*}
		because $K^{\size}$ is a constant.  Hence, the moment condition of \citet[Theorem 5.20]{W_2001} is fulfilled and the univariate convergence result in \eqref{eqn:Phiscorecrameruniv} follows, given Assumption \ref{ass:itemlowlevePhicltvar}. Further, by the Cramér–Wold device, the conditions of the univariate central limit theorem for mixing series hold for all linear combinations, leading to the desired result,
		\begin{align} 
			T^{-1/2}  \mOmega_T^{-1/2} \sum_{t=1}^T \big(a_{1t}, c_{1t}, a_{2t}, c_{2t}, e_{t} \big)^\top \toD \mathcal{N}(0, I_5). 
		\end{align}

		For the verification of the high-level Assumption \ref{ass:Normality}\ref{item:ContExpScore}, we notice that by Assumption \ref{ass:Philowlevel}\ref{ass:itemlowlevePhicont}, 
		$\phi(\mW_{it}^\top\mtheta)$ is continuously differentiable in $\mtheta$, 
		implying that $\myS(\mW_{it}^\top\mtheta,Y_t)$ is also continuously differentiable in $\mtheta$ with derivative 
		\begin{align*}
			\nabla_{\mtheta} \myS(\mW_{it}^\top\mtheta,Y_t) =  \phi''(\mW_{it}^\top\mtheta) \mW_{it}(\mW_{it}^\top\mtheta - Y_t),
		\end{align*}
		which can be bounded by
		\begin{align*}
			\Vert \phi''(\mW_{it}^\top\mtheta) \mW_{it}(\mW_{it}^\top\mtheta - Y_t) \Vert 
			\le 
			\sup_{\mtheta \in \mTheta}
			\Vert \phi''(\mW_{it}^\top\mtheta) \mW_{it}(\mW_{it}^\top\mtheta - Y_t) \Vert,
		\end{align*}
		whose expectation is finite by Assumption~\ref{ass:Philowlevel}\ref{ass:itemlowlevePhiadditionalmoms}.
		Hence, differentiability carries over to the expectation by the Leibniz integral rule. 
		Moreover, under the same assumptions and reasoning, both the first- and second-order derivatives of the expected
		score with respect to $\mtheta$ remain continuously differentiable.

		For the verification of high-level Assumption \ref{ass:Normality}\ref{item:SE}, and similar to Lemma~\ref{lem:GPL_Equicontinuous} below, we get Lipschitz continuity,
		\begin{align*}
			\left| \myS \big(\mW_{it}^\top \mtheta, Y_t \big) - \myS \big(\mW_{it}^\top \bm\tau, Y_t \big) \right| 
			\le \bar b_t \big\Vert \mtheta - \bm\tau \big\Vert
		\end{align*}
		where
		\begin{align*}
			\bar b_t := \sup_{\mtheta \in \mTheta}\left\vert\phi''(\mW_{it}^\top\mtheta) \right\vert \left\Vert\mW_{it}\right\Vert \cdot \left| \mW_{it}^\top \mtheta - Y_t \right|.
		\end{align*}
		Hence, we employ \citet[Theorem~3]{hansen1996stochastic} for mixing arrays, for which we verify his Assumption~4 using the function space $F := \big\{ \myS(\mW_{it}^\top\mtheta, Y_t): \; \mtheta \in \mTheta \big\}$.
		First, the mixing condition from Assumption~\ref{ass:Philowlevel} implies the first condition of \citet[Assumption~4]{hansen1996stochastic}.
		The second condition of \citet[Assumption~4]{hansen1996stochastic} holds by stationarity and 
		\begin{align*}
			\limsup_{T\to\infty} \left\{ \frac{1}{T} \sum_{t=1}^T \left( \E \left\vert \myS(\mW_{it}^\top\mtheta, Y_t) \right\vert^s \right)^{2/s} \right\}^{1/2}
			= \left( \E \left\vert \myS(\mW_{it}^\top\mtheta, Y_t) \right\vert^s \right)^{1/s} 
			< \infty.
		\end{align*}
		Third, a similar condition holds for the Lipschitz constant
		\begin{align*}
			\limsup_{T\to\infty} \left\{ \frac{1}{T} \sum_{t=1}^T \left( \E \left\vert \sup_{\mtheta \in \mTheta}\left\vert\phi''(\mW_{it}^\top\mtheta) \right\vert \left\Vert\mW_{it}\right\Vert \cdot \left| \mW_{it}^\top \mtheta - Y_t \right| \right\vert^s \right)^{2/s} \right\}^{1/2}
			< \infty,
		\end{align*} 
		as $\E \left\vert \sup_{\mtheta \in \mTheta}\left\vert\phi''(\mW_{it}^\top\mtheta) \right\vert \left\Vert\mW_{it}\right\Vert \cdot \left| \mW_{it}^\top \mtheta - Y_t \right| \right\vert^s < \infty$ by assumption.
		
		Hence, we can employ \citet[Theorem~3]{hansen1996stochastic} to conclude that their Condition~1 holds with the $L_q$-norm $\rho_q(\cdot) = \left( \E \big\vert \cdot \big\vert^q\right)^{1/q}$ (see the proof of Lemma~\ref{lem:GPL_Equicontinuous} for further details).

		Thus, we have shown that Assumption~\ref{ass:Normality} holds such that Theorem \ref{thm:joint_nomal_mcb_dsc} applies under the given conditions.
		We now continue to verify the high-level conditions in Assumption~\ref{ass:AsymptoticsDSC0}.
		
		\bigskip
		\noindent
		\textbf{Verification of Assumption~\ref{ass:AsymptoticsDSC0}:} \\
		For the verification of the high-level Assumption~\ref{ass:AsymptoticsDSC0}\ref{item:SmoothSDSC}, recall that 
		\begin{align}
			\begin{split}
				h_T(\mtheta) &=  \frac{1}{T} \sum_{t=1}^T   \left\vert \myS'''\left( \mW_{it}^\top\mtheta, Y_t\right)  \right\vert \left\Vert \mW_{it} \right\Vert^3 \\  
				&= \frac{1}{T}\sum_{t=1}^T \left\vert 2 \phi'''(\mW_{it}^\top\mtheta) +  \phi''''(\mW_{it}^\top\mtheta) (\mW_{it}\mtheta-Y_t) \right\vert \cdot \left\Vert \mW_{it} \right\Vert^3 \\
				&=: \frac{1}{T} \sum_{t=1}^T  g_t(\mtheta), 
			\end{split}
		\end{align}
		and note that
		\begin{align*}
			\begin{split}
				h_T'(\mtheta) &= \frac{1}{T} \sum_{t=1}^T \sgn\left\{ \myS'''\left( \mW_{it}^\top\mtheta, Y_t\right) \right\} \cdot \left( \phi''''' \left(\mW_{it}^\top\mtheta\right) \left( \mW_{it}^\top\mtheta - Y_t\right) + 3\phi''''\left(\mW_{it}^\top\mtheta\right) \right)  \cdot \mW_{it} \cdot \left\Vert \mW_{it} \right\Vert^3 \\
				&=: \frac{1}{T} \sum_{t=1}^T  g'_t(\mtheta).
			\end{split}
		\end{align*}
		
		Concerning the first part of the high level Assumption~\ref{ass:AsymptoticsDSC0}\ref{item:SmoothSDSC}, we use that $h_T(\mtheta)$ is Lipschitz-continuous with (random) Lipschitz constant $L = \sup_{\mtheta \in \mTheta} \Vert h_T'(\mtheta)\Vert$ as $\phi$ is assumed to be five times differentiable  with bounded derivatives; see Assumption~\ref{ass:Philowlevel}\ref{ass:itemlowlevePhiadditionalmoms}. Therefore, we get that
		\begin{align*}
			T^{-1/2} \sup_{\Vert \mtheta - \bar\mtheta_i \Vert \le \delta}  \big| h_T(\mtheta ) - h_T(\bar{\mtheta}_{i}) \big|
			\;\le\; T^{-1/2} \sup_{\Vert \mtheta - \bar\mtheta_i \Vert \le \delta}  L \big\Vert \mtheta - \bar{\mtheta}_{i} \big\Vert  
			\;\le\; T^{-1/2} L \delta.
		\end{align*}
		
		To establish that $T^{-1/2} L \delta = o_\P(1)$, it suffices to verify a uniform law of large numbers for $\Vert h_T'(\mtheta) \Vert$. By taking the supremum of $\Vert g'_t(\mtheta) \Vert$, we obtain the dominating function
		\begin{align*}
			d_4\left(\mW_{it}^\top\mtheta, Y_t\right) := \sup_{\mtheta \in \mTheta} \left\vert  \phi'''''\left(\mW_{it}^\top\mtheta\right) \cdot \left(\mW_{it}^\top\mtheta - Y_t\right)  + 3\phi''''\left(\mW_{it}^\top\mtheta\right) \right\vert \cdot\left\Vert \mW_{it} \right\Vert^4, 
		\end{align*}
		which is bounded in expectation, i.e., $\E\left[d_4\left(\mW_{it}^\top\mtheta, Y_t\right)\right]<\infty$, given Assumptions \ref{ass:Philowlevel}\ref{ass:itemlowlevePhiadditionalmoms}.
		By Assumption \ref{ass:TrueLinearSpecification} the data is stationary and the parameter space $\mTheta$ is compact. As the function $\Vert g'_t(\mtheta)\Vert$ is continuous for all $\mtheta\in\mTheta$ by Assumption \ref{ass:Philowlevel}\ref{ass:itemlowlevePhicont}, we can apply the uniform law of large numbers in \citet[Lemma 2.4]{NM_1994} to get that 
		\begin{align}
			\sup_{\mtheta\in\mTheta} \left\Vert \frac{1}{T}\sum_{t=1}^T \Vert g'_t(\mtheta) \Vert  - \E\big[\Vert g'_t(\mtheta) \Vert\big] \right\Vert \toP 0. 
		\end{align}
		Hence, the desired stochastic Lipschitz bound follows.

		Concerning the second part of the high level Assumption~\ref{ass:AsymptoticsDSC0}\ref{item:SmoothSDSC}, we show that $T^{-1/2} h_T(\bar{\mtheta}_{i}) = o_\P(1)$ by establishing that $h_T(\bar{\mtheta}_{i}) $ satisfies the pointwise law of large numbers.
		Under Assumptions~\ref{ass:TrueLinearSpecification} and \ref{ass:Philowlevel}\ref{ass:itemlowlevePhiadditionalmoms}, \citet[Theorem 3.49]{W_2001} implies that $g_t(\bar{\mtheta}_{i})$ is stationary and $\alpha$-mixing of size $-\size/(\size-2), \size>2$. 
		Moreover, ergodicity of $g_t(\bar{\mtheta}_{i})$ follows since the function itself is a  measurable transformation of $\left(\mW_{it}^\top, Y_t\right)$ \citep[Proposition 3.44]{W_2001}. We apply the triangle inequality to $\left\vert g_t(\bar{\mtheta}_{i})\right\vert$ which, together with Assumption~\ref{ass:Philowlevel}\ref{ass:itemlowlevePhiadditionalmoms} results in
		\begin{align}\label{eqn:hbound}
			\E\left[ g_t(\mtheta)\right] \leq \E\left[ \left\vert \phi'''\left(\mW_{it}\bar\mtheta_i \right)\right\vert \left\Vert \mW_{it} \right\Vert^3  \right] + \E\left[ \left\vert \phi''''\left(\mW_{it}\bar\mtheta_i \right)\right\vert \cdot \left\vert \mW_{it}\bar\mtheta_i - Y_t\right\vert \cdot  \left\Vert \mW_{it} \right\Vert^3  \right] < \infty.
		\end{align}
		Given \eqref{eqn:hbound}, Theorem 3.34 in \citet{W_2001} implies the point-wise law of large numbers and we have that $h_T(\bar\mtheta_i)= \E\left[g_t(\bar\mtheta_i)\right] + o_\P(1)$. Hence, $T^{-1/2} h_T(\bar\mtheta_i)= o_\P(1)$, as $T\to\infty$.
		\color{black}
		
		\bigskip
		For the verification of the high-level Assumption \ref{ass:AsymptoticsDSC0}\ref{item:BahadurMestJoint}, first note that the true parameter $\bar \mtheta_i$ lies in $\interior(\mTheta)$ such that consistency ensures that $\widehat\mtheta_{iT}$ is in the interior of the parameter space with probability approaching one as $T \to \infty$ and the first-order condition $\bm 0 = \nabla_{\mtheta}\widehat{Q}_{iT}(\widehat\mtheta_{iT})$ holds with probability approaching one.\footnote{\label{fn:UniqueMinMeanReg} For \emph{correctly specified} models, the uniqueness of the minimizer in \eqref{eqn:PseudoTrue} follows naturally. Since the score is strictly consistent for the mean, \citet[Theorem 1]{HE_2014} implies that under correct model specification, 
			$\E\left[ \myS(\mW_{it}^\top \bar\mtheta_i,Y_t)  \right] \leq \E\left[ \myS(\mW_{it}^\top \mtheta,Y_t)  \right]$ for all $\mtheta \in \mTheta$, 
			with equality if and only if $\mW^\top_{it}\bar\mtheta_i \overset{\text{a.s.}}{=} \mW^\top_{it}\mtheta$,
			which implies $\left(\bar\mtheta_i -\mtheta\right)^\top \E\left[\mW_{it}\mW^\top_{it}\right]\left(\bar\mtheta_i -\mtheta\right) = 0$.
			As the matrix $\E\left[\mW_{it}\mW^\top_{it}\right]$ is positive definite, this can hold only if $\bar\mtheta_i-\mtheta = \bm 0$ such that $\mtheta=\bar\mtheta_i$ is the \emph{unique} minimizer of $\overline{Q}_{iT}(\mtheta)$.
		}

		We apply the mean-value theorem to get that  
		\begin{align}
			\label{eqn:MeanValueExpansionLoss}
			\bm 0 &= 
			\nabla_{\mtheta}\widehat{Q}_{iT} (\widehat\mtheta_{iT})
			= \nabla_{\mtheta}\widehat{Q}_{iT}\left(\bar\mtheta_{i}\right)+\nabla_{\mtheta\mtheta}\widehat{Q}_{iT}\left(\widetilde\mtheta_{iT}\right)\left(\widehat\mtheta_{iT} - \bar{\mtheta}_{i}\right),
		\end{align}
		where $\widetilde{\mtheta}_{iT}$ is a (component-wise differing) mean value between $\widehat{\mtheta}_{iT}$ and $\bar{\mtheta}_{i}$, such that $\left\Vert \widetilde{\mtheta}_{iT}- \bar\mtheta_i \right\Vert \leq \left\Vert \widehat{\mtheta}_{iT}- \bar\mtheta_i \right\Vert$. 
		Thus, consistency of $\widehat{\mtheta}_{iT}$ and  continuous mapping implies $\left[\nabla_{\mtheta\mtheta}\widehat Q_{iT}\left( \widetilde\mtheta_{iT}\right)\right]^{-1}  = \left[\nabla_{\mtheta\mtheta}\widehat Q_{iT}\left( \bar\mtheta_i\right)\right]^{-1} + o_\P(1)$. 
		Notice that the empirical Hessian $\nabla_{\mtheta\mtheta}\widehat Q_{iT}\left( \widetilde\mtheta_{iT}\right)$ is invertible for $T$ large enough as $\widetilde{\mtheta}_{iT} \toP \bar{\mtheta}_{i}$ and the empirical Hessian converges to the expected Hessian  
		\begin{align*}
			\nabla_{\mtheta\mtheta}\overline{Q}_{iT}(\bar\mtheta_i)  
			&=\E \left[ \phi''(\mW_{it}^\top\bar\mtheta_i)\mW_{it}\mW_{it}^\top + \phi'''(\mW_{it}^\top\bar\mtheta_i)\mW_{it}\mW_{it}^\top(\mW_{it}^\top\bar\mtheta_i - Y_t)  \right]
		\end{align*}
		% $\nabla_{\mtheta\mtheta}\widehat Q_{iT}(\bar{\mtheta}_{i})$ 
		at the pseudo-true parameter, which is invertible by Assumption~\ref{ass:Philowlevel}\ref{ass:itemlowlevePhiCLTcramer}.

		Therefore, rearranging \eqref{eqn:MeanValueExpansionLoss} gives
		\begin{align}
			\label{eqn:phimest}
			\begin{split}
				\sqrt{T}\left(\widehat\mtheta_{iT} - \bar{\mtheta}_{i}\right) &= - \left(\left[ \nabla_{\mtheta\mtheta}\widehat{Q}_{iT}\left(\bar\mtheta_{i}\right)\right]^{-1} + o_\P(1) \right)\sqrt{T}\nabla_{\mtheta}\widehat{Q}_{iT}\left(\bar\mtheta_{i}\right) \\ 
				&=- \left[ \nabla_{\mtheta\mtheta}\widehat{Q}_{iT}\left(\bar\mtheta_{i}\right)\right]^{-1} \sqrt{T}\nabla_{\mtheta}\widehat{Q}_{iT}\left(\bar\mtheta_{i}\right) + o_\P(1) \\
				&= - \left(\frac{1}{T}\sum_{t=1}^T  \phi''(\mW_{it}^\top\bar\mtheta_i)\mW_{it}\mW_{it}^\top + \phi'''(\mW_{it}^\top\bar\mtheta_i)\mW_{it}\mW_{it}^\top(\mW_{it}^\top\bar\mtheta_i - Y_t)  \right)^{-1} \times \\ 
				&\qquad\qquad T^{-1/2} \sum_{t=1}^T \phi''(\mW_{it}^\top \bar\mtheta_{i} ) \mW_{it}(\mW_{it}^\top \bar\mtheta_{i}  - Y_t)   + o_\P(1) \\
				&= - \left( \frac{1}{T} \sum_{t=1}^T \myS''(\mW_{it}^\top \bar \mtheta_i, Y_t)  \mW_{it} \mW_{it}^\top\right)^{-1}  \left( T^{-1/2} \sum_{t=1}^T \myS'(\mW_{it}^\top \bar \mtheta_i, Y_t) \mW_{it} \right) + o_\P(1).
			\end{split}
		\end{align}

		For the case $\DSCp_i = 0$, we first note that $\bar r = \mW_{it}^\top \bar \mtheta_i$ as in the proof of Theorem~\ref{thm:MCBandDSC0}.
		The asymptotic behavior of the unconditional expectation $\bar{r}$ follows as above for $\mW_{it}=1$ and by noting that $\frac{1}{T}\sum_{t=1}^T\phi'''(\bar r)(\bar r - Y_t) \toP 0$. 
		Based on this, we stack the estimators emerging from \eqref{eqn:phimest} to get that 
		\begin{align}
			\label{eqn:Mstacked}
			\begin{split}
				\sqrt{T}	
				\begin{pmatrix}
					\widehat{r}_{T} - \bar r   \\
					\widehat{\mtheta}_{iT} - \bar\mtheta_{i} 
				\end{pmatrix} 
				&  = -
				\begin{pmatrix}
					\left( \frac{1}{T} \sum_{t=1}^T \phi''(\bar r) \right)^{-1},   & 0 \\
					\multicolumn{2}{c}{\left(\frac{1}{T}\sum_{t=1}^T  \phi''(\bar r)\mW_{it}\mW_{it}^\top + \phi'''(\bar r)\mW_{it}\mW_{it}^\top(\bar r - Y_t)  \right)^{-1}}
				\end{pmatrix} \times 
				\\&\qquad
				T^{-1/2} \sum_{t=1}^T \phi''(\bar r) \mW_{it}(\bar r  - Y_t)    + o_\P(1). 
			\end{split}
		\end{align}
		The law of large numbers established in \eqref{eqn:phiLLNhesse} ensures convergence of the first terms in \eqref{eqn:phimest} and \eqref{eqn:Mstacked} and invertibility follows from Assumption \ref{ass:Philowlevel}\ref{ass:itemlowlevePhihessian} for $T$ large enough.
		Asymptotic normality of the second terms in \eqref{eqn:phimest} and  \eqref{eqn:Mstacked} follows from the central limit theorem in \eqref{eqn:phiCLT}.
		This verifies Assumption \ref{ass:AsymptoticsDSC0}\ref{item:BahadurMestJoint} and hence concludes this proof.
	\end{proof}

	\begin{proof}[Proof of Proposition~\ref{prop:Quantiles_Verification}]
		We check the conditions of Assumption~\ref{ass:Normality}, for which we use
		\begin{align*}
			\widehat{Q}_{iT}(\mtheta) &:= \frac{1}{T} \sum_{t=1}^T \myS(\mW_{it}^\top \mtheta,Y_t)  
			\qquad \text{ and } \qquad \\
			\overline{Q}_{iT}(\mtheta) &:= \E \left[ \widehat{Q}_{iT}(\mtheta) \right] = \E \left[\frac{1}{T} \sum_{t=1}^T  \myS(\mW_{it}^\top \mtheta,Y_t) \right]
			= \E \left[\myS(\mW_{it}^\top \mtheta,Y_t) \right],
		\end{align*}
		with 
		\begin{align*}
			\myS(\mW_{it}^\top \mtheta,Y_t)
			=  \big( \mathds{1}\{\mW_{it}^\top \mtheta \ge Y_t\} - \alpha \big) \big( g(\mW_{it}^\top \mtheta) - g(Y_t) \big).
		\end{align*}
		While the scoring function $\myS$ itself is non-smooth (at $Y_t = \mW_{it}^\top \mtheta$), we show that its expectation is twice (as required below) continuously differentiable everywhere, which implies condition~\ref{item:ContExpScore} of Assumption~\ref{ass:Normality}.
		For this, we use the law of iterated expectations and write 
		\begin{align}
			\nonumber 
			\overline{Q}_{iT}(\mtheta) 
			% = \E \big[ \myS\big(\mW_{it}^\top \mtheta, Y_t \big) \big] 
			&= \E \Big[ \E \big[
			(\mathds{1}\{Y_t \le \mW_{it}^\top \mtheta \} - \alpha) (g(\mW_{it}^\top \mtheta) - g(Y_t)) \mid \mathcal{W}_{it} \big] \Big] \\
			&= \E \Big[ \big( F_{it}(\mW_{it}^\top \mtheta) - \alpha \big) g(\mW_{it}^\top \mtheta)\Big] 
			- \E \left[ \int_{-\infty}^{\mW_{it}^\top \mtheta} g(y) f_{it}(y) \mathrm{d} y \right] + \alpha \E[g(Y_t)].
			\label{eqn:GPL_Differentiable}
		\end{align}
		We interchange expectation and differentiation by the dominated convergence theorem  as there exist functions $G$ and $G'$ such that $\Vert \mW_{it} g(\mW_{it}^\top \mtheta) \Vert  \le G(\mW_{it})$ with $\E|G(\mW_{it})| < \infty$ (as $0 < f_{it}(\cdot) \le K < \infty$ on its support) and $\Vert  \mW_{it} g'(\mW_{it}^\top \mtheta) \Vert \le G'(\mW_{it})$ with $\E|G'(\mW_{it})| < \infty$.
		The first term in \eqref{eqn:GPL_Differentiable} is continuously differentiable as $F_{it}$ and $g$ are continuously differentiable, and has derivative $\E \big[ \mW_{it} f_{it}(\mW_{it}^\top \mtheta) g(\mW_{it}^\top \mtheta) + \mW_{it} ( F_{it}(\mW_{it}^\top \mtheta) - \alpha ) g'(\mW_{it}^\top \mtheta) \big]$,
		and the third term does not depend on $\mtheta$.
		For the middle term, as $g$ and $f_{it}$ are continuous, the fundamental theorem of calculus implies
		$\nabla_{\mtheta} \E \left[  \int_{-\infty}^{\mW_{it}^\top \mtheta} g(y) f_{it}(y) \mathrm{d} y \right] = \E \left[ \mW_{it} g(\mW_{it}^\top \mtheta) f_{it}(\mW_{it}^\top \mtheta) \right]$.
		Hence, we obtain differentiability such that condition~\ref{item:ContExpScore} of Assumption~\ref{ass:Normality} is satisfied and get that 
		\begin{align}
			\label{eqn:QuantNablaQ}
			\nabla_{\mtheta} \overline{Q}_{iT}(\mtheta) = \E \left[ \mW_{it} g'(\mW_{it}^\top \mtheta) \big( F_{it}(\mW_{it}^\top \mtheta) - \alpha \big) \right]. 
		\end{align}
		For the second derivative, which will be used below, we similarly get that
		\begin{align}
			\label{eqn:QuantNabla2Q}
			\nabla_{\mtheta \mtheta} \overline{Q}_{iT}(\mtheta) = \E \Big[ (\mW_{it} \mW_{it}^\top) \left\{ g''(\mW_{it}^\top \mtheta) \big( F_{it}(\mW_{it}^\top \mtheta) - \alpha \big) + g'(\mW_{it}^\top \mtheta) f_{it}(\mW_{it}^\top \mtheta)\right\} \Big],
		\end{align} 
		which is positive definite at $\mtheta = \bar\mtheta_i$ by assumption.
		(Under correct model specification, we have $F_{it}(\mW_{it}^\top \bar\mtheta_i) = \alpha$, and using the tick loss, we have $g''(z) = 0$. In either of these cases, we obtain the simplification $\nabla_{\mtheta \mtheta} \overline{Q}_{iT}(\bar\mtheta_i) = \E \left[ (\mW_{it} \mW_{it}^\top) g'(\mW_{it}^\top \bar\mtheta_i) f_{it}(\mW_{it}^\top \bar\mtheta_i) \right]$.)
		
		We next show condition~\ref{item:SE} of Assumption~\ref{ass:Normality}, i.e., stochastic equicontinuity of the empirical process $\nu_{iT}(\mtheta)$ from \eqref{eqn:EmpProcessSE}.
		For this, we first show the following lemma:
		\begin{lemma}
			\label{lem:GPL_Lipschitz}
			Given that the function $g$ is Lipschitz with constant $L_g$, the GPL loss \eqref{eqn:GPLLoss} is Lipschitz with constant $L_g \max(\alpha, 1-\alpha)$.
		\end{lemma}
		
		\begin{proof}[Proof of Lemma~\ref{lem:GPL_Lipschitz}]
			We start to write the function $\myS$ from \eqref{eqn:GPLLoss} as $\myS(x,y) = \rho(g(x), g(y))$ with
			\begin{align}
				\rho(u,v) = 
				\begin{cases} 
					\alpha (v-u), \quad &u \le v, \\
					(1-\alpha) (u-v), \quad &u > v,
				\end{cases}    
			\end{align}
			which is a piecewise linear function with slopes $\alpha$ and $1-\alpha$, respectively.
			Hence, we get that for all $u_1, u_2, v \in \R$ that
			$\big\vert \rho(u_1,v) - \rho(u_2,v) \big\vert \le \max(\alpha, 1-\alpha) \big\vert  u_1 - u_2 \big\vert$ and as a direct consequence,
			\begin{align*}
				\big\vert \myS(x_1,y) - \myS(x_2,y) \big\vert
				&= \big\vert \rho(g(x_1), g(y)) - \rho(g(x_2), g(y)) \big\vert \\
				&\le \max(\alpha, 1-\alpha) \big\vert g(x_1) - g(x_2) \big\vert \\
				&\le L_g \max(\alpha, 1-\alpha) \big\vert x_1 - x_2 \big\vert
			\end{align*}
			holds for all $x_1, x_2, y \in \R$, which concludes the proof of the lemma.
		\end{proof}
		
		The following lemma shows equicontinuity of $\nu_{iT} (\mtheta)$ from \eqref{eqn:EmpProcessSE}.
		
		\begin{lemma}
			\label{lem:GPL_Equicontinuous}
			Assume that the function $g$ from \eqref{eqn:GPLLoss} is Lipschitz continuous, $(Y_t, \mW^\top_{it})$ is $\alpha$-mixing of size $-(s-\tilde{k})/(\tilde{k}s)$ for some  $s > \tilde k > k \ge 2$ (where $\mTheta \subseteq \mathbb{R}^k$), and $\E \big\Vert \mW_{it} \big\Vert^s < \infty$, $\E \big| g(\mW_{it}^\top\mtheta) - g(Y_t) \big|^s < \infty$.
			Then, the empirical process $\nu_{iT} (\mtheta)$ from \eqref{eqn:EmpProcessSE} based on the GPL loss $\myS$ from \eqref{eqn:GPLLoss} is stochastically equicontinuous; see \eqref{eqn:StochasticEquicontinuity}.
		\end{lemma}
		
		\begin{proof}[Proof of Lemma~\ref{lem:GPL_Equicontinuous}]
			Using Lemma~\ref{lem:GPL_Lipschitz}, we get for all $\mtheta, \bm{\tau} \in \mTheta$ that 
			\begin{align*}
				\left| \myS \big(\mW_{it}^\top \mtheta, Y_t \big) - \myS \big(\mW_{it}^\top \bm\tau, Y_t \big) \right| 
				&\le \big|\mW_{it}^\top \mtheta - \mW_{it}^\top \bm\tau \big| L_g \max(\alpha, 1-\alpha) \\
				&\le L_g \max(\alpha, 1-\alpha)  \big\Vert \mW_{it} \big\Vert  \big\Vert \mtheta - \bm\tau \big\Vert
			\end{align*}
			is Lipschitz continuous.
			Hence, we employ \citet[Theorem~3]{hansen1996stochastic} for mixing arrays, for which we verify his Assumption~4.
			For this, we choose 
			the function space $F := \big\{ \myS(\mW_{it}^\top\mtheta, Y_t): \; \mtheta \in \mTheta \big\}$.
			
			First, the mixing condition from Assumption~\ref{ass:QuantilesLowLevel} implies the first condition of \citet[Assumption~4]{hansen1996stochastic}.
			The second condition of \citet[Assumption~4]{hansen1996stochastic} holds as we impose stationarity and $\E \big| g(\mW_{it}^\top\mtheta) - g(Y_t) \big|^s < \infty$ such that for any $\myS(\mW_{it}^\top\mtheta, Y_t) \in F$,
			\begin{align*}
				\limsup_{T\to\infty} \left\{ \frac{1}{T} \sum_{t=1}^T \left( \E \left\vert \myS(\mW_{it}^\top\mtheta, Y_t) \right\vert^s \right)^{2/s} \right\}^{1/2}
				= \left( \E \left\vert \myS(\mW_{it}^\top\mtheta, Y_t) \right\vert^s \right)^{1/s} 
				< \infty.
			\end{align*}
			Third, a similar condition holds for the Lipschitz constant
			\begin{align*}
				\limsup_{T\to\infty} \left\{ \frac{1}{T} \sum_{t=1}^T \left( \E \left\vert L_g \max(\alpha, 1-\alpha)  \big\Vert \mW_{it} \big\Vert \right\vert^s \right)^{2/s} \right\}^{1/2}
				< \infty,
			\end{align*} 
			as $\E \big\Vert \mW_{it} \big\Vert^s < \infty$ by assumption.
			
			Hence, we can employ \citet[Theorem~3]{hansen1996stochastic} to conclude that their Condition~1 holds with the $L_q$-norm $\rho_q(\cdot) = \left( \E \big\vert \cdot \big\vert^q\right)^{1/q}$.
			To see that the formulation of stochastic equicontinuity in \citet[Condition~1]{hansen1996stochastic} implies \eqref{eqn:StochasticEquicontinuity}, notice that
			\begin{align*}
				\Big\{ f_1, f_2 \in F: \; \rho_q(f_1 - f_2) < \delta  \Big\}
				&= \Big\{ \mtheta, \bm{\tau} \in \mTheta: \; \rho_q \big( \myS(\mW_{it}^\top\mtheta, Y_t)  - \myS(\mW_{it}^\top\bm{\tau}, Y_t)  \big) < \delta  \Big\} \\
				&= \Big\{ \mtheta, \bm{\tau} \in \mTheta: \; \E \big\vert \myS(\mW_{it}^\top\mtheta, Y_t)  - \myS(\mW_{it}^\top\bm{\tau}, Y_t) \big\vert^q  < \delta^q  \Big\} \\
				&\supseteq \left\{ \mtheta, \bm{\tau} \in \mTheta: \; \big\Vert \mtheta - \bm{\tau} \big\Vert^q  < \frac{\delta^q}{L_g^q \max(\alpha, 1-\alpha)^q \E \Vert \mW_{it} \Vert^q}  \right\},
			\end{align*}
			such that the supremum over the the left-hand side (in \citet{hansen1996stochastic}) is larger or equal than the supremum over the right-hand side (in \eqref{eqn:StochasticEquicontinuity}).
			Furthermore, Markov's inequality can be used to transfer the $L_q$-statement in \citet{hansen1996stochastic} to the probability statement in \eqref{eqn:StochasticEquicontinuity}, which concludes this proof.
		\end{proof}

		We continue to show consistency of $\widehat{\mtheta}_{iT}$ through \citet[Theorem~2.1]{NM_1994}, which is required for the verification of Assumption~\ref{ass:Normality}~\ref{item:EstOP1}. 
		Notice that $\mTheta$ is compact by assumption and that the population objective $\overline{Q}_{iT}(\mtheta)$ is continuous on $\mTheta$ as shown above.
		The objective function $\mtheta \mapsto \overline{Q}_{iT}(\mtheta) = \E \left[\myS(\mW_{it}^\top \mtheta ,Y_t) \right]$ is \emph{uniquely} minimized by $\bar \mtheta_i$ by Assumption~\ref{ass:TrueLinearSpecification}.\footnote{
			Under correct specification of the linear quantile regression, the uniqueness follows immediately:
			As the GPL scoring function $\myS$ from \eqref{eqn:GPLLoss} is strictly consistent for the $\alpha$-quantile, we employ \citet[Theorem~1]{HE_2014} to conclude that
			$\E \left[\myS(\mW_{it}^\top \bar \mtheta_i,Y_t) \right] 
			\le \E \left[\myS(\mW_{it}^\top \mtheta,Y_t) \right]$,
			with equality if and only if $\mW_{it}^\top \mtheta = \mW_{it}^\top \bar \mtheta_i$, which is equivalent to $\mtheta = \bar \mtheta_i$ as shown in footnote \ref{fn:UniqueMinMeanReg} in the proof of Proposition~\ref{prop:verification_phi}.
			Thus, the objective function $\mtheta \mapsto \overline{Q}_{iT}(\mtheta) = \E \left[\myS(\mW_{it}^\top \mtheta ,Y_t) \right]$ is \emph{uniquely} minimized by $\bar \mtheta_i$.
		}

		It remains to show the ULLN, $\sup_{ \mtheta \in \mTheta}  \left| \widehat{Q}_{iT}(\mtheta) -  \overline{Q}_{iT}(\mtheta)\right| = o_\P(1)$ 
		through \citet[Theorem 21.9]{Davidson1994}.
		For this, the pointwise LLN, $\widehat{Q}_{iT}(\mtheta) -  \overline{Q}_{iT}(\mtheta) = o_\P(1)$ for all $\mtheta \in \mTheta$, follows by \citet[Theorem~3.34]{W_2001} as $\myS(\mW_{it}^\top \mtheta,Y_t)$ is ergodic by \citet[Proposition~3.44]{W_2001} and as $\E \left| \myS(\mW_{it}^\top \mtheta,Y_t) \right| \le \E \big| g(\mW_{it}^\top \mtheta) \big| + \E \big| g(Y_t) \big|  < \infty$ by assumption.
		Furthermore, the second condition of \citet[Theorem 21.9]{Davidson1994}, that $\big\{ \widehat{Q}_{iT}(\mtheta) - \overline{Q}_{iT}(\mtheta) \big\}$ is stochastically equicontinuous, follows immediately as 
		\begin{align*}
			\nu_{iT} (\mtheta) 
			= T^{-1/2} \sum_{t=1}^T \myS( \mW_{it}^\top \mtheta, Y_t) - \E \big[ \myS (\mW_{it}^\top \mtheta , Y_t ) \big] = T^{1/2} \left( \widehat{Q}_{iT}(\mtheta) - \overline{Q}_{iT}(\mtheta) \right)
		\end{align*}
		is shown to be stochastically equicontinuous in  Lemma~\ref{lem:GPL_Equicontinuous} above.
		
		Thus, consistency holds and we continue to show asymptotic normality through \citet[Theorem~7.1]{NM_1994}.
		The conditions $\widehat{Q}_{iT}(\widehat{\mtheta}_{iT}) \le \inf_{\mtheta \in \mTheta} \widehat{Q}_{iT}(\mtheta) - o_P(T^{-1})$ is satisfied as $\widehat{\mtheta}_{iT}$ is defined as the minimizer of $\widehat{Q}_{iT}(\mtheta)$ over $\mTheta$.
		
		Item~(i) of \citet[Theorem~7.1]{NM_1994} holds as $\overline{Q}_{iT}(\mtheta)$ is minimized on $\mTheta$ at $\bar \mtheta_i$.
		Item~(ii) of \citet[Theorem~7.1]{NM_1994}, $\bar \mtheta_i \in \operatorname{int}(\mTheta)$ holds by assumption, and item~(iii),
		twice continuous differentiability of $\overline{Q}_{iT}(\mtheta)$ with second derivative $\nabla_{\mtheta \mtheta} \overline{Q}_{iT}(\bar\mtheta_i)$ (at $\bar \mtheta_i$), is shown above around \eqref{eqn:QuantNabla2Q} and the matrix $\nabla_{\mtheta \mtheta} \overline{Q}_{iT}(\bar\mtheta_i)$ is non-singular by assumption.
		
		We continue to show item~(iv), i.e., that a multivariate CLT holds for 
		\begin{align}
			\sqrt{T} \widehat{D}_{iT} (\bar\mtheta_{i}) =  T^{-1/2} \sum_{t=1}^T \mW_{it} g'(\mW_{it}^\top \bar\mtheta_{i}) \big( \mathds{1}\{Y_t \le \mW_{it}^\top \bar\mtheta_i \} - \alpha \big),
		\end{align} 
		which is the derivative of $\widehat{Q}_{iT}(\bar\mtheta_{i})$ almost surely (except when $Y_t = \mW_{it}^\top \bar\mtheta_{i}$).
		We use the time-series CLT in \citet[Theorem 5.20]{W_2001}
		and note that $\nabla_{\mtheta}\overline{Q}_{iT}(\bar\mtheta_i) = \E[ \widehat{D}_{iT} (\bar\mtheta_{i})] = 0$; see \eqref{eqn:QuantNablaQ}. 
		For this, let $\mlambda\in\R^k$ with $\mlambda^\top \mlambda=1$ and consider $\mlambda^\top \sqrt{T} \widehat{D}_{iT} (\bar\mtheta_{i})$ through the Cramér-Wold device.

		First, by the mixing condition from Assumption~\ref{ass:QuantilesLowLevel} with mixing coefficients $\alpha_m$, we have that $\sum_{m=1}^\infty \alpha_m^{(s-\tilde{k})/(\tilde{k}s)} < \infty$ for $s > \tilde k > k \ge 2$.
		As $\frac{s-\tilde{k}}{\tilde{k}s} < \frac{s-2}{s}$ for any $s > 2$ and $\tilde{k} > 2$, and $\alpha_m < 1$, we also have that     
		$\sum_{m=1}^\infty \alpha_m^{(s-2)/{s}} < \infty$, such that the sequence $(Y_t, \mW^\top_{it})_{t \in \N}$ is also strong mixing of size $-s/(s-2)$.
		\citet[Theorem 3.49]{W_2001} further implies that the sequences $\mlambda^\top \mW_{it} g'(\mW_{it}^\top \bar\mtheta_{i}) \big( \mathds{1}\{Y_t \le \mW_{it}^\top \bar\mtheta_i \} - \alpha \big)$ are mixing of the same size as the underlying random variables.
		Furthermore, for $s >2$, we have
		\begin{align*}
			\E \left[ \left| g'(\mW_{it}^\top\bar\mtheta_i) \mlambda^\top \mW_{it}
			\big( \mathds{1}\{Y_t \le \mW_{it}^\top \bar\mtheta_i \} - \alpha \big) \right|^s \right] 
			\le \E \left[ \left| g'(\mW_{it}^\top\bar\mtheta_i) \mlambda^\top \mW_{it} \right|^s \right]  
			\le \E \left[ \Vert g'(\mW_{it}^\top\bar\mtheta_i) \mW_{it} \Vert^s \right],
		\end{align*}
		which is finite by assumption, as $g'$ is bounded from above.
		Using that    
		\begin{align*}
			\mlambda^\top \mathbf{V}_T \mlambda = \var \left( T^{-1/2} \sum_{t=1}^{T} g'(\mW_{it}^\top\bar\mtheta_i) 
			\big( \mathds{1}\{Y_t \le \mW_{it}^\top \bar\mtheta_i \} - \alpha \big) \mlambda^\top \mW_{it} \right) 
		\end{align*}
		is uniformly bounded away from zero by assumption, \citet[Theorem 5.20]{W_2001} together with the Cramér-Wold device yields that
		\begin{align}
			\sqrt{T} \, \mathbf{V}_{iT}^{-1/2} \widehat{D}_{iT} (\bar\mtheta_{i}) \toD \mathcal{N}(0, I_k).
		\end{align}
		
		Finally, item~(v) of \citet[Theorem~7.1]{NM_1994} is implied by stochastic equicontinuity of the process $\nu_{iT}$ shown in Lemma~\ref{lem:GPL_Equicontinuous}; see the discussion on page 2187 of \citet{NM_1994}.
		
		Hence, we can conclude that
		\begin{align*}
			\sqrt{T}  \Big( \big(\nabla_{\mtheta \mtheta} \overline{Q}_{iT}(\bar\mtheta_i) \big)^{-1} \mathbf{V}_{iT} \big(\nabla_{\mtheta \mtheta} \overline{Q}_{iT}(\bar\mtheta_i) \big)^{-1} \Big)^{-1/2} \left( \widehat{\mtheta}_{iT} - \bar \mtheta_i \right) \toD \mathcal{N}(0, I_k),
		\end{align*}
		which particularly implies $\sqrt{T} \big( \widehat{\mtheta}_{iT} - \bar\mtheta_i \big) = \mathcal{O}_\P(1)$.
		Invoking $\mW_{it} = 1$, this also implies $\sqrt{T} \big( \widehat{r}_T - \bar r \big) = \mathcal{O}_\P(1)$ such that condition (i) of Assumption~\ref{ass:Normality} holds.
		
		The high-level Assumption \ref{ass:Normality}\ref{item:CLT} holds by the same arguments following \eqref{eqn:Phiscorecrameruniv}, applied to the GPL loss function and using the respective moment conditions from Assumption~\ref{ass:QuantilesLowLevel}.
	\end{proof}

	\begin{proof}[\bf Proof of Proposition \ref{prop:sim}]
		Consider the squared error scoring function and notice that
		\begin{align} \label{eq:eq:appsim_mcb}
			\begin{split}
				\MCBp_{i} &= \E \Bigr[ \myS(X_{it},Y_t) -  \myS(\E[Y_t \mid \mW_{it}],Y_t)  \Bigr]\\
				&=  \E \Bigr[ (X_{it} - Y_t)^2 -  (\E[Y_t \mid \mW_{it}] - Y_t)^2  \Bigr] \\
				%&=  \E\Bigr[X_{it}^2 - \E[Y_{t}\mid \mW_{it}]^2 - 2X_{it}Y_t + 2Y_t\E[Y_t\mid  \mW_{it} ]\Bigr] \\
				&=  \E[X_{it}^2] - \E[\E[Y_t\mid \mW_{it}]^2] - 2\E[X_{it}Y_t] + 2\E[Y_t\E[Y_{t}\mid \mW_{it}]] \\
				&=  \E[X_{it}^2] + \E[\E[Y_t\mid \mW_{it}]^2] - 2\E[X_{it}Y_t] ,  
			\end{split}
		\end{align}
		and
		\begin{align}\label{eq:eq:appsim_dsc}
			\begin{split}
				\DSCp_i &= \E \Bigr[\myS(\E[Y_t],Y_t)\Bigr] - \E \Bigr[\myS(\E[Y_t \mid \mW_{it}],Y_t)\Bigr] \\
				&= \E \Bigr[ ( \E[Y_t]-Y_t )^2  - ( \E[Y_t \mid \mW_{it}]-Y_t )^2 \Bigr]  \\
				%&= 	\E \Bigr[  \E[Y_t]^2 - 2\E[Y_t]Y_t + Y_t^2 -  \E[Y_t\mid  \mW_{it} ]^2 + 2\E[Y_t\mid \mW_{it}]Y_t - Y_t^2 \Bigr]   \\
				&= - \E\Bigr[ \E[Y_t \mid \mW_{it} ] ^2\Bigr] + 2 \E\Bigr[Y_t \E[Y_t \mid \mW_{it}]\Bigr]  \\
				&=   \E\Bigr[ \E[Y_t \mid \mW_{it} ] ^2\Bigr],  
			\end{split}
		\end{align}
		where the final equality follows from the law of iterated expectations.
		
		Recall that $ \varsigma = \var(M_t) =  \var(K_t)   = \var(L_t)  = \var(N_t)  = \dfrac{1}{1-\beta^2} $. 
		As the processes $ M_t, K_t,  L_t, N_t$ are independent, it also follows that  $\E[M_t K_t] = \E[M_t L_t] = \E[M_t N_t] = \E[K_t L_t] = \E[K_t N_t] = \E[L_t N_t] = 0$. 
		Hence, for forecaster $i=1$, we have $\E[X_{1t}] = \delta_0$ and 
		\begin{align}  
			\label{eq:appsim_X1sq}
			\E[X_{1t}^2]  &= \delta_0^2+(\delta_M^2+\delta_K^2+\delta_N^2)  \varsigma = \delta_0^2+\varsigma \cdot \bm{\delta}^\top \bm\delta, \qquad\text{and} \\
			\label{eq:appsim_X1Ycov}
			\E[X_{1t}Y_t] &= (\delta_K\gamma_K+\gamma_M\delta_M+\gamma_N\delta_N) \varsigma = \varsigma\cdot \bm\delta^\top \bm\gamma.
		\end{align}
		For forecaster $ i=2 $, we accordingly have $\E[X_{2t}] = \xi_0$ and 
		\begin{align}
			\label{eq:appsim_X2sq}
			\E[X_{2t}^2]  &= \xi_0^2+(\xi_M^2+\xi_L^2+ \xi_N^2) \varsigma =  \xi_0^2 + \varsigma\cdot\bm\xi^\top \bm\xi, \qquad\text{and} \\
			\label{eq:appsim_X2Ycov}
			\E[X_{2t}Y_t] &= (\xi_M\gamma_M+\gamma_L\xi_L+\xi_N\gamma_N) \varsigma = \varsigma\cdot \bm\xi^\top \bm\gamma.
		\end{align}
		
		For the term $\E\bigr[ \E[Y_t \mid \mW_{it} ]^2 \bigr]$, with $ \mW_{it} = (1,X_{it})^\top $, we use $ \E[Y_t \mid \mW_{it}] =  \E[Y_t \mid X_{it}]$. 
		Given the linearity of our DGP in \eqref{eqn:SimProcess} and as $\E[Y_t] = 0$, we get that  
		\begin{align*}
			\E[Y_t \mid \mW_{it}] =  \mW^\top_{it}  \bar\mtheta_i =  \bar{k}_i +  \bar\vartheta_{i} X_{it}, 
		\end{align*}
		with 
		\begin{align*}
			\bar\vartheta_{i} =  \cov(X_{it},Y_t) / \var(X_{it}) \quad \text{and} \quad  \bar{k}_i =  - \bar\vartheta_{i} \cdot \E[X_{it}].
		\end{align*}
		
		Using the expressions \eqref{eq:appsim_X1sq}--\eqref{eq:appsim_X2Ycov}, we get
		\begin{align} 
			\label{eq:appsim_YgivenX1sq}
			\E \big[ \E[Y_t \mid \mW_{1t}]^2 \big] &= 
			\E \left[ \left(\frac{\bm\delta^\top \bm\gamma}{\bm{\delta}^\top \bm{\delta}} \right)^2 (X_{1t} - \delta_0)^2 \right]
			= \varsigma \cdot \frac{(\bm\delta^\top \bm\gamma)^2}{\bm{\delta}^\top \bm{\delta}}, \\
			\label{eq:appsim_YgivenX2sq}
			\E \big[ \E[Y_t \mid  \mW_{2t}]^2 \big] &=      \E \left[ \left(\frac{\bm\xi^\top \bm\gamma}{\bm{\xi}^\top \bm{\xi}} \right)^2 (X_{2t} - \xi_0)^2 \right]
			= \varsigma \cdot \frac{(\bm\xi^\top \bm\gamma)^2}{\bm{\xi}^\top \bm{\xi}}.
		\end{align} 
		Hence, the expression for $\MCBp_1$ from \eqref{eqn:MCBDSCSimFormula} follows by substituting \eqref{eq:appsim_X1sq}, \eqref{eq:appsim_X1Ycov} and \eqref{eq:appsim_YgivenX1sq} into \eqref{eq:eq:appsim_mcb} and by using  the identity $a + b^2/a - 2b = (a-b)^2/a$ with $a=\bm{\delta}^\top \bm{\delta}$ and $b=\bm{\delta}^\top \bm{\gamma}$.
		The formula for forecaster $X_{2t}$ follows analogously.
		
		For the uncertainty component under squared error loss  $ \UNCp = \E \big[ (\E[Y_t] - Y_t)^2 \big] = \E[Y^2_t] $ since $ \E[Y_t] = 0 $. 
		Since $ \varepsilon_{Y,t} \sim \mathcal{N}(0,1) $ and $ \var(\varepsilon_{Y,t}) = \E[\varepsilon_{Y,t}^2] = 1 $, the uncertainty component is given by
		\begin{align*}
			\UNCp &= \E[Y^2_t] \\
			&= \gamma_K^2 \E[K_t^2] + \gamma_L^2 \E[L_t^2]  + \gamma_M^2 \E[M_t^2]  + \gamma_N^2 \E[N_t^2] + \E[\varepsilon_{Y,t}^2] \\
			&= (  \gamma_K^2 + \gamma_L^2 + \gamma_M^2  + \gamma_N^2) \varsigma + 1 \\
			&= \varsigma \cdot \bm{\gamma}^\top \bm{\gamma} + 1,
		\end{align*}
		which concludes the proof.
	\end{proof}

	\section{Relation to other decompositions}
	\label{sec:connecttolit}
	
	Several early contributions to the forecast evaluation literature propose score decompositions that are structurally similar to the one employed in this paper. 
	
	First, the population decomposition of the squared error score in \citet{M_1973} coincides with the population score decomposition in \eqref{eqn:PopScoreDecomp}.
	For the squared error score, \eqref{eqn:PopScoreDecomp} simplifies to 
	$\Sp_{i} = \E\!\left[ \myS(X_{it}, X^{\textsf{C}}_{it})\right] - \E\!\left[ \myS(X^{\textsf{C}}_{it}, \bar{r})\right] + \UNCp$.
	\citet{M_1973} labels these components calibration, resolution, and uncertainty.
	This representation is prevalent in the literature and it is often applied to the \citet{B_1950} score, which describes the squared error in the context of probability  forecasts for binary outcomes; see e.g., \citet{ET_2016} and \citet{Pohle2020}.
	As such, our paper augments the original \citet{M_1973} decomposition with inference under linear recalibration.

	Second, the squared error score decomposition of \citet[Equation 5a]{MZ_1969} reads as
	\begin{align}
		\label{eqn:originalMZdec}
		\text{Average Score} = \text{Mean Component} + \text{Slope Component} + \text{Residual Component}. 
	\end{align}
	In terms of our notation, and the refinement of unconditional and conditional miscalibration terms of \citet{GR_2023}, the components in \eqref{eqn:originalMZdec}  are  $\Sp_{i} \;=\; \uMCBp_{i} \;+\; \cMCBp_{i} \;+\; \left(\UNCp\;-\;\DSCp_i\right)$,
	where $\uMCBp_{i}$ and $\cMCBp_{i}$ are denoted as unconditional and conditional miscalibration components, which are non-negative and satisfy $\MCBp_{i} = \uMCBp_{i} + \cMCBp_{i}$ given that they are obtained from
	\begin{align*}
		\uMCBp_{i} = \Sp_i \;-\; \E[\myS(X_{it}+c,Y_t)] \qquad\text{and}\qquad \cMCBp_{i} = \E\left[\myS(X_{it}+c,Y_t)\right] \;-\; \Sp_i^{\textsf{C}},
	\end{align*}
	with constant $c$ chosen such that $X_{it}+c$ is unconditionally calibrated, i.e., $\E[X_{it}+c] = \E[Y_t]$.
	
	In this sense, the score decomposition in \citet{MZ_1969} closely mirrors the structure considered here. 
	Although the residual component in  \eqref{eqn:originalMZdec} conflates uncertainty and discrimination, this is unproblematic for comparative evaluation, as the uncertainty term is forecast-invariant and cancels out in pairwise comparisons.
	Accordingly, the residual component can be viewed as a re-normalized version of discrimination. 
	\citet{MZ_1969} however place little emphasis on this residual component, referring to it merely as a residual term without further interpretation or diagnostic use. 
	
	In summary, the aforementioned decompositions can be viewed as special cases of the general score decomposition in \eqref{eqn:PopScoreDecomp}. The asymptotic theory developed in this paper is sufficiently general to be readily applicable to these alternative representations.
	For example, in the case of the \citet{MZ_1969} decomposition, inference for the combined term $\left(\UNCp-\DSCp_i\right)$ follows directly from our results, and inference for $\uMCBp_{i}$ and $\cMCBp_{i}$ only requires one additional unconditional estimation step for the constant $c$.

	\section{Details on the simulation setup for quantile forecasts}
	\label{app:QuantileSimDGP}
	
	Here, we provide detailed calculations for the population $\MCBp_i$ and $\DSCp_i$ components on the check loss decomposition for quantile forecasts for the simulations from Section~\ref{sec:sim}.
	
	\begin{prop}
		\label{prop:sim_q}
		Consider the DGP in \eqref{eqn:SimProcess} and \eqref{eqn:SimFcastq}. For $\alpha\in(0,1)$, define $z_\alpha=\Phi^{-1}(\alpha)$.
		Using $\mathcal W_{it}=\sigma\{\mW_{it}\}$ with $\mW_{it}=(1,X^\alpha_{it})^\top$ and the check loss $\myS_\alpha(x,y)=(\mathds{1}\{x \ge y\}-\alpha)(x-y)$, we get
		\begin{align}
			\label{eq:MCBq_closed}
			\MCBp_1
			&= m_1\Bigl(\Phi(\kappa_1)+\alpha-1\Bigr)+s_1\,\phi(\kappa_1)-\sigma_{1\mid X}\,\phi(z_\alpha), \\
			\label{eq:DSCq_sim}
			\DSCp_1
			&=\phi(z_\alpha)\bigl(\sigma_Y-\sigma_{1\mid X}\bigr), 
		\end{align}
		and $\UNCp = \sigma_Y\,\phi(z_\alpha)$.
		In the above,
		$\sigma_Y^2 = \var(Y_t) = \varsigma\,\bm\gamma^\top\bm\gamma+1$ with $\varsigma=\frac{1}{1-\beta^2}$, and
		\begin{align*}
			\sigma_{1\mid X}^2 &:= \var(Y_t\mid X^\alpha_{1t})
			= \sigma_Y^2-\varsigma\frac{(\bm\delta^\top\bm\gamma)^2}{\bm\delta^\top\bm\delta} \\
			m_1 &:= \E[Y_t- X^\alpha_{it}] = -(\delta_0+z_\alpha), \\
			s_1^2 &:= \var(Y_t - X^\alpha_{it}) = 1+\varsigma\bigl(\bm\gamma^\top\bm\gamma+\bm\delta^\top\bm\delta-2\bm\delta^\top\bm\gamma\bigr), \\
			\kappa_1 &:= \frac{m_1}{s_1}.
		\end{align*}
		It further holds that $\DSCp_1$ is increasing with $\frac{(\bm\delta^\top\bm\gamma)^2}{\bm\delta^\top\bm\delta}$ and $\MCBp_1$ is increasing in $s_1$ and in $\frac{(\bm\delta^\top\bm\gamma)^2}{\bm\delta^\top\bm\delta}$.
		For the second forecaster $X_{2t}^\alpha$, equivalent statements hold with $(\delta_0,\bm\delta)$ replaced by $(\xi_0,\bm\xi)$.
	\end{prop}
	
	For the four parameterizations given in Table~\ref{tab:SimQuantiles}, we continue to simplify the closed-form expressions for $\MCBp_i$ and $\DSCp_i$ from Proposition~\ref{prop:sim_q} as functions of $k$, $\alpha$, and $\varsigma$ such that the $\MCBp$ and $\DSCp$ relations from Table~\ref{tab:SimQuantiles} can be verified.
	Notice that in the following, the expressions for $\kappa(k)$ vary for the different DGPs denoted 1), 2), 4) and 5).
	
	First, for the check loss decomposition for DGP 1) in Table~\ref{tab:SimQuantiles}, we have
	\begin{align}
		\begin{aligned}
			\label{eqn:MCBDSC_QuantDGP1}
			\MCBp_1(k) &= \MCBp_2(k) = m \Bigl(\Phi\!\Bigl(\tfrac{m}{s(k)}\Bigr) + \alpha - 1\Bigr)
			+ s(k)\, \phi\!\Bigl(\tfrac{m}{s(k)}\Bigr)
			- \sqrt{1 + \tfrac{\varsigma}{16}} \, \phi(z_\alpha), \\
			\DSCp_1 &= \DSCp_2
			= \phi(z_\alpha) \left( \sqrt{1 + \tfrac{3 \varsigma}{16}} - \sqrt{1 + \tfrac{\varsigma}{16}} \right), \\
			m &= -\sqrt{1+ \frac{\varsigma}{16}} z_\alpha, \\
			s(k) &= \sqrt{\var(Y_t - X^\alpha_{it})} = \sqrt{\,1 + \varsigma \left( \tfrac{k^2}{2} + \tfrac{1}{16} \right)\,}.
		\end{aligned}
	\end{align}
	Here, $s(k)$ is increasing in $k$, such that $\MCBp_1(k) = \MCBp_2(k)$ are increasing in $k$.
	$\MCBp_1(0) = \MCBp_2(0)$ can be obtained by simplifying the above formula for $k=0$.
	Furthermore, $\DSCp_1 = \DSCp_2 > 0$ does not depend on $k$.

	Second, for the check loss decomposition for DGP 2) in Table~\ref{tab:SimQuantiles}, we have
	\begin{align}
		\begin{aligned}
			\label{eqn:MCBDSC_QuantDGP2}
			\MCBp_1 &= 0, \\[0.5ex]
			\MCBp_2(k)
			&= m_2 \Bigl(\Phi\!\Bigl(\tfrac{m_2}{s_2(k)}\Bigr) + \alpha - 1\Bigr)
			+ s_2(k)\, \phi\!\Bigl(\tfrac{m_2}{s_2(k)}\Bigr)
			- \sqrt{1 + \tfrac{\varsigma}{16}} \, \phi(z_\alpha), \\[1ex]
			\DSCp_1 &= \DSCp_2
			= \phi(z_\alpha) \left( \sqrt{1 + \tfrac{3 \varsigma}{16}} - \sqrt{1 + \tfrac{\varsigma}{16}} \right), \\[1ex]
			m_2 &= -\sqrt{1+ \frac{\varsigma}{16}} z_\alpha, \\
			s_2(k) &= \sqrt{\var(Y_t - X^\alpha_{2t})} = \sqrt{\,1 + \tfrac{\varsigma}{16} + \tfrac{\varsigma k^2}{2}\,}.
		\end{aligned}
	\end{align}
	As before, $s_2(k)$ is increasing in $k$, such that $\MCBp_2(k)$ is increasing in $k$, and $\DSCp_1 = \DSCp_2 > 0$ are independent of $k$.
	$\MCBp_2(0)$ can again be obtained by simplifying the above formula for $k=0$.

	Third, for the check loss decomposition for DGP 4) in Table~\ref{tab:SimQuantiles}, we have
	\begin{align}
		\begin{aligned}
			\label{eqn:MCBDSC_QuantDGP3}
			\MCBp_1(k) = \MCBp_2(k)
			&= m_\alpha(k) \Bigl(\Phi\!\Bigl(\tfrac{m_\alpha(k)}{s(k)}\Bigr) + \alpha - 1\Bigr)
			+ s(k)\, \phi\!\Bigl(\tfrac{m_\alpha(k)}{s(k)}\Bigr)
			- \sqrt{\tfrac{\varsigma k^2}{2} + 1} \, \phi(z_\alpha), \\
			\DSCp_1(k) = \DSCp_2(k)
			&= \phi(z_\alpha) \left( \sqrt{\varsigma k^2 + 1}
			- \sqrt{\tfrac{\varsigma k^2}{2} + 1} \right), \\
			m_\alpha(k) &= -\sqrt{1+ \frac{8}{15}k^2} z_\alpha, \\
			s(k) &= \sqrt{\var(Y_t - X^\alpha_{it})} = \sqrt{\,1 + \tfrac{\varsigma k^2}{2} + \tfrac{\varsigma}{8}\,}.
		\end{aligned}
	\end{align}
	Here, $\DSCp_1(k) = \DSCp_2(k)$ are increasing in $k$ (as $k^2$ grows faster than $k^2/2$) and we have $\DSCp_1(0) = \DSCp_2(0) = 0$.
	Furthermore, $\MCBp_1(k) = \MCBp_2(k) > 0$ holds by \eqref{eqn:MCBDSC_QuantDGP3}.

	Fourth, for the check loss decomposition for DGP 5) in Table~\ref{tab:SimQuantiles} with $\xi_0 = \xi_0(k,\alpha)$, we have
	\begin{align}
		\begin{aligned}
			\label{eqn:MCBDSC_QuantDGP5}
			\DSCp_1 &= 0, \\
			\DSCp_2(k) &= \phi(z_\alpha) \left( \sqrt{\varsigma k^2 + 1}
			- \sqrt{\tfrac{\varsigma k^2}{2} + 1} \right), \\
			\MCBp_1(k) &= - (\delta_0 + z_\alpha) \Bigl(\Phi\!\Bigl(\tfrac{- (\delta_0 + z_\alpha)}{s_1(k)}\Bigr) + \alpha - 1\Bigr)
			+ s_1(k)\, \phi\!\Bigl(\tfrac{- (\delta_0 + z_\alpha)}{s_1(k)}\Bigr)
			- \sqrt{\varsigma k^2 + 1} \, \phi(z_\alpha), \\
			\MCBp_2(k) &= - (\xi_0 + z_\alpha) \Bigl(\Phi\!\Bigl(\tfrac{- (\xi_0 + z_\alpha)}{s_2(k)}\Bigr) + \alpha - 1\Bigr) + s_2(k)\, \phi\!\Bigl(\tfrac{- (\xi_0 + z_\alpha)}{s_2(k)}\Bigr)
			- \sqrt{\tfrac{\varsigma k^2}{2} + 1} \, \phi(z_\alpha), \\
			s_1(k) &:= \sqrt{\var(Y_t - X^\alpha_{1t})} = \sqrt{\,1 + \varsigma k^2 + \tfrac{\varsigma}{8}\,}, \\
			s_2(k) &:= \sqrt{\var(Y_t - X^\alpha_{2t})} = \sqrt{\,1 + \tfrac{\varsigma k^2}{2} + \tfrac{\varsigma}{8}\,}.
		\end{aligned}
	\end{align}
	As before, $\DSCp_2(k)$ is increasing in $k$ (as $k^2$ grows faster than $k^2/2$) and we have $\DSCp_2(0) = 0$.
	Notice for this that the choice of $\xi_0$ does not affect the discrimination of the second forecast at all.
	
	\begin{figure}[tbp]
		\centering
		\includegraphics[width=\textwidth]{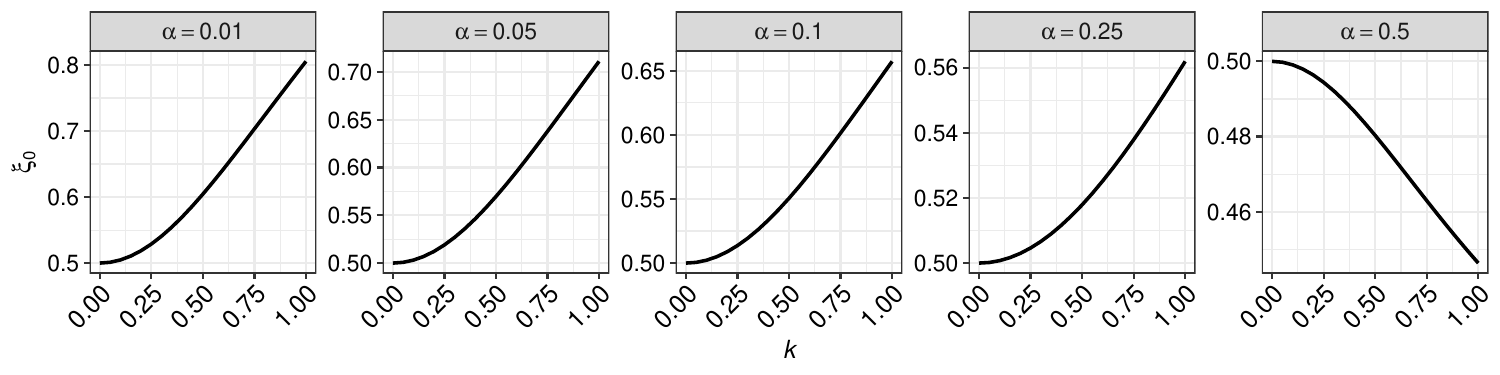}
		\caption{Values for $\xi_0(k,\alpha)$ used for the quantile DGP 5) in Table~\ref{tab:SimQuantiles} for $k \in [0,1]$ and $\alpha \in \{0.01, 0.05, 0.1, 0.25, 0.5 \}$, such that $\MCBp_1(k) = \MCBp_2(k)$ holds as accurately as possible for all $k$ and $\alpha$ for the quantities from \eqref{eqn:MCBDSC_QuantDGP5}.
		}
		\label{fig:xi0_choice}
	\end{figure}
	
	Equality of the $\MCBp_i$, $i=1,2$ terms in \eqref{eqn:MCBDSC_QuantDGP5} is difficult to obtain analytically. 
	Hence, we choose $\xi_0 = \xi_0(k,\alpha)$ depending on $\alpha \in \{0.01,\, 0.05,\, 0.1,\, 0.25,\, 0.5 \}$ and a sequence of values of $k \in [0,1]$ 
	numerically such that $\MCBp_1(k,\alpha) = \MCBp_2(k,\alpha)$ holds as accurately as possible.
	Figure~\ref{fig:xi0_choice} displays these numeric solutions for  $\xi_0 = \xi_0(k,\alpha)$ as functions of $k$ and $\alpha$.

	\begin{figure}[tbp]
		\centering
		\includegraphics[width=\textwidth]{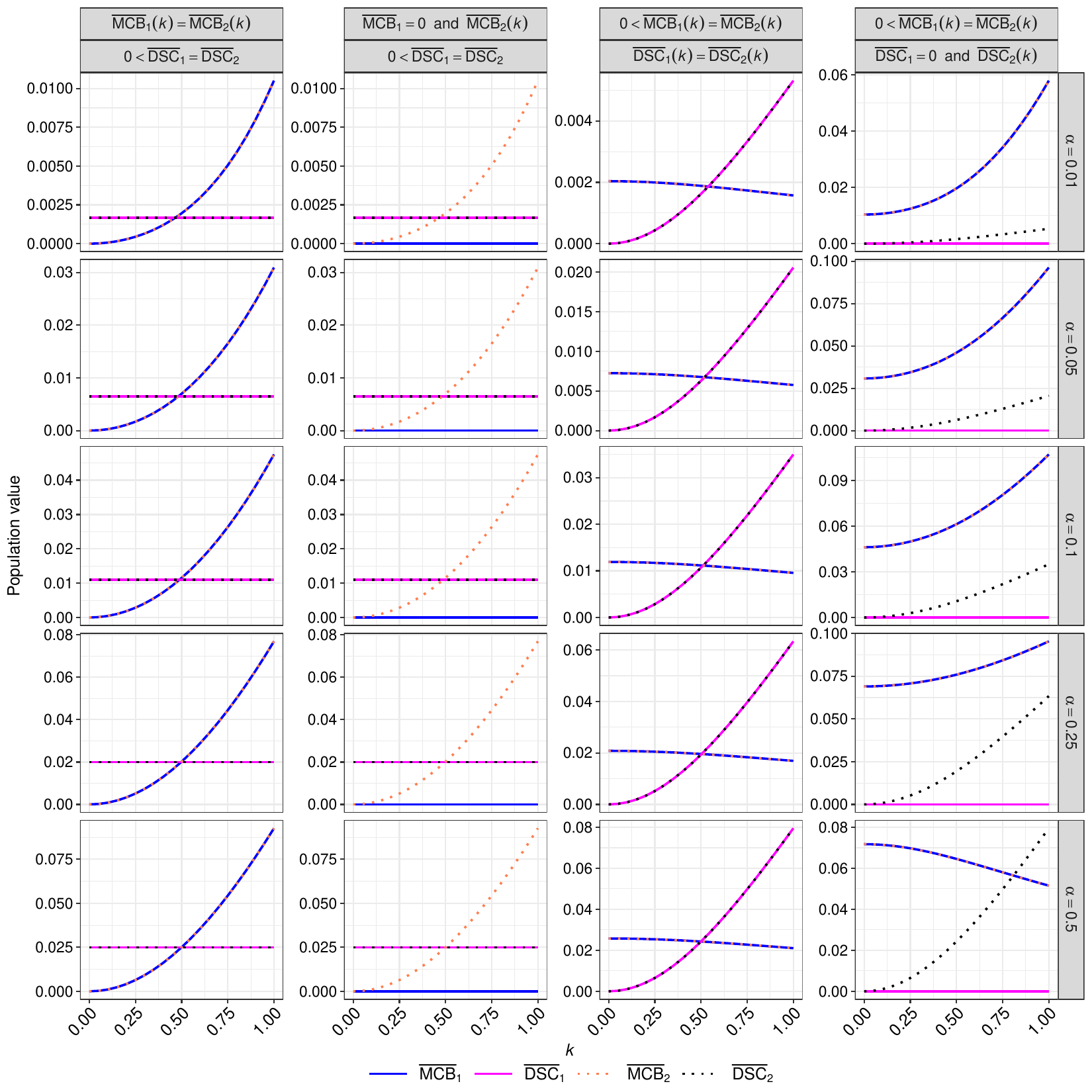}
		\caption{Population values of $\MCBp_i$ and $\DSCp_i$, $i=1,2$ from \eqref{eqn:MCBDSC_QuantDGP1}--\eqref{eqn:MCBDSC_QuantDGP5} for a range of values for $k$, our five choices of $\alpha$, and the parameter values $\xi_0(k,\alpha)$ displayed in Figure \ref{fig:xi0_choice}.}
		\label{fig:TrueMCBDSC_QuantDPG}
	\end{figure}

	Figure~\ref{fig:TrueMCBDSC_QuantDPG} displays the expressions of $\MCBp_i$ and  $\DSCp_i$, $i=1,2$ from \eqref{eqn:MCBDSC_QuantDGP1}--\eqref{eqn:MCBDSC_QuantDGP5} for the considered values of $k$ and $\alpha$, which confirms the properties of our simulation setup discussed above.
	In particular, Figure~\ref{fig:TrueMCBDSC_QuantDPG} confirms that the numeric choices of $\xi_0(k,\alpha)$ yield the desired results.

	\begin{proof}[Proof of Proposition \ref{prop:sim_q}]
		We prove the claims for $i=1$; the case $i=2$ follows by replacing $(\delta_0,\bm\delta)$ with $(\xi_0,\bm\xi)$.
		
		First, under \eqref{eqn:SimProcess}, the vector $\mathbf V_t$ has mean zero and covariance matrix $\varsigma I_4$, hence
		$\bm\gamma^\top\mathbf V_t\sim \mathcal N(0,\varsigma\,\bm\gamma^\top\bm\gamma)$ and
		\[
		Y_t=\bm\gamma^\top\mathbf V_t+\varepsilon_{Y,t}\sim \mathcal N(0,\sigma_Y^2),
		\qquad 
		\sigma_Y^2=\varsigma\,\bm\gamma^\top\bm\gamma+1.
		\]
		Moreover, with
		$X_{1t}^\alpha =  \mathbf{V}_t^\top \bm{\delta} + \delta_0  + z_\alpha$
		we have
		\begin{align*}
			\var(X_{1t}^\alpha)=\varsigma\,\bm\delta^\top\bm\delta,
			\qquad \text{ and } \qquad
			\cov(Y_t,X_{1t}^\alpha)=\varsigma\,\bm\delta^\top\bm\gamma.
		\end{align*}
		Since $(Y_t,X_{1t}^\alpha)$ is jointly Gaussian, $Y_t\mid X_{1t}^\alpha$ is Gaussian with constant conditional variance
		\begin{align}
			\label{eqn:CondVarianceSimQuant}
			\sigma_{1\mid X}^2
			=\var(Y_t)-\frac{\cov(Y_t,X_{1t}^\alpha)^2}{\var(X_{1t}^\alpha)}
			=\sigma_Y^2-\varsigma\frac{(\bm\delta^\top\bm\gamma)^2}{\bm\delta^\top\bm\delta}.
		\end{align}
		
		We continue by deriving the formula for $\UNCp$.
		Since the $\alpha$-quantile functional is strictly consistent for the check loss, the expected loss $\E[\myS_\alpha(x,Y_t)]$ is uniquely minimized at the true conditional quantile $Q_\alpha(Y_t) = \sigma_Y z_\alpha = \sigma_Y \Phi^{-1}(\alpha)$.

		To obtain an explicit formula for the check loss at the true quantile, let $Y_t = \sigma_Y Z_t$ with $Z_t \sim \mathcal N(0,1)$.
		Then,
		\begin{align*}
			\E[\myS_\alpha(\sigma_Y z_\alpha, Y_t)]
			&= \E\big[(\alpha - \mathds{1}\{Z_t \le z_\alpha\})(\sigma_Y Z_t - \sigma_Y z_\alpha)\big] \\
			&= \sigma_Y \E\big[(\alpha - \mathds{1}\{Z_t \le z_\alpha\})(Z_t - z_\alpha)\big] \\
			&= \sigma_Y \left( \alpha \E[Z_t - z_\alpha] - \E\big[(Z_t - z_\alpha)\mathds{1}\{Z_t \le z_\alpha\} \big] \right).
		\end{align*}
		Since $\E[Z_t]=0$, the first term becomes $-\alpha z_\alpha$.
		For the second term,
		\[
		\E\big[(Z_t - z_\alpha) \mathds{1}\{Z_t \le z_\alpha\}\big]
		= \int_{-\infty}^{z_\alpha} (z - z_\alpha)\phi(z)\,dz 
		= -\phi(z_\alpha) - z_\alpha \alpha,
		\]
		which yields
		\begin{align}
			\label{eqn:CheckLossGaussian}
			\E[\myS_\alpha(\sigma_Y z_\alpha,Y_t)]
			= \sigma_Y \big[-\alpha z_\alpha + \phi(z_\alpha) + z_\alpha \alpha \big]
			= \sigma_Y \phi(z_\alpha).
		\end{align}

		We continue to derive the close-form expression for the discrimination component.
		Because $(Y_t,X_{1t}^\alpha)$ is jointly Gaussian, the conditional distribution
		$Y_t\mid X_{1t}^\alpha$ is normal with conditional variance
		$\sigma_{1\mid X}^2=\var(Y_t\mid X_{1t}^\alpha)$ from \eqref{eqn:CondVarianceSimQuant} that does not depend on $X_{1t}^\alpha$.
		Hence,
		\[
		Y_t\mid X_{1t}^\alpha \sim \mathcal{N} \big( \E[Y_t\mid X_{1t}^\alpha], \; \sigma_{1\mid X}^2 \big)
		\qquad\text{and}\qquad
		Q_\alpha(Y_t\mid X_{1t}^\alpha)=\E[Y_t\mid X_{1t}^\alpha]+\sigma_{1\mid X} z_\alpha.
		\]
		By the check loss identity for Gaussian distribution from \eqref{eqn:CheckLossGaussian}, we obtain
		\[
		\E\!\bigl[\myS_\alpha \big(Q_\alpha(Y_t\mid X_{1t}^\alpha),Y_t\big) \mid X_{1t}^\alpha\bigr]
		=   \sigma_{1\mid X} \phi(z_\alpha) \quad \text{a.s.},
		\]
		and therefore also
		\[
		\E\!\bigl[\myS_\alpha\bigl(Q_\alpha(Y_t\mid X_{1t}^\alpha),Y_t\bigr)\bigr]
		=\sigma_{1\mid X}\phi(z_\alpha).
		\]
		Consequently,
		\[
		\DSCp_1 =\UNCp - \E\!\bigl[\myS_\alpha \bigl(Q_\alpha(Y_t\mid X_{1t}^\alpha),Y_t\bigr)\bigr]
		=\sigma_Y\phi(z_\alpha)-\sigma_{1\mid X}\phi(z_\alpha),
		\]
		which is \eqref{eq:DSCq_sim}, 
		and
		\[
		\MCBp_1=\E[\myS_\alpha(X_{1t}^\alpha,Y_t)]-\E[\myS_\alpha(Q_\alpha(Y_t\mid X_{1t}^\alpha),Y_t)]
		=\E[\myS_\alpha(X_{1t}^\alpha,Y_t)]-\sigma_{1\mid X}\phi(z_\alpha),
		\]
		for which we derive a closed form for $\E[\myS_\alpha(X_{1t}^\alpha,Y_t)]$ in the following.
		
		For this, define the forecast error $E_{1t}:=Y_t- X_{1t}^\alpha$.
		Then $E_{1t}$ is Gaussian with mean
		\[
		m_1=\E[E_{1t}] = \E[Y_t] - \E[ X_{1t}^\alpha] =-(\delta_0+z_\alpha),
		\]
		and variance
		\begin{align*}
			s_1^2
			&=\var(Y_t-X_{1t}^\alpha)
			=\var(Y_t)+\var(X_{1t}^\alpha)-2\cov(Y_t,X_{1t}^\alpha) \\
			&=1+\varsigma\bigl(\bm\gamma^\top\bm\gamma
			+\bm\delta^\top\bm\delta
			-2\bm\delta^\top\bm\gamma\bigr),
		\end{align*}
		
		To compute $\E[\myS_\alpha(X_{1t}^\alpha,Y_t)]$, note that for $U=x+E$,
		\[
		\myS_\alpha(x,U)
		=(\mathds{1}\{U<x\}-\alpha)(x-U)
		=(\alpha-\mathds{1}\{E<0\})E.
		\]
		Hence,
		\[
		\E[\myS_\alpha(x,U)]
		=\alpha\,\E[E]-\E\!\bigl[E\,\mathds{1}\{E<0\}\bigr].
		\]
		For $E\sim\mathcal N(m,s^2)$, a direct calculation yields
		\[
		\E\!\bigl[E\,\mathds{1}\{E<0\}\bigr]
		= m\,\Phi(-m/s)-s\,\phi(m/s),
		\]
		and therefore
		\[
		\E[\myS_\alpha(x,U)]
		= m\bigl(\Phi(m/s)+\alpha-1\bigr)+s\,\phi(m/s).
		\]
		Applying this with $x=X^\alpha_{1t}$ and $U=Y_t$ and using $\kappa_1=m_1/s_1$ gives
		\[
		\E[\myS_\alpha(X^\alpha_{1t},Y_t)]
		= m_1\bigl(\Phi(\kappa_1)+\alpha-1\bigr)+s_1\,\phi(\kappa_1).
		\]
		
		We continue to establish monotonicity of $\MCBp_1$ and $\DSCp_1$ with respect to $\kappa_1$ and $R_1 = \frac{(\bm\delta^\top\bm\gamma)^2}{\bm\delta^\top\bm\delta}$.
		For $\DSCp_1$, it is immediate that $\sigma_{1|X}$ is decreasing in $R_1$, such that $\DSCp_1$ is increasing in $R_1$.
		The same argument applies to $\MCBp_1$ being increasing in $R_1$ as only the latter part in \eqref{eq:MCBq_closed} depends on $\sigma_{1|X}$.
		
		From \eqref{eq:MCBq_closed}, only the first two terms of $\MCBp_1$ depend on $s_1$. 
		Define
		\[
		g(s_1)
		:= m_1\bigl(\Phi(\kappa_1)+\alpha-1\bigr)
		+ s_1\,\phi(\kappa_1),
		\qquad
		\kappa_1=\frac{m_1}{s_1}.
		\]
		Differentiating with respect to $s_1$ and using $\frac{d\kappa_1}{ds_1}=-\kappa_1/s_1$,  we obtain
		\begin{align*}
			\frac{dg(s_1)}{ds_1}
			&= m_1 \phi(\kappa_1)\frac{d\kappa_1}{ds_1}
			+ \phi(\kappa_1)
			+ s_1 \phi'(\kappa_1)\frac{d\kappa_1}{ds_1} \\
			&= -\frac{m_1\kappa_1}{s_1}\phi(\kappa_1)
			+ \phi(\kappa_1)
			+ \kappa_1^2 \phi(\kappa_1) \\
			&= \phi(\kappa_1) > 0,
		\end{align*}
		where we used $\phi'(x)=-x\phi(x)$ and $\kappa_1=m_1/s_1$.
		Since the remaining term $-\sigma_{1\mid X}\phi(z_\alpha)$ does not depend on $s_1$, $\MCBp_1$ is strictly increasing in $s_1$.
	\end{proof}

	\section{Additional Figures}
	\label{sec:AddFigs}

	\begin{figure}[tbp]
		\centering
		\includegraphics[width=\textwidth]{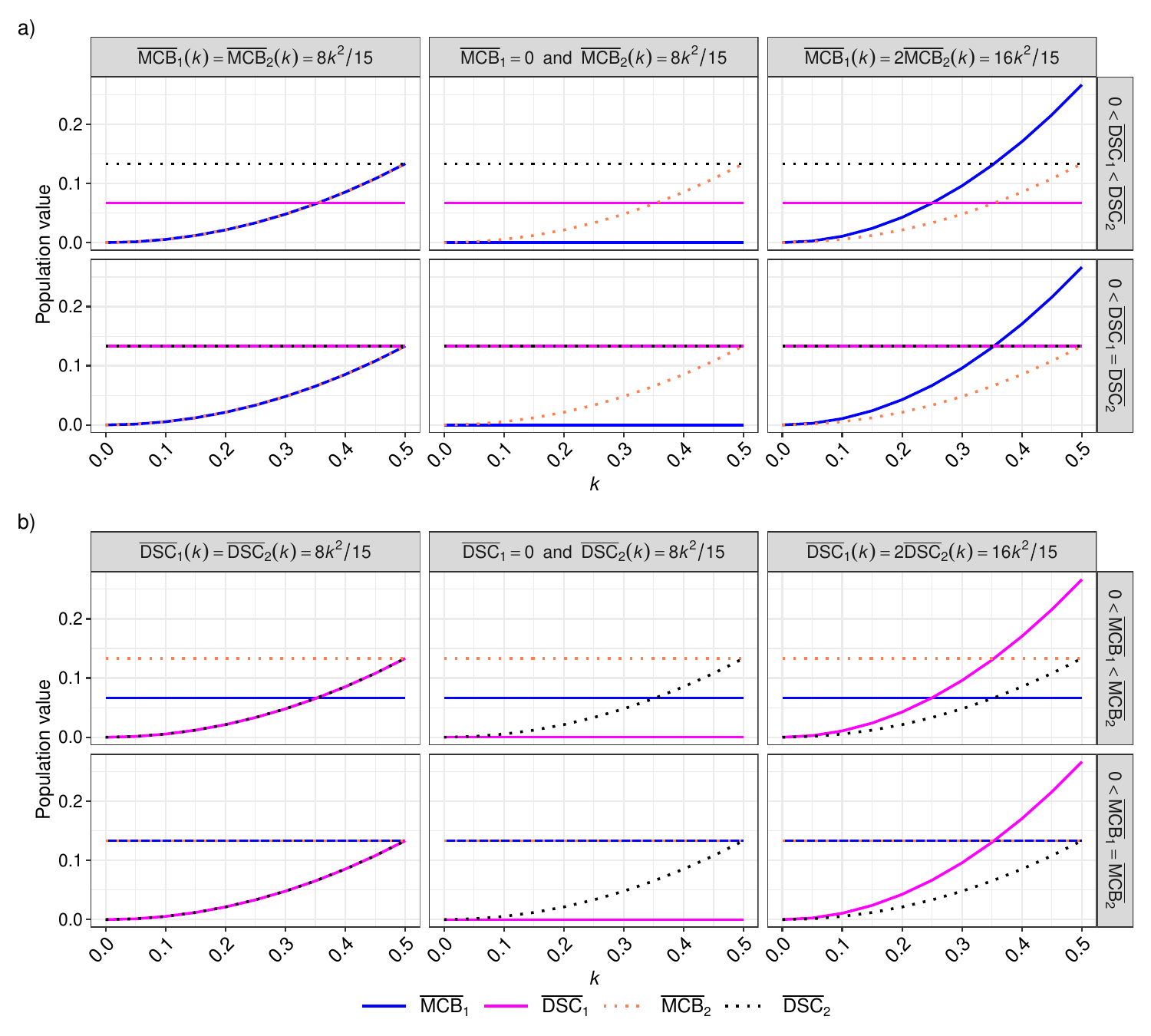}
		\caption{Population values of $\MCBp_i$ and $\DSCp_i$, $i=1,2$ for the mean forecasts from the DGP in \eqref{eqn:SimProcess}--\eqref{eqn:SimFcast2} for the twelve parametrizations from Table~\ref{tab:main}, following the close-form expressions of Proposition~\ref{prop:sim}.}
		\label{fig:TrueMCBDSC_MeanDPG}
	\end{figure}
	
	Figure~\ref{fig:TrueMCBDSC_MeanDPG} plots the population values for $\MCBp_i$ and $\DSCp_i$, $i=1,2$ for the mean forecasts for the twelve DGPs from Table~\ref{tab:main}.

	\begin{figure}[tbp]
		\centering
		\includegraphics[width=\textwidth]{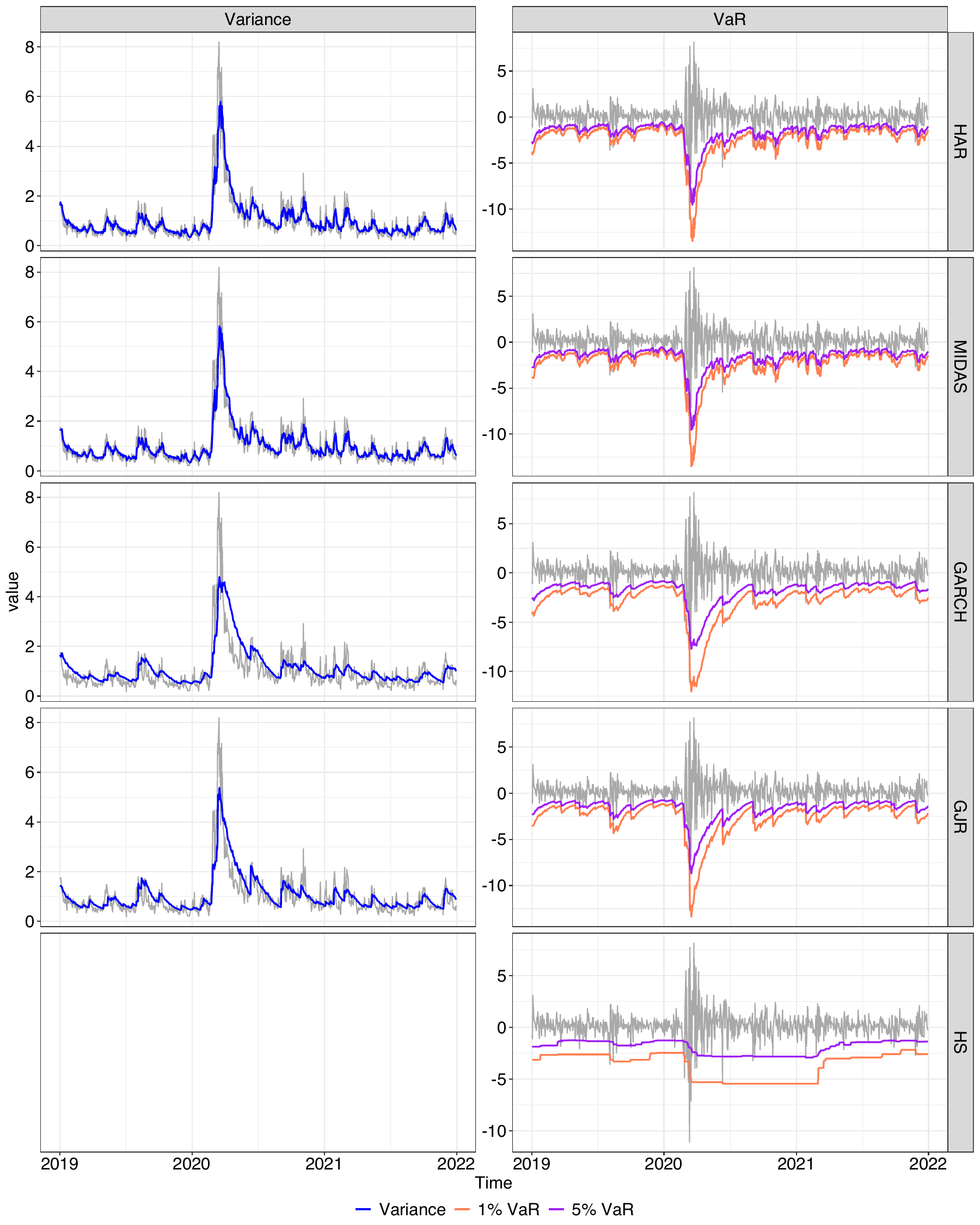}
		\caption{Daily E-mini returns (right) and Realized Variances (left) in gray, together with the corresponding variance forecasts in blue (left) and VaR forecast at levels $1\%$ in orange and $5\%$ in purple (right) for the forecasting models considered in Section~\ref{sec:appl_vola} and the time period after the year 2019.}
		\label{fig:RiskApplicationTimeSeriesPlot}
	\end{figure}
	
	Figure~\ref{fig:RiskApplicationTimeSeriesPlot} plots the time series of the variance and VaR forecasts together with their corresponding evaluation targets for the application of Section~\ref{sec:appl_vola}.
	
	\begin{figure}[tbp]
		\centering
		\includegraphics[width=\textwidth]{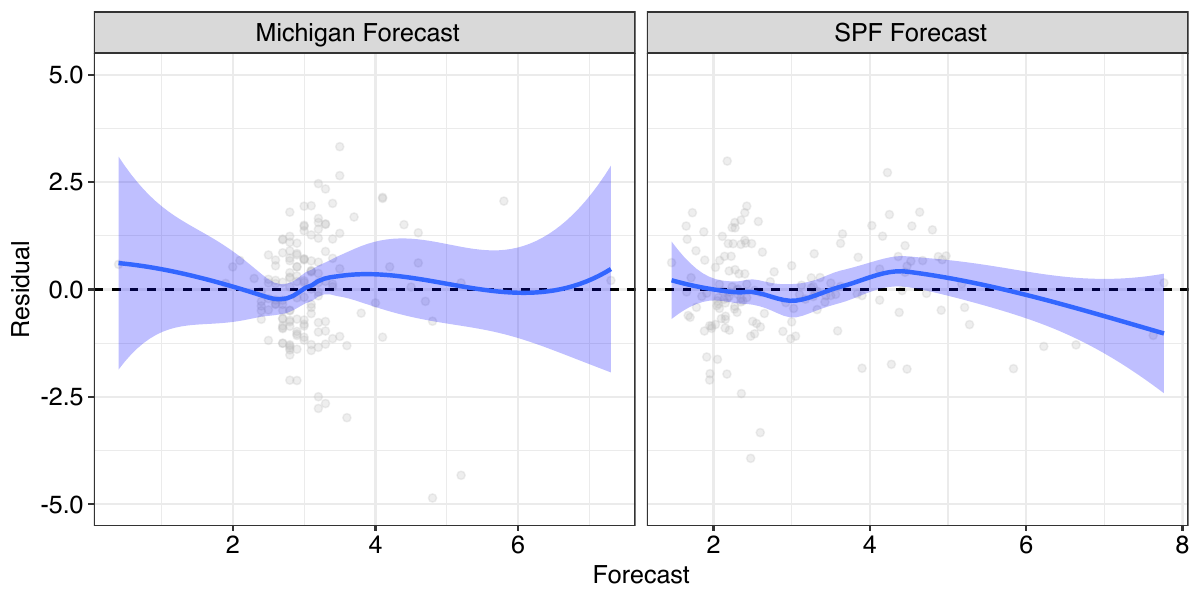}
		\caption{Model diagnostics for the linear recalibration regression for the inflation rate forecasts in Section~\ref{sec:appl_infl}. We plot the model residuals against the forecasts (as the covariates of the underlying recalibration regression) together with a nonparametric polynomial regression and its $95\%$ confidence bands. 
			The zero-line being mostly contained in the confidence bands implies no evidence against correct model specification. For further details, see the text of Appendix~\ref{sec:AddFigs}.}
		\label{fig:MZdiagnosticsInflation}
	\end{figure}

	\begin{figure}[tbp]
		\centering
		\includegraphics[width=\textwidth]{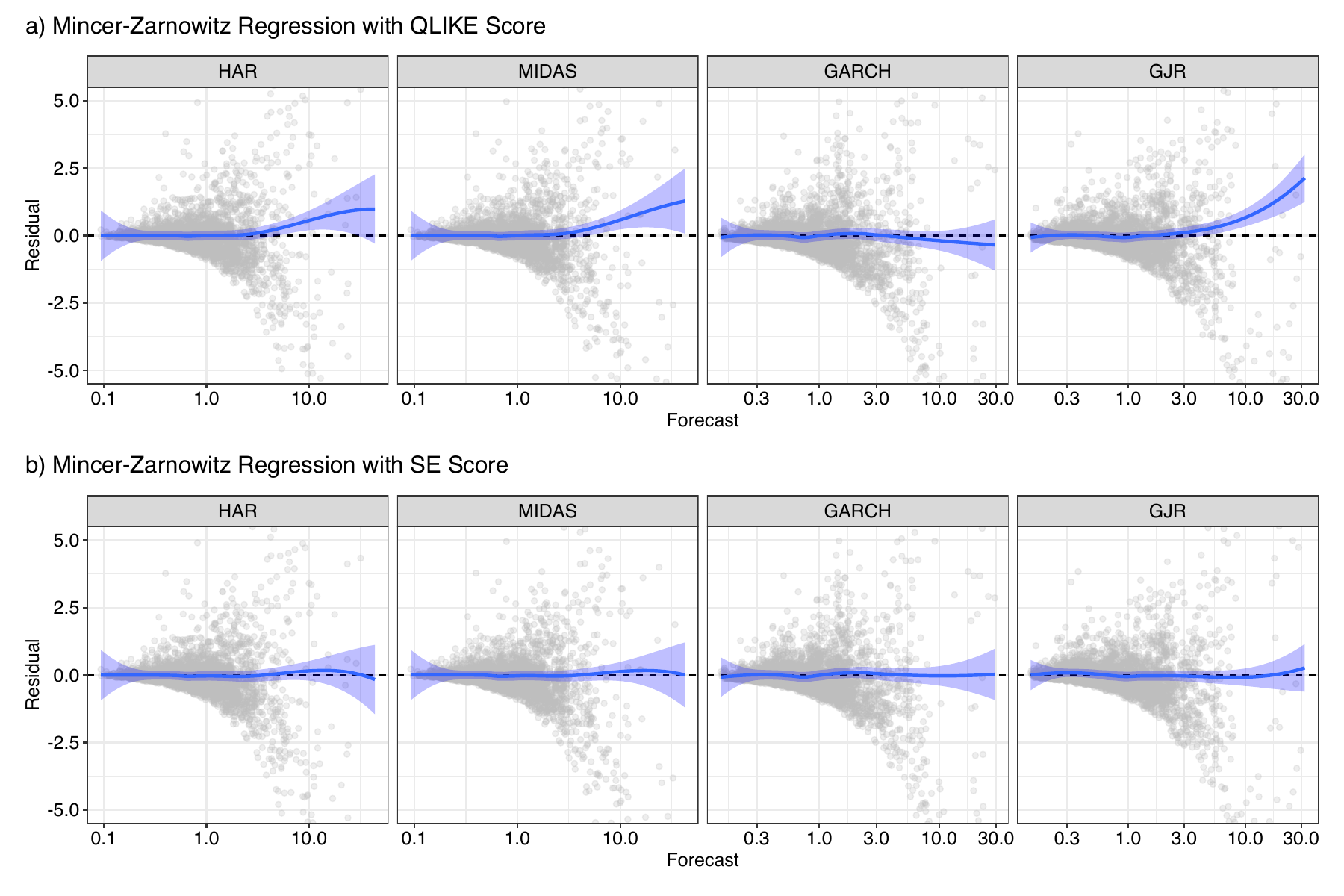}
		\caption{Model diagnostics for the linear recalibration regression for the variance forecasts in Section~\ref{sec:appl_vola}. We plot the model residuals against the forecasts (as the covariates of the underlying recalibration regression) on a log-transformed $x$-axis together with a nonparametric polynomial regression and its $95\%$ confidence bands. 
			The zero-line being mostly contained in the confidence bands implies no evidence against correct model specification. For further details, see the text of Appendix~\ref{sec:AddFigs}.}
		\label{fig:MZdiagnosticsVola}
	\end{figure}

	\begin{figure}[tbp]
		\centering
		\includegraphics[width=\textwidth]{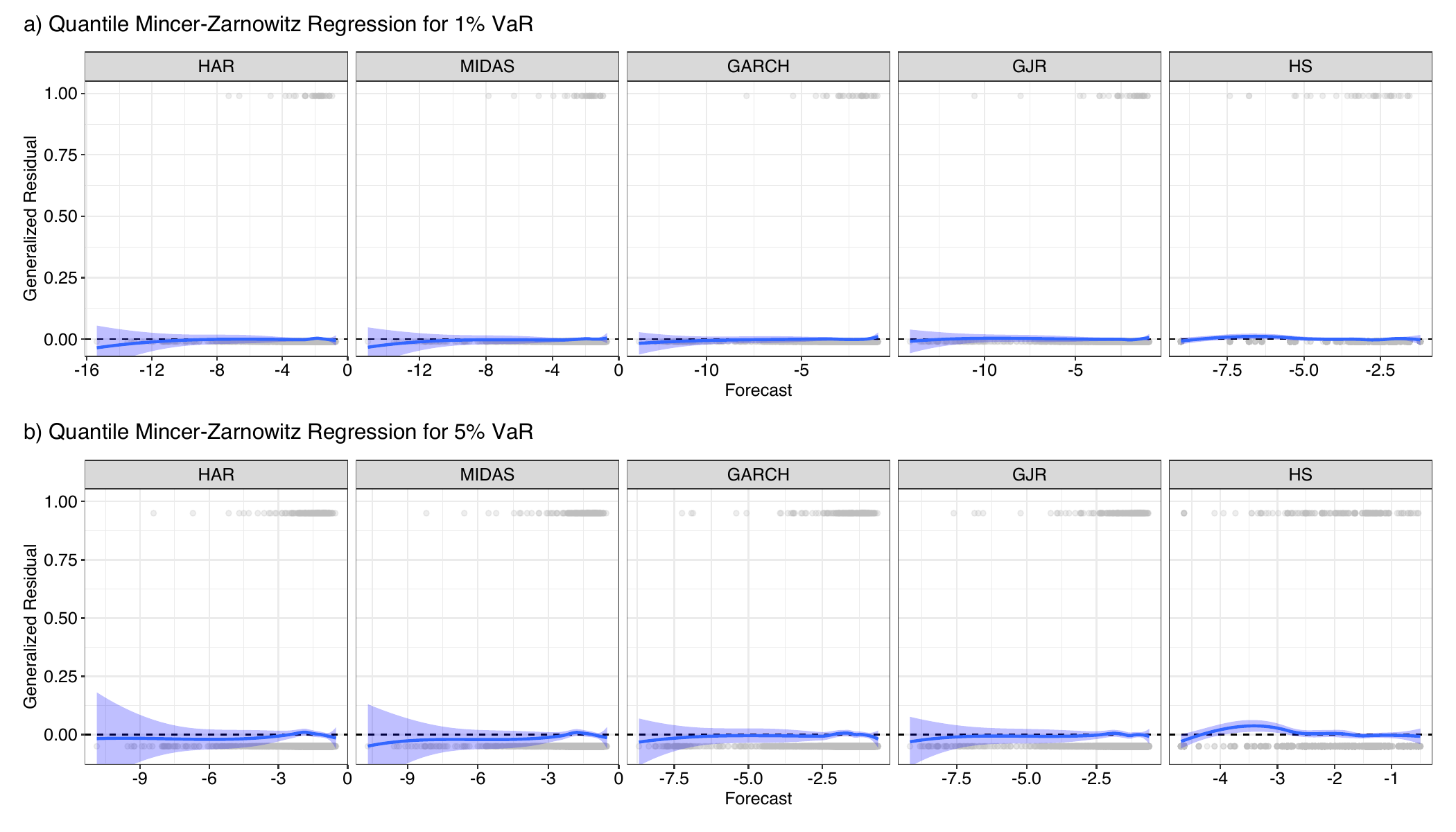}
		\caption{Model diagnostics for the linear recalibration regression for the VaR forecasts in Section~\ref{sec:appl_vola}. We plot the generalized model residuals against the forecasts (as the covariates of the underlying recalibration regression) together with a nonparametric polynomial regression and its $95\%$ confidence bands. 
			The zero-line being mostly contained in the confidence bands implies no evidence against correct model specification. For further details, see the text of Appendix~\ref{sec:AddFigs}.}
		\label{fig:MZdiagnosticsVaR}
	\end{figure}

	Figures~\ref{fig:MZdiagnosticsInflation}--\ref{fig:MZdiagnosticsVaR} show model diagnostics for the linear recalibration regressions that are used for the score decompositions for the forecasts of the applications of Sections~\ref{sec:appl_infl}--\ref{sec:appl_vola}.
	For the mean forecasts in Figures~\ref{fig:MZdiagnosticsInflation}--\ref{fig:MZdiagnosticsVola}, we plot the model residuals $Y_t - \mW_{it}^\top \widehat{\mtheta}_{iT}$ on the $y$-axis against the forecasts $X_{it}$ (as the covariates of the underlying recalibration regressions) on the $x$-axis.
	The blue curve and its pointwise $95\%$ confidence band are obtained from a local polynomial mean regression, implemented via the \texttt{loess} function in the statistical software \texttt{R}, with automatically selected tuning parameters.
	The black-dashed zero-line being inside these confidence bands for the majority of forecast values indicates no evidence against correct specification of the linear recalibration regressions.
	
	In Figure~\ref{fig:MZdiagnosticsVaR}, we follow \citet{Pohle_GenResid} and \citet[Figures~4 and 6]{dimitriadis2025regressions} and plot the generalized quantile residuals, i.e., the identification functions, $\myV \big( \mW_{it}^\top \widehat{\mtheta}_{iT}, Y_t) = \mathds{1}\{Y_t \le \mW_{it}^\top \widehat{\mtheta}_{iT}\} - \alpha$, against the quantile forecasts $X_{it}$ together with a corresponding nonparametric mean regression and their pointwise 95\% confidence bands.
	The zero-line being mostly contained in the confidence bands again does not provide evidence against correct specification of the linear recalibration regressions.

	%=============================================================
\end{document}